\documentclass{article}
\usepackage{graphicx} % Required for inserting images
\usepackage{fullpage}
\usepackage{amsmath}
\usepackage{amsfonts}
\usepackage{amssymb}
\usepackage{physics}
\usepackage{xcolor}
\usepackage{amsthm}
\usepackage{hyperref}
\usepackage{comment}
\usepackage{authblk}     % For multiple affiliations and emails
\usepackage{subcaption}
\usepackage{threeparttable}
\usepackage{booktabs} % for \toprule, \midrule, \bottomrule

\definecolor{daniel}{rgb}{.8,.5,.3}
\definecolor{ruben}{rgb}{.3,.5,.8}

\newcommand{\cS}{\mathcal{S}}

\newcommand{\cH}{\mathcal{H}}

\newcommand{\cD}{\mathcal{D}}
\newcommand{\cC}{\mathcal{C}}
\newcommand{\cX}{\mathcal{X}}
\newcommand{\cE}{\mathcal{E}}
\newcommand{\cU}{\mathcal{U}}
\newcommand{\cK}{\mathcal{K}}

\newcommand{\nuHS}{\nu_{\textrm{HS}}}
\newcommand{\nudr}{\nu_{d,r}}

\newcommand{\bE}{\mathbb{E}}
\newcommand{\bEx}[1]{\underset{#1}{\mathbb{E}}}

\newcommand{\esssup}{\mathop{\mathrm{ess\,sup}}}

\newtheorem{theorem}{Theorem}
\newtheorem{lemma}[theorem]{Lemma}
\newtheorem{proposition}[theorem]{Proposition}
\newtheorem{definition}[theorem]{Definition}
\newtheorem{corollary}[theorem]{Corollary}
\newtheorem{example}[theorem]{Example}
\newtheorem{remark}[theorem]{Remark}
\newtheorem{fact}[theorem]{Fact}

\title{Average Contraction Coefficients of Quantum Channels}
\author[1,2]{Ruben Ibarrondo\thanks{ruben.ibarrondo@ehu.eus}}
\author[3,4]{Daniel Stilck Fran\c{c}a\thanks{dsfranca@math.ku.dk}}

\affil[1]{Department of Physical Chemistry, University of the Basque Country UPV/EHU, Apartado 644, 48080 Bilbao, Spain}
\affil[2]{EHU Quantum Center, University of the Basque Country UPV/EHU, Apartado 644, 48080 Bilbao, Spain}
\affil[3]{Department of Mathematical Sciences, University of Copenhagen, Universitetsparken 5, 2100 Copenhagen, Denmark}
\affil[4]{Inria, ENS Lyon, UCBL, LIP, F-69342 Lyon Cedex 07, France}

\date{\today}

\begin{document}

\maketitle
\begin{abstract}
The data-processing inequality ensures quantum channels reduce state distinguishability, with contraction coefficients quantifying optimal bounds. However, these can be overly optimistic and not representative of the usual behavior. We study how noise contracts distinguishability of “typical” states, beyond the worst–case. To that end, we introduce and study a family of moments of contraction for quantum divergences, which interpolate between the worst-case contraction coefficient of a channel and its average behavior under a chosen ensemble of input states. We establish general properties of these moments, relate moments for different divergences, and derive bounds in terms of channel parameters like the entropy or purity of its Choi state.

Focusing on the trace distance, we obtain upper and lower bounds on its average contraction under tensor-product noise channels, and prove that—depending on the local noise strength—there is a phase transition in the limit of many channel uses: below a critical error rate the average contraction remains near unity, whereas above it decays exponentially with system size.  We extend these phase-transition phenomena to random quantum circuits with unital noise, showing that constant-depth noisy circuits do not shrink the trace distance on average, even when given highly entangled states as input. In contrast, even at $\log \log(n)$ depth, the average trace distance can become superpolynomially small.

Finally, we explore moments of contraction for $f$-divergences and discuss applications to local differential privacy, demonstrating that noise regimes ensuring privacy can render outputs essentially indistinguishable on average.  Thus, our results provide a fine-grained framework to quantify typical-case channel noise in quantum information and computation and unveil new phenomena in contraction coefficients, such as phase transitions for average contraction.

\end{abstract}

\section{Introduction}
Since the foundation of information theory~\cite{Shannon1948}, the study of random processes has been at its core. One of the fundamental achievements of (quantum) information theory is the data processing inequality~\cite{Shannon1948,Lieb1973,Umegaki1962,Lindblad1975}, which formalizes the intuition that two quantities cannot become more distinguishable when we apply noise to them. More formally, for quantum states $\rho$ and $\sigma$ we can quantify their distinguishability in terms of a divergence $D(\rho\|\sigma)$, such as the trace distance or relative entropy. Then the data processing inequality $D$ states that for any quantum channel $T$ we have that:
\begin{equation}\label{eq:dpi}
    D(T(\rho)\| T(\sigma))\leq D(\rho\| \sigma),
\end{equation}
i.e. the distinguishability, as measured by $D$, always decreases if we apply $T$ to the states. 
However, for many noisy processes we expect that the distinguishability of two states will actually decrease strictly, that is, we have a strict inequality in Eq.~\eqref{eq:dpi}. One way to make this statement quantitative is to consider contraction coefficients of the channel and the divergence~\cite{Hiai15,Hirche2022,Berta2023,Hirche2024}. The contraction coefficient is the smallest possible $\eta\in[0,1]$ such that 
\begin{align}\label{equ:dpicc}
D(T(\rho)\| T(\sigma))\leq \eta D(\rho\| \sigma)
\end{align}
for any pair of states $\rho$ and $\sigma$. Another variation of contraction coefficients that have received considerable attention in the literature~\cite{Hiai15,MllerHermes2016,Temme10} are those where we restrict the state $\sigma$ to be a fixed-point of the channel, i.e. $T(\sigma)=\sigma$, as these provide a natural way of quantifying to what extent all states are converging to it. Contraction coefficients have been widely studied and they have found applications in quantum information \cite{Hirche23, Hiai15, Temme10}, quantum machine learning~\cite{Angrisani2023DPamplificationquantum,Angrisani2023QDP}, the analysis of runtimes of quantum algorithms and mixing times~\cite{Temme10,gibbs_samplers}, and limitations of quantum computation~\cite{StilckFrana2021,aharonov_depo}, to name a few. Furthermore, several works have computed or bounded the value of $\eta$ for various divergences and channels and related the contraction between different choices of $D$~\cite{Hiai15,Hirche23,MllerHermes2016,MllerHermes2018}.

However, one should note that Eq.~\eqref{equ:dpicc} needs to hold for \emph{any} pair of states. Worst–case bounds can be overly optimistic for high–dimensional or random inputs. For instance, it is known~\cite{MllerHermes2016} that if we let $D$ be the relative entropy and we consider the global depolarizing channel with parameter $p$ on $n$ qudits, it will have approximately the same contraction coefficient as the tensor power of $n$ local depolarizing channels. But for the global channel, we expect $p$ errors per application of the channel, whereas for the local one we expect $np$ and it should be ``noisier" for ``typical" states. This high-level intuition has been confirmed in some cases. For instance, it has been shown that in random circuits with noise, after a certain depth the contraction for typical circuits is exponentially worse than in the best-case scenario~\cite{Quek2024, Gonzalez-Garcia2022} and the use of average properties of the noise was crucial to obtain convergence rates that depend on the number of qubits, not only on the depth of the noisy circuit.

The main motivation of our work is to investigate the difference between worst-case contraction and average case for quantum channels, while also providing a framework to study this question. We generalize the study of contraction coefficients by introducing the notion of moments of contraction for divergences. These are quantities that interpolate between the standard worst-case notion of contraction coefficients and their average behavior w.r.t. a distribution of input states, formalizing the notion of how noisy a channel is on a ``typical" state. 

Besides introducing the notion of moment of contraction or average contraction in Sec.~\ref{sec:moments_contra_defi}, we derive various basic properties of these quantities and start to investigate how they behave for various choices of divergences, quantum channels and distributions over input states. 

In Sec.~\ref{sec:av_contract_tr_dist} we investigate closely the case of the trace distance, as it is arguably one of the most fundamental distinguishability measures between two states. Furthermore, although we state most results for arbitrary channels, we often specialize to tensor products of quantum channels, as these are relevant toy models for noise both in quantum Shannon theory and quantum computing. We then derive various upper and lower bounds for the average contraction of the trace distance under such channels that depend on properties of the Choi matrix and the image of the maximally mixed state. Some of these bounds confirm the intuition that, in some cases, the gap between the average and worst-case contraction can be exponentially large. 

However, interestingly, we uncover the existence of various phase transitions for the average contraction coefficient in terms of the (local) noise strength in Sec.~\ref{sec:lower_bounds_trace}. We show that, depending on the local noise strength $p$, for $p<p_1$ the average contraction goes to $1$ (no distinguishability loss), and for $p>p_0$ it converges exponentially fast to $0$ (indistinguishable). Two canonical examples are presented in Fig.~\ref{fig:phase_transitions}. The main technical tools required to show these results are various concentration inequalities for random quantum states~\cite{Puchala16}, combined with tools from the theory of unitary or state designs to compute expectation values. Furthermore, we often discuss to what extent our results depend on the particular choice of distribution over random inputs and the operational motivation between them.

In addition to results on tensor product channels, we extend such phase transition results to the more general case of random quantum circuits under unital noise in Sec.~\ref{sec:quantum_circuits}. We strengthen and extend the results of~\cite{Deshpande2021TightBO} on lower bounds for distinguishability and show that constant depth noisy quantum circuits do not contract the trace distance asymptotically, i.e. the average contraction converges to $1$. In contrast, in~\cite{Deshpande2021TightBO} the authors showed a constant lower bound that decayed exponentially with the depth. In addition, their bound required the initial state to be product, whereas our bound only requires that the initial states come from a state $1$-design. Considering that also highly entangled states can form a $1$-design (i.e. random quantum states are a $1$-design), our finding shows that constant depth, noisy quantum circuits cannot affect the distinguishability of even highly entangled states. Our results neatly complement those of e.g.~\cite{Quek2024, Gonzalez-Garcia2022}, which show that at slightly superconstant depths, for instance $\textrm{poly}\log \log(n)$ in the case of~\cite{Quek2024}, there are examples of noisy random circuits that have a contraction coefficient that is superpolynomially small in $n$. These results suggest a more universal picture of phase transitions in the study of average contraction coefficients for various ensembles, as can be seen in Table~\ref{tab:your_label}. 
\begin{table}[h!]
\centering
\begin{tabular}{|c|c|c|}
\hline
Work & Bound & Setup \\
\hline
 \cite{Deshpande2021TightBO}& $\Omega(p^{cD})$ &  Random brickwork circuit, product input.\\
 \hline
 \cite{Quek2024}& $O(n^{-c\log(n)})$ & $D=\Omega(\log \log (n))$, tailored family of random circuits \\
 \hline
 Our work & $1-O(e^{-cn})$ &  Depth $D<1/S(\tau_1)$, $1$-design layers of unitaries, $1$-design of input states.   \\
\hline
\end{tabular}
\caption{Bound on the expected value of the trace distance to maximally mixed for depth $D$ circuits under local depolarizing noise with depolarizing parameter $p$ and $S(\tau_1)$ the von Neumann entropy of its Choi state. We see that our bound exponentially improves upon that of \cite{Deshpande2021TightBO} and works in a more general setting. Taken together with \cite{Quek2024}, these results also show phase transitions for such setups.}
\label{tab:your_label}
\end{table}

\begin{figure}
	\subcaptionsetup[figure]{skip=-10pt,slc=off,margin={0pt,0pt}}
		\subcaptionbox{\label{fig:pt_depol}}
		{\includegraphics[width=.45\textwidth]{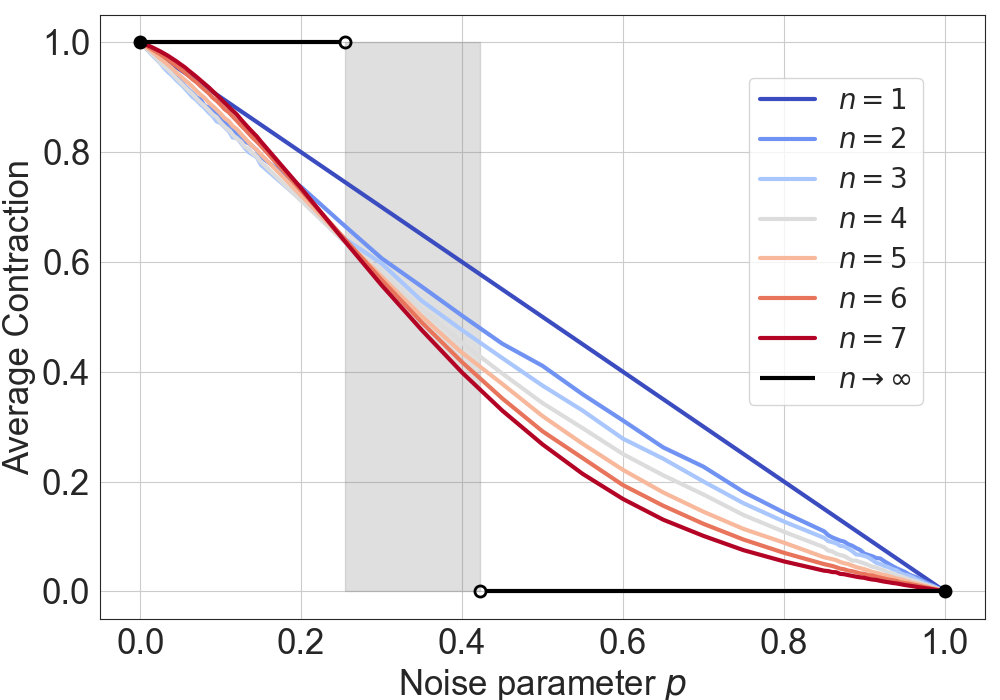}}
	\hfill
	\subcaptionsetup[figure]{skip=-10pt,slc=off,margin={0pt,0pt}}
		\subcaptionbox{\label{fig:pt_discard}}
		{\includegraphics[width=.45\textwidth]{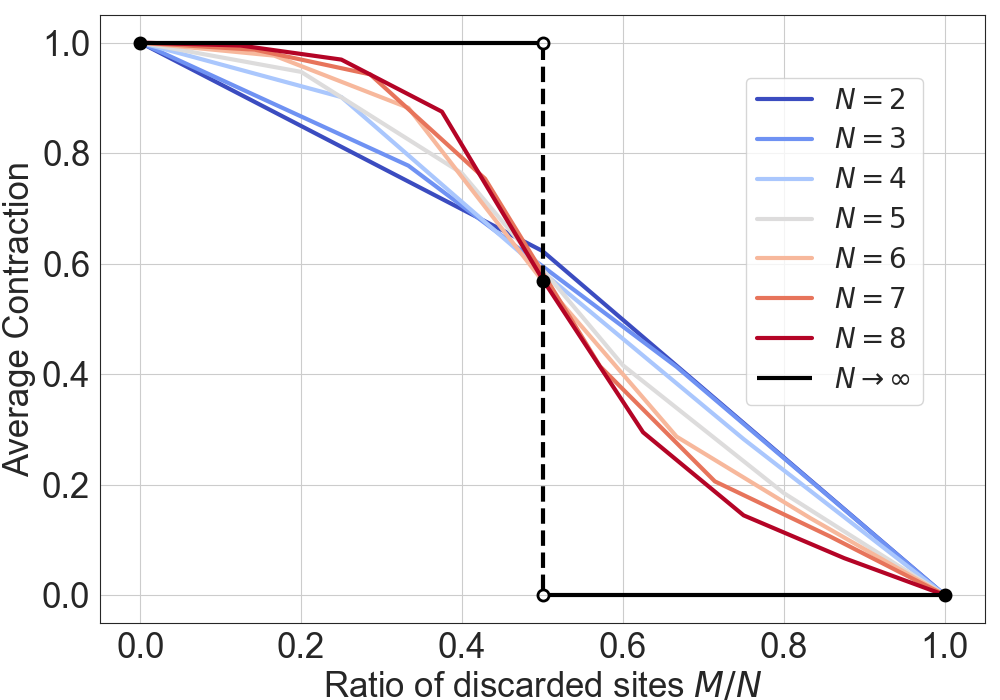}}
	\caption{
	Average contraction with state pairs with independent Haar distributions for two different families of quantum channels and different system sizes. (\subref{fig:pt_depol}) shows the $n$-fold product of a single-qubit depolarizing channel, where the horizontal axis is the depolarizing probability $p$ with $p_0\approx 0.25$ and $p_1\approx 0.42$; (\subref{fig:pt_discard}) shows the partial trace channel, starting with $N$ qubits and discarding $M$ of them, the horizontal axis is the ratio $M/N$ of discarded qubits. The thick black lines indicate the asymptotic predictions, in agreement with Prop.~\ref{prop:avc_local_depol} and~\ref{prop:avc_partial_trace_asymp_lim}. The shaded area depicts a regime where none of our results can predict the asymptotic behavior.}
	\label{fig:phase_transitions}
\end{figure}

We also provide results for divergences beyond the trace distance in Sec.~\ref{sec:bound_for_dfs}. In the case of the usual contraction coefficient, some of the most celebrated results are various relations between contraction coefficients. For instance, it is known for $f$-divergences that the trace distance has the largest contraction coefficient, and the $\chi^2$ the smallest~\cite{Hiai15,Hirche23}. Although we do not obtain the same inequality for the average case contraction, we establish various criteria when an average contraction w.r.t. one divergence, say the trace distance, implies the average contraction of another $f$ divergence. These results allow us to show that the average contraction coefficients converge to $0$ for a large class of quantum divergences.

Finally, in Sec.~\ref{sec:ldp} we use the results on contraction for more general divergences to discuss how notions of average contraction can shed light on local differential privacy~\cite{Angrisani2025}. In a nutshell, differential privacy is a framework to ensure that outputs of algorithms (say a classifier) do not convey information about their inputs and, in this way, guarantee the privacy of the users. The idea of differential privacy is that inputs that are ``similar" should lead to indistinguishable outputs, so that the algorithm cannot be used to infer information about the input. Such privacy guarantees are naturally formulated in terms of divergences~\cite{Hirche2023QuantumDP} and various recent works have extended notions of differential privacy to the quantum setting~\cite{Hirche2023QuantumDP,Angrisani2023DPamplificationquantum,Nuradha25,hirche_privateht,Nuradha2024}. 
A typical mechanism to ensure differential privacy is to add noise to the output of an algorithm~\cite{Du2021,Angrisani2025}. However, although adding more noise improves the privacy guarantees, it also reduces the overall utility. Furthermore, the guarantees are often stated by considering the worst-case over possible inputs, in a similar spirit to the contraction coefficient. Thus, in the spirit of this work, it is natural to investigate if the noise acts in a much more destructive way on average inputs and gives outputs that are of little value. We give some examples where this is the case, i.e. we give examples where the noise necessary to achieve differential privacy leads to average outputs to be essentially indistinguishable. However, we only show this for the case of the uniform distribution over mixed states. This distribution is of questionable operational significance, as it is mostly supported on highly mixed states, and we leave it as an open problem to derive similar results for more operationally justified models.  We then attempt to deliver more operationally motivated results by considering classifiers and show how differential privacy can limit their performance.

Thus, our work establishes the moment of contractions and average contraction as a useful tool to quantify how noisy a quantum channel is, going beyond the worst-case performance and allowing for a more fine-grained analysis. In addition, we believe that the phase transitions of average contraction are interesting in their own right and justify further investigation. There is still room to extend the relations between moments of contraction for different divergences, generalizing our results to more families of divergences and other distributions. Finally, it would be interesting to see if the idea of average contraction can be used to obtain better mixing times for the preparation of quantum Gibbs states~\cite{gibbs_samplers}, which bound the worst-case contraction coefficient of the trace distance.

\section{Preliminaries and notation}
\label{sec:preliminaries}
We will now introduce the basic concepts and notations we use from quantum information theory and probability theory. We refer the reader to~\cite{Heinosaari2008} for a more thorough overview.

\paragraph{Quantum information theory:} We let $\cH_d$ be the $d$-dimensional complex Hilbert space associated with a quantum system. The set of quantum states $\cS(\cH_d)$ corresponds to positive semidefinite operators of unit-trace. 
For $p\geq1$ the Schatten $p$-norm of $X$ a linear operator on $\cH_d$ is defined through its singular values $s_i(X)$:
\begin{equation}
    \norm{X}_p = \qty(\sum_{i=1}^d s_i(X)^p)^{1/p},
\end{equation}
and we let $\norm{X}_\infty$ be given by the largest singular value.
The case $p=1$ gives rise to the trace distance:
\begin{equation}\label{eq:tr_dist}
    D_{\Tr}(\rho, \sigma) = \frac{1}{2}\norm{\rho-\sigma}_1.
\end{equation}
This quantity is of fundamental importance in quantum information theory, as it is directly related to the optimal probability of distinguishing two quantum states~\cite{MWolf2012Guidedtour}.
The $2$-norm, also known as the Frobenius norm or Hilbert-Schmidt norm, admits a quadratic form $\norm{X}_2^2 = \langle X, X\rangle_{\text{HS}}=\Tr X^{\dag}X$, making it particularly suitable for analytic manipulation. These norms are decreasing in $p$, such that if $p<p'$ then $\norm{X}_{p}\geq \norm{X}_{p'}$, particularly $\norm{X}_{\infty}\leq \norm{X}_{2}\leq \norm{X}_{1}$. Furthermore, they satisfy H\"older's inequality
\begin{equation}
    \Tr\left(X^\dagger Y\right)\leq \norm{X^\dagger Y}_1 \leq \norm{X}_p \norm{Y}_q
    \quad \text{for} \quad
    \frac{1}{p}+\frac{1}{q} = 1.
\end{equation}
In particular, if $X$ is of rank $r$ and we take $Y$ to be the orthogonal projector into its support, we obtain the bounds $\norm{X}_1\leq \sqrt{r}\norm{X}_2 \leq r \norm{X}_{\infty}$. Readers interested in a more detailed treatment of matrix norms may consult \cite{Bhatia1996}.

Quantum channels are linear maps $T:\cS(\cH_d)\rightarrow\cS(\cH_{d'})$ that are completely positive and trace preserving.
Let us summarize some properties and representations of quantum channels and refer the reader to \cite{MWolf2012Guidedtour} for more details. Given a quantum channel $T:\cS(\cH_d)\rightarrow\cS(\cH_{d'})$ its \emph{Choi matrix} or \emph{Choi state} $\tau (T)\in\cS(\cH_{d'}\otimes\cH_d)$ is 
\begin{equation}
    \tau(T) =(T\otimes \text{id})(\ketbra{\Omega}) 
    \quad \text{with} \quad \ket{\Omega}=\frac{1}{\sqrt{d}}\sum_{i=1}^d \ket{i\, i}.
\end{equation}
That is, $\tau(T)$ is the state resulting from applying the channel to a part of a bipartite system that is initially maximally entangled. If the channel $T$ is clear from context we will denote it by $\tau$. By the Choi-Jamiolkowski isomorphism, $T$ can be recovered from $\tau$ by $\Tr[Y T(X)] = d \Tr[\tau\, Y \otimes X^{\top}]$ for any operators $X$ on $\cH_d$ and $Y$ on $\cH_{d'}$. This allows a suitable description of properties of the channel in terms of the Choi state. For instance, as we discuss in App.~\ref{sec:tau_and_pi}, the rank of the Choi state and its purity are indicators of the noise level in the channel. The Choi-rank $K=\rank (\tau)$ is the minimum number of terms $\qty{\ket{\psi_k}}_{k=1}^{K}$ required in the spectral decomposition
\begin{equation}
    \tau = \sum_{k=1}^K \omega_k \ketbra{\psi_k},
\end{equation}
which also gives the minimum number of operators $E_k:\cH_d\rightarrow\cH_{d'}$ with $k=1,..., K$ required in the Kraus representation of the channel $T(X)=\sum_{k=1}^K E_k X E_k^{\dag}$. The purity of the Choi state $\Tr\tau^2=\sum_{k=1}^K \omega_k^2$ is directly related to the mean of the squared singular values of the channel $\Tr\tau^2=\frac{1}{d^2}\sum_{i=1}^{d^2} s_i(T)^2$ (see Fact~\ref{fact:purity_to_singvals}).

\paragraph{Quantum divergences:} Part of this work will involve studying the contraction of general functionals that satisfy the data processing inequality under quantum channels in the sense of Eq.~\eqref{eq:dpi}. Such functionals are usually referred to as divergences. In particular, we will consider a quantum extension of the classical $f$-divergences, which form a wide family comprising many common divergences as specific instances. Among the various non-equivalent extensions to quantum states \cite{Wilde2018,Hiai2010,Hirche23}, here, we consider quantum $f$-divergences defined by the integral formulation \cite{Hirche23}. This choice is because the integral formulation has proven useful in deriving bounds between different divergences, which will be necessary for generalizing our results beyond the trace distance. The fundamental building block is the hockey-stick divergence $E_{\gamma}(\rho\|\sigma)$ defined for some real parameter $\gamma>0$:
\begin{equation}
    E_{\gamma}(\rho\|\sigma) = \Tr (\rho-\gamma\sigma)_+,
\end{equation}
where $X_{+}$ denotes the positive part of a Hermitian operator. Then, given a twice differentiable convex function $f:(0,\infty)\rightarrow \mathbb{R}$ with $f(1)=0$ the corresponding $f$-divergence is
\begin{equation}
D_f(\rho\Vert \sigma) = \int_1^{\infty} f''(\gamma)E_{\gamma}(\rho\Vert \sigma)+\gamma^{-3}f''(\gamma^{-1})E_{\gamma}(\sigma\Vert \rho) \dd \gamma,
\end{equation}
whenever the integral is finite and $D_f(\rho\Vert \sigma)=+\infty$ otherwise.  The relative entropy is a special case with $f(x)=x\log x$, and takes the form
\begin{equation}
    D_{\text{RE}}(\rho\|\sigma) = \Tr(\rho(\log\rho-\log\sigma)),
\end{equation}
when $\operatorname{supp}(\rho)\subset\operatorname{supp}(\sigma)$ and $D_{\text{RE}}(\rho\|\sigma) =\infty$ otherwise. These divergences satisfy the data processing inequality in Eq.~\eqref{eq:dpi}, motivating the definition of the \emph{contraction coefficient}
\begin{equation}\label{equ:contraction_Coeff}
   \eta(T, D) = \sup_{\rho, \sigma} \frac{D(T(\rho)\|T(\sigma))}{D(\rho\|\sigma)},
\end{equation}
where the optimization is over states satisfying $0<D(\rho\|\sigma)<\infty$. The contraction coefficient finds various applications ranging from bounding capacities~\cite{Hirche2022} to the rigorous study of the runtime of Monte Carlo algorithms~\cite{Temme10}. The dependence of $\eta(T, D)$ on the divergence and their relations has been extensively studied~\cite{Hiai15,Hirche23} and, for instance, it is known that
\begin{equation}
    \eta(T, E_{\gamma})\leq \eta(T, \norm{\cdot}_1) \quad \text{and} \quad \eta(T, D_f)\leq \eta(T, \norm{\cdot}_1).
\end{equation}
In Sec.~\ref{sec:bound_for_dfs}, we extend these inequalities to the average contraction. However, many properties of such coefficients remain widely open, such as their behavior when taking tensor products of quantum channels. As mentioned before, one of the main topics of interest in this work is to develop a more fine-grained understanding of contraction coefficients compared to the worst-case analysis of~\eqref{equ:contraction_Coeff}.

\paragraph{Probability measures on quantum states:} In our work, we will also consider various probability distributions over the sets of density matrices, and we briefly review some of them and their properties here, and refer the reader to~\cite{Aubrun2017} for more details.  The Haar measure $\mu_H$ is the uniform measure on unitaries $U(d)$ in the sense that it is invariant under unitary operations, i.e. for any measurable set of unitary matrices $\mathcal{A}$ and any unitary $U$ we have $\mu_H(U \mathcal{A})=\mu_H(\mathcal{A}U)=\mu_H(\mathcal{A})$. When a random variable $X$ has a probability measure $\nu$ we shall denote $X\sim\nu$. The Haar measure naturally induces a uniform measure on pure states $\ket{\psi}\in \cH_d$ by letting $U\sim\mu_H$ act on a fixed arbitrary state $U\ket{\phi_0}$ and we denote $\ket{\psi}\sim\mu_H$. There is no unique `natural' measure on mixed states, in the sense of it being invariant under unitary conjugation. Imposing that the measure to be invariant under unitary conjugations $\rho\rightarrow U \rho U^{\dag}$ implies that the distribution of the states can always be described with a product distribution $\rho=U\rho'U^{\dag}$ where $U\sim\mu_H$ and $\rho'$ has an independent distribution on density matrices~\cite[Prop. 10.4]{Aubrun2017}. The distribution of $\rho'$ determines the distribution of the spectrum and there are different reasonable choices for it. A common choice is the family of distributions induced by pure states $\ket{\psi}\sim\mu_H$ on bipartite systems $\cH_d\otimes\cH_r$ and taking the partial trace on the second subsystem $\rho = \Tr_r \ketbra{\psi}$. If $r\leq d$, the states have rank $r$ almost surely, whereas if $r\gg d$ they concentrate near the maximally mixed state \cite[Chapter 10]{Aubrun2017}. The special case $r=d$ corresponds to the Hilbert-Schmidt measure, equivalent to the measure induced by the Hilbert-Schmidt norm. Other choices include Dirac distributions with a unique state $\delta_{\sigma}(\mathcal{A})=1$ if $\sigma\in\mathcal{A}$ and $0$ elsewhere, discrete ensemble distributions $\delta_S(\mathcal{A})=\sum_i p_i \delta_{\sigma_i}(\mathcal{A})$ with $S=\qty{(p_i, \sigma_i)}$ and distributions induced by substituting the Haar measure by $t$-designs \cite{Ambainis2007}. We review the definition and properties of $t$-designs in App.\ref{sec:expvals_unit_invar}, where we also collect various expectation values related to random inputs to quantum channels and random states.

\section{Moments of contraction}\label{sec:moments_contra_defi}
We are now ready to define the main object of study in this work, namely, a generalization of the contraction coefficient of a quantum channel with respect to a divergence:
\begin{definition}[Moments of contraction]
Consider a quantum channel $T:\cS(\cH_d)\rightarrow\cS(\cH_{d'})$, a divergence $D$, and a probability measure on pairs of density matrices $\nu$ with essential support on $\{(\rho, \sigma)\in \cS(\cH_d)\times\cS(\cH_d)\; :\; 0 < D(\rho\Vert\sigma) <\infty \}$. For $p\geq1$, the $p$-th moment of contraction for a measure $\nu$ on pairs of states, a channel $T$, and divergence $D$ is
\begin{equation}
\eta_{p} (T, D, \nu) = \qty( \mathbb{E}_{(\rho,\sigma) \sim\nu} \left[\qty(\frac{D(T(\rho)\Vert T(\sigma))}{D(\rho\Vert \sigma)})^p\right])^{\frac{1}{p}}.\\
\end{equation}
\end{definition}

Let us note that the moments are always finite, as by the DPI the ratio $D(T(\rho)\Vert T(\sigma)) / D(\rho\Vert \sigma)$ lies in $[0,1]$ and hence $\eta_p(T, D,\nu)$ are the $p$-th root of the moments of a bounded random variable. Furthermore, we show in Prop.~\ref{prop:basic_prop} below that $\eta_{p} (T, D, \nu)$ monotonically increases with $p$ and converges to the essential supremum as $p\rightarrow\infty$. Thus, $p$ interpolates between $p=1$, corresponding to the average contraction of the divergence, and $p\rightarrow\infty$, corresponding to the standard contraction coefficient defined in Eq.~\eqref{equ:contraction_Coeff} if the measure is supported everywhere. This provides a tool to quantify the contraction induced by the channel in a refined way. More precisely, the case $p=1$ allows us to quantify how noisy the quantum channel is (measured by $D$) for average inputs coming from $\nu$, which is why we will also refer to it as the \emph{average contraction}. In contrast, as $p$ goes to infinity, it converges to how noisy the channel is in the worst case. The moments $1\leq p$ can then be used to obtain concentration inequalities for the contraction using standard tools.

We believe that, like the usual contraction coefficient, this quantity is of independent interest in quantum information theory, but its definition is also motivated by some applications, which parallel some of the applications of contraction coefficients. To name a few, contraction coefficients were extensively used to understand limitations of noisy quantum computers~\cite{StilckFrana2021}. In~\cite{Quek2024}, the authors already indirectly resorted to average case contraction to show how for certain random circuits, the average contraction was exponentially worse than the worst-case at a certain depth, imposing limitations on how far we can push techniques like error mitigation~\cite{Cai2023}. This observation invites a more systematic study of average contraction, as discussed here. Another application of contraction coefficients is to (quantum) differential privacy~\cite{Zhou2017}. In a nutshell, the idea of differential privacy is to inject noise into an algorithm to protect the privacy of the data used as input. Differential privacy guarantees can be formulated in terms of divergences~\cite{Hirche2023QuantumDP}, which rely on worst-case guarantees. It is usually the case that the more noise we inject, the more privacy we can guarantee. However, this will also reduce the reliability of the output of the algorithm. In particular, it is natural to ask how the injected noise will affect the output of average inputs. As we will argue, ideas from average contraction can be used to study this question. 

The definition of the average contraction coefficient depends on an underlying probability measure $\nu$. We will mainly focus on product measures $\nu=\nu_1\times\nu_2$ such that the two states have independent distributions with $\rho\sim\nu_1$ and $\sigma\sim\nu_2$. However, one could be interested in scenarios where $\rho$ and $\sigma$ are correlated through a common source, such as being the partial states of a bipartite pure state, and the choice should ultimately be informed by the application in mind. 

We will now collect some basic properties of the moments of contraction:
\begin{proposition}\label{prop:basic_prop}
\begin{enumerate}
    \item \textbf{Convexity:} If $D$ is jointly convex, i.e. $D(\lambda\rho_1+(1-\lambda)\rho_2\|\lambda \sigma_1+(1-\lambda)\sigma_2)\leq \lambda D(\rho_1\|\sigma_1)+(1-\lambda) D(\rho_2\|\sigma_2)$ for $0\leq \lambda\leq 1$, then
   \begin{equation}
        \eta_p(\lambda T_1 + (1-\lambda) T_2, D, \nu) \leq \lambda \eta_p(T_1, D, \nu) + (1-\lambda) \eta_p(T_2, D, \nu),
   \end{equation}
    for all quantum channels $T_1, T_2$.

\item \textbf{Monotonicity under channel composition:}
   \begin{equation}
   \eta_p(T_2 \circ T_1, D, \nu) \leq \eta_p(T_1, D, \nu)
   \end{equation}
   for all quantum channels $T_1, T_2$.

\item \textbf{Invariance under unitary operations:} 
   \begin{equation}
        \eta_p(\mathcal{U} \circ T, D, \nu) = \eta_p(T, D, \nu)
   \end{equation}
   for all unitary channels $\mathcal{U}$. Additionally, if $\nu$ is invariant under the map $(\rho,\sigma)\rightarrow (U\rho U^{\dag}, U\sigma U^{\dag})$ for any unitary $U$ then
   \begin{equation}
        \eta_p(T \circ \mathcal{U}, D, \nu) = \eta_p(T, D, \nu).
   \end{equation}
   
\item \textbf{Monotonicity in $p$:}
   \begin{equation}
   \text{If } p < q, \text{ then } \eta_{p}(T, D, \nu) \leq \eta_{q}(T, D, \nu)
   \end{equation}
   for all quantum channels $T$.

\item \textbf{Continuity in $p$:}
   \begin{equation}
   \lim_{{p \to p_0}} \eta_p(T, D, \nu) = \eta_{p_0}(T, D, \nu)
   \end{equation}

\item \textbf{Asymptotic behavior in $p$:}
   \begin{equation}\label{eq:pinfty}
   \lim_{{p \to \infty}} \eta_p(T, D, \nu) = \esssup_{\nu} \frac{D(T(\rho)\Vert T(\sigma))}{D(\rho\Vert \sigma)}
   \end{equation}
   where $\esssup_{\nu} X = \inf\qty{C\in \mathbb{R}\,:\, X\leq C \textrm{ almost surely} }$ is the essential supremum of a real random variable $X$. In particular, let $\nu=\nu_1 \times \delta_{\sigma}$ with $\nu_1$ be a distribution with support on every state and $\sigma >0$ full rank and $\nu=\nu_2$ be a distribution with support on every pair of states, then we recover the contraction coefficient at input state $\sigma$ and the contraction coefficient,
   \begin{align}
       \lim_{{p \to \infty}} \eta_p(T, D, \nu_1\times\delta_{\sigma}) 
       &= \sup_{\rho} \frac{D(T(\rho)\Vert T(\sigma))}{D(\rho\Vert \sigma)}=:\eta_{\infty}(T,D,\sigma),\\
       \lim_{{p \to \infty}} \eta_p(T, D, \nu_2) 
       &= \sup_{\rho, \sigma} \frac{D(T(\rho)\Vert T(\sigma))}{D(\rho\Vert \sigma)}=:\eta_{\infty}(T, D),
   \end{align}
    where the first supremum is taken over $\rho\in\cS(\cH_d)$ and the second over $\rho, \sigma\in\cS(\cH_d)$ with $0<D(\rho\|\sigma)<\infty$.
	
\end{enumerate}
\end{proposition}

\begin{proof}
Inequalities (1) and (2) are direct consequences of point-wise inequalities given by the joint convexity of the divergence and the data processing inequality. The invariance under unitary transformations (3) is a direct consequence of the point-wise equality for the divergence $D(\rho\|\sigma)=D(U\rho U^{\dagger}\|U \sigma U^{\dagger})$ and the invariance of $\nu$.

For the other results, we define a random variable by the ratio $R = D(T(\rho)\Vert T(\sigma)) / D(\rho\Vert \sigma)$. The probability measure in state pairs $\nu$ induces a probability measure $\nu_R$ in the unit interval. By assumption, $\nu$ is concentrated on values where $R$ is well defined, thus $R$ is bounded almost surely with $R \in\qty[0,1]$, due to the positivity of $D$ and the data processing inequality. This implies that $R$ is in the $L^p$ space $L^p([0,1],\nu_R)$ for any $p$ where the norm corresponds to the $p$-th moment $\norm{R}_p=(\int_0^1 R^p \dd \nu_R)^{1/p}=\eta_p(D,T,\nu)$. Therefore, monotonicity in $p$ (4) is a direct consequence of H\"older's inequality for $L^p$ spaces $\norm{R}_p\leq\norm{R}_q$ for $p\leq q$, continuity (4) follows from the dominated convergence theorem, and the asymptotic limits (6) follow from general results in $L^p$ spaces. The special cases for $\nu_1\times\delta_{\sigma}$ and $\nu_2$ follows from the equivalence of $\esssup$ and $\sup$ in this case. Indeed, on one side, we have the elementary upper bound $\esssup \leq \sup$. On the other, consider the supremum takes the value $\eta_{\infty}$ then for any $\varepsilon$ there is a state pair $(\rho_{\varepsilon},\sigma_{\varepsilon})$ with a contraction at least $\eta_{\infty}-\varepsilon$. But as $\nu$ has support in every state-pair, we can find a set $S$ with $\nu(S)>0$ and $(\rho_{\varepsilon},\sigma_{\varepsilon})\in S$, so the $\esssup$ of the contraction is lower bounded by $\eta_{\infty}-\varepsilon$ for any $\varepsilon >0$, taking the limit to zero shows the equivalence.
\end{proof}

\begin{remark}
It is worth commenting on a property that the average contraction coefficient will generally not have, namely, that for the usual contraction, we have the multiplicative property:
\begin{equation}
    \eta_{\infty}(T_2\circ T_1,D)\leq \eta_{\infty}(T_2,D)\eta_{\infty}(T_1,D),
\end{equation}
whereas, in general,
\begin{equation}
    \eta_{p}(T_2\circ T_1,D, \nu)\nleq \eta_{p}(T_2,D,\nu)\eta_{p}(T_1,D, \nu).
\end{equation}
This property is essential for certain applications of contraction coefficients, like deriving mixing times for Markov chains~\cite{Temme10}. We leave it to future work to understand to what extent the multiplicative property holds for the moments of contraction under suitable conditions. The complexity lies in the fact that the distribution induced on the output states is different from the input distribution, $(T_1(\rho), T_1(\sigma))\nsim \nu$. Indeed, the input and output dimensions for $T_1$ are potentially different, so it is even possible that $\eta(T_2, D, \nu)$ cannot be defined. Additionally, assuming the input and output spaces are the same, one can find the simplest counterexamples with channels that behave like projectors. For instance, consider a completely dephasing channel $T$, which sets every off-diagonal term to zero, and a probability distribution concentrated on states that mostly differ on the coherent part, i.e. the difference $T(\rho)-T(\sigma)$ is small with high probability. Then, $\eta_p(T,D,\nu) < 1$, but, because the channel is idempotent $T^2=T$, $\eta_p(T^k,D,\nu)=\eta_p(T,D,\nu) > (\eta_p(T,D,\nu))^k$ for any $k\geq 2$.
\end{remark}

\section{Average contraction of the trace distance} 
\label{sec:av_contract_tr_dist}

In this section, we specialize our analysis to the case where the divergence is the trace distance $D_{\Tr}(\rho, \sigma)=\frac{1}{2}\norm{\rho-\sigma}_1$. The trace distance arises naturally as a distinguishability measure, as it tightly bounds the optimal probability $p$ to distinguish between the preparation of the states $\rho$ and $\sigma$ with a single measurement as $p\leq \frac{1}{2} (1+\frac{1}{2}\norm{\rho - \sigma}_1)$~\cite{Helstrom1969}. This connection is at the core of various applications of the contraction of the trace distance to understand the limitations of noisy computations. Besides its natural operational motivation, another reason to first focus on the trace distance is that it arises from a norm. This makes it easier to bound the moments of contraction. In addition, as is often the case in information theory, we will be particularly interested in the case where the noise is i.i.d. (i.e. of the form $T^{\otimes n}$). 

However, even though this is arguably one of the simplest setups to study, we show that it has a rich behavior. By deriving both upper and lower bounds on the average contraction coefficient under various setups, we show that the average contraction coefficient exhibits a stark phase transition: for strong enough (but constant) local noise, the average contraction coefficient will be \emph{exponentially small} in the number of qubits, showing that most states will become essentially indistinguishable unless we are given exponentially many samples. In contrast, for small enough (but constant) local noise, we will show that the average contraction coefficient converges to $1$ as we increase the number of qubits, i.e. the states are \emph{perfectly distinguishable} on average. Such a behavior is reminiscent of phenomena like (reverse) threshold theorems. Reverse threshold theorems~\cite{Kempe2008} assert that it is possible to classically simulate noisy quantum circuits at high enough local noise. In contrast, threshold theorems~\cite{Aharonov2008} assert we can perform quantum error correction if local noise is low enough.
 
Our proofs for upper bounds consist of computing exact results for the 2-norm contraction and using tighter average-case bounds between the norms due to the concentration of measure. For the lower bounds, we will construct projections that distinguish states well on average.

\subsection{Upper bounds on the contraction coefficient for the trace distance}\label{sec:upper_bounds_tr}

One of the key motivations of this work is to identify state distributions, grounded in operational considerations, where the average contraction of a quantum channel significantly deviates from its contraction coefficient. In this section, we establish upper bounds that play a crucial role in demonstrating this phenomenon for a broad class of channels. Our primary focus is on the trace distance, a fundamental metric for quantifying the distinguishability of quantum states. We leave the extensions of these results to other divergences for Sec.~\ref{sec:bound_for_dfs}.

We begin by analyzing the case where state pairs are i.i.d. with the Hilbert-Schmidt measure on mixed states. The main result is Thm.~\ref{thm:HSaveragecontraction}, which has Cor.~\ref{cor:travcontract_nproduct} as its most illustrative consequence. The latter demonstrates that for product channels $\Phi^{\otimes n}$, the average contraction of the trace distance undergoes exponential suppression as $n\to\infty$ for any non-unitary channel.

Following this, we turn our attention to the Haar measure over pure states. Theorem~\ref{thm:avc_from_operator_root_bound} establishes an upper bound with a similar structure to the previous one, but in this case, the bound is only informative for sufficiently noisy channels. Cor.~\ref{cor:travcontract_nproduct_2designs} extends the exponential contraction result for product channels, again contingent on the mixedness of the Choi state.

We follow a similar proof strategy in both cases. Firstly, we establish that for the desired distributions, the trace distance between random states can be lower bounded with high probability. This allows us to focus on the numerator, which is easier to upper bound. The fact that a tail bound enables focusing on the numerator is a consequence of the following elementary lemma.

\begin{lemma}\label{lemma:tailbound2ratiobound}
	Let $X$ and $Y$ be two random variables such that $0\leq X\leq Y$, and define $\mu_x:=\bE[X]$. Let us assume that there is some $\mu_y>\mu_x$ and a positive function $t:(0,\mu_y)\to [0,+\infty)$, such that for any $0 < \delta < \mu_y$ the following tail bound holds
	\begin{equation}
		\Pr[Y < \mu_y - \delta] \leq t(\delta).
	\end{equation}
	Assume $\Pr[Y=0]=\xi < 1$, then, for any such $\delta$,
	\begin{equation}
		\bE\qty[\frac{X}{Y} \mid Y > 0] \leq \frac{1}{1-\xi} \qty[ \frac{\mu_x}{\mu_y} \qty(1 + \frac{\delta}{\mu_y-\delta}) + t(\delta)].
	\end{equation}
\end{lemma}

\begin{proof}
	Let $P$ denote the probability distribution for the events governing $X$ and $Y$, and define $A_*$ as the set of events where $Y=0$. By assumption we have that $P(A_*)=\Pr[Y=0]=\xi<1$. The probability distribution conditioned on $Y>0$ is obtained by
	\begin{equation}
		P|_{A_*^c}(B) = \frac{P(B \cap A_*^c)}{1-\xi},
	\end{equation}
	with this, we evaluate the expectation as
	\begin{equation}
		\bE\qty[\frac{X}{Y} \mid Y > 0] = \int_{A_*^c} \frac{X}{Y} \frac{\dd P}{1-\xi}.
	\end{equation}
	Now, consider $0<\delta<\mu_y$ and the set $A_{\delta}=\qty{Y < \mu_y-\delta}$ with its complement $A_{\delta}^c$. Noting that $A_{\delta}^c\subseteq A_*^c$, split the integral
    \begin{align}
        \int_{A_*^c} \frac{X}{Y} \dd P & = \int_{A_*^c \cap A_{\delta}} \frac{X}{Y} \dd P + \int_{A_{\delta}^c} \frac{X}{Y} \dd P \\
        &\leq \int_{A_*^c \cap  A_{\delta}} 1 \dd P + \int_{A_{\delta}^c} \frac{X}{\mu_y-\delta} \dd P \leq P(A_*^c \cap A_{\delta}) + \frac{\mu_x}{\mu_y-\delta},
    \end{align}
    where the first inequality uses $X/Y\leq 1$ and $X/Y\leq X/(\mu_y-\delta)$ in $A_{\delta}^c$, and the last comes from noting that $X$ is a positive variable $\int_{A_{\delta}^c} X \dd P \leq \int X \dd P = \mu_x$. The final result follows from taking $P(A_*^c \cap A_{\delta})\leq P(A_\delta)\leq t(\delta)$ and rearranging the terms.
\end{proof}

In the following, we resort to this lemma in situations where $\xi\approx 0$ and $t(\delta)$ is rapidly decreasing. This enables taking $\delta$ very small, implying that the expectation value of the ratio is effectively upper bounded by the ratio of the expectation values with some small additive deviation.

\subsubsection{The Hilbert-Schmidt measure}\label{sec:upperbound_hs}

Let us first consider the case where states $\rho,\sigma\in\cS(\cH_d)$ are distributed uniformly over full-rank states, i.e. we consider the distribution $\nuHS$, in a $d$-dimensional Hilbert space $\cH_d$. It is known that for sufficiently large dimensions, the trace distance concentrates exponentially.

\begin{lemma}[Prop.~2 \cite{Puchala16}]\label{lemma:exp_concentration_of_1norm}
    Let $\rho$ and $\sigma$ be two random states distributed according to the Hilbert-Schmidt distribution on $\cS(\cH_d)$. Then, for sufficiently large $d$ the trace distance exponentially concentrates around $D=\frac{1}{4}+\frac{1}{\pi}$,
    \begin{equation}
          \Pr[\abs{\frac{1}{2}\norm{\rho-\sigma}_1- D} > \delta]\leq 2 \exp(-\lambda d^2 \delta^2)
    \end{equation}
    where $\lambda = 1 / (4\cdot 9\pi^3)$.
\end{lemma}

Combining this with Lemma~\ref{lemma:tailbound2ratiobound} yields our first upper bound for the average contraction of the trace distance.

\begin{theorem}\label{thm:HSaveragecontraction}
    Let $T:\cS(\cH_d)\to\cS(\cH_{d'})$ be a quantum channel, $\tau$ be its Choi state and $\pi=T(\frac{\mathbb{I}}{d})$. Then the average contraction of the trace distance under the Hilbert-Schmidt measure $\nuHS$ is upper bounded by
    \begin{equation}
    \eta_1(T,\norm{\cdot}_1,\nuHS\times\nuHS)
    %\bEx{\rho,\sigma\sim\nuHS}\qty[ \frac{\norm{T(\rho)-T(\sigma)}_1}{\norm{\rho-\sigma}_1}]
    \leq 
    \alpha \sqrt{\frac{d'}{d}} \sqrt{\Tr \tau^2 - \frac{1}{d}\Tr \pi^2}
    +\order{\frac{\sqrt{\ln d}}{d}}
    ,
    \end{equation}
    where $\alpha=\frac{2\sqrt{2}\pi}{4+\pi}\approx 1.24$.
\end{theorem}

\begin{proof}
	In Lemma~\ref{lemma:tailbound2ratiobound}, let us consider a tail bound with $\mu_y=D=\frac{1}{4}+\frac{1}{\pi}$ and $t(\delta)=2\exp(-\lambda d^2 \delta^2)$ for the denominator motivated by Lemma~\ref{lemma:exp_concentration_of_1norm}. Then, for sufficiently large $d$ we can choose $\delta = \sqrt{\ln(\lambda d^2)  / \lambda d^2} < D$ and we arrive at
	\begin{equation}
		\bE\qty[\frac{X}{Y}] - \frac{\mu_x}{D}
		\leq
		\frac{\mu_x}{D} \frac{\sqrt{\ln(\lambda d^2)/ \lambda d^2}}{D-\sqrt{\ln(\lambda d^2)/ \lambda d^2}} + 2\frac{1}{\lambda d^2}
		=
		\order{\frac{\sqrt{\ln d}}{d}}
		.
	\end{equation}
	Therefore,
	\begin{equation}
		\bE\qty[\frac{\frac{1}{2}\norm{T(\rho)-T(\sigma)}_1}{\frac{1}{2}\norm{\rho-\sigma}_1}] \leq
		\frac{1}{D}\bE\qty[\frac{1}{2}\norm{T(\rho)-T(\sigma)}_1] + \order{\frac{\sqrt{\ln d}}{d}}.
	\end{equation}
	
	Applying the relation between norms $\norm{X}_1\leq \sqrt{d'}\norm{X}_2$ for matrices of size $d'$ and that by Jensen's inequality $\bE[x]\leq \sqrt{\bE[x^2]}$ for any real variable $x$:
	\begin{equation}
		\bE\qty[\norm{T(\rho)-T(\sigma)}_1]
		\leq
		\sqrt{d'}\bE\qty[\norm{T(\rho)-T(\sigma)}_2]
		\leq \sqrt{d'}\sqrt{\bE\qty[\norm{T(\rho)-T(\sigma)}_2^2]}.
	\end{equation}
	The expectation value inside the square root can be computed as in App.~\ref{sec:expvals_unit_invar} with $\bE[\Tr\rho^2]=2d/(d^2+1)$,
	\begin{align}
		\bE\qty[\norm{T(\rho)-T(\sigma)}_2^2]
		&=
		\frac{2d^2}{d^2-1}\qty(\frac{2d}{d^2 + 1}-\frac{1}{d}) \qty[\Tr \tau^2 - \frac{1}{d} \Tr \pi^2] \\
		&=
		\frac{2d}{d^2+1} \qty[\Tr \tau^2 - \frac{1}{d} \Tr \pi^2]
		\leq 
		\frac{2}{d} \qty[\Tr \tau^2 - \frac{1}{d} \Tr \pi^2]
		.
	\end{align}
	Defining $\alpha=1/\sqrt{2}D$ and putting it all together shows the result. 
\end{proof}

The upper bound vanishes when $d \Tr\tau^2 = \Tr\pi^2$, which happens if and only if the channel is a replacer channel as presented in Fact~\ref{fact:purities}. In the other extreme, subject to $d=d'$, the maximum value is obtained for a unitary channel, which gives
\begin{equation}
    \eta_1(T,\norm{\cdot}_1,\nuHS\times\nuHS)=1
    \leq 
    \alpha \sqrt{1 - \frac{1}{d^2}}
    +\order{\frac{\sqrt{\ln d}}{d}}
    =\alpha + \order{\frac{\sqrt{\ln d}}{d}},
\end{equation}
which is off by a factor $\alpha$. When the output dimension is larger than the input one, $d'\geq d$, the bound requires smaller values for the purity of the Choi matrix to be informative. Oddly, the bound apparently diverges as $d'\to\infty$ for channels $T(\rho)=V \rho V^{\dag}$ consisting of a conjugation by an isometry $V:\mathbb{C}^d\to\mathbb{C}^{d'}$,
\begin{equation}
    \eta_1(T,\norm{\cdot}_1,\nuHS\times\nuHS)
    \leq 
    \alpha \sqrt{\frac{d'}{d}} \sqrt{1 - \frac{1}{d^2}}
    +\order{\frac{\sqrt{\ln d}}{d}}.
\end{equation}
However, note that in this case the bound $\norm{V(\rho-\sigma)V^{\dag}}_1 \leq \sqrt{d'} \norm{V(\rho-\sigma)V^{\dag}}_2$ can be improved noticing that the conjugation with $V$ preserves the rank of $\rho-\sigma$, so $\norm{V(\rho-\sigma)V^{\dag}}_1 \leq \sqrt{d} \norm{V(\rho-\sigma)V^{\dag}}_2$ and we recover the bound of the unitary case. When $d' \gg d$, one can generally consider substituting $d'$ by the product of $d$ with the Kraus-rank of the channel in an attempt to get a tighter result. Alternatively, the case $d'\leq d$ leads to a descriptive bound presented in the following corollary.

\begin{corollary} \label{cor:dtodprime}
    Let $T:\cS(\cH_d) \to \cS(\cH_{d'})$ be a quantum channel with Choi state $\tau$ and define the state $\pi = T\qty(\frac{\mathbb{I}}{d})$. The averate trace distance contraction $\eta_1(T,\norm{\cdot}_1,\nuHS\times\nuHS)$ is upper bounded by
	\begin{equation}
		\eta_1(T,\norm{\cdot}_1,\nuHS\times\nuHS)
		\leq 
	    \alpha \frac{d'}{d}
   		+\order{\frac{\sqrt{\ln d}}{d}}.
	\end{equation}
\end{corollary}

\begin{proof}
	It is a direct consequence of $\Tr\tau^2\leq d'/d$, demonstrated in Fact~\ref{fact:purities}.
\end{proof}

Therefore, when $d'$ is significantly smaller than $d$, a large amount of distinguishability between states is lost on average.

The strong concentration inequality induced by the Hilbert-Schmidt distribution allows us to demonstrate similar bounds for larger orders of the $p$-moments of contraction.

\begin{corollary}\label{cor:2ndmomentHS}
    Let $T:\cS(\cH_d)\to\cS(\cH_{d'})$ be a quantum channel, $\tau$ be its Choi state and $\pi=T(\frac{\mathbb{I}}{d})$. Then the $p$-moment of the trace distance under the Hilbert-Schmidt measure $\nuHS$ is upper bounded by
    \begin{equation}
    \eta_p(T,\norm{\cdot}_1,\nuHS\times\nuHS)
    \leq 
    \alpha \sqrt{d'} \qty(\bEx{\rho,\sigma\sim\nuHS}\qty[\norm{T(\rho)-T(\sigma)}_2^p])^{1/p}
    +\order{\frac{\sqrt{\ln d}}{d}}
    ,
    \end{equation} 
    where $\alpha=\frac{2\sqrt{2}\pi}{4+\pi}\approx 1.24$.
    In particular, for $p=2$
    \begin{equation}
    \eta_2(T,\norm{\cdot}_1,\nuHS\times\nuHS)
    \leq 
    \alpha \sqrt{\frac{d'}{d}} \sqrt{\Tr \tau^2 - \frac{1}{d}\Tr \pi^2}
    +\order{\frac{\sqrt{\ln d}}{d}}.
    \end{equation} 
\end{corollary}

\begin{proof}
	To apply Lemma~\ref{lemma:tailbound2ratiobound} we need a tail bound on $(\frac{1}{2}\norm{\rho-\sigma}_1)^p$. With this aim, we find that for $0\leq x \leq 1$
	\begin{equation}\label{eq:ineq_for_moments}
		\abs{x^p-D^p}\leq \frac{1-D^p}{1-D}\abs{x-D},
	\end{equation}
	which is trivially true for $x=D$. To see this, define the following function in the restricted domain $x\in [0,D)\cup (D,1]$,
	\begin{equation}
		F_p(x) = \frac{x^p-D^p}{x-D}.
	\end{equation}
	Then, Eq.~\eqref{eq:ineq_for_moments} is equivalent to $F_p(x)\leq F_p(1)$, which is verified observing that $F_p$ is monotonically increasing for $x > 0$. Consider its derivative 
	\begin{equation}
		(F_p)'(x) = \frac{(p-1) x^p - p x^{p-1} D + D^p}{(x-D)^2},
	\end{equation}
	and introduce the inequality $(F_p)'(x)\geq 0$ arranging the numerator in the following form
	\begin{equation}
		\frac{p-1}{p} + \frac{1}{p} \qty(\frac{D}{x})^p \geq \frac{D}{x}.
	\end{equation}
	Inspecting the latter, we identify Young's inequality \cite{Hardy1934}, which claims that for $a,b >0$ and $p,q>1$ with $\frac{1}{p}+\frac{1}{q}=1$ it holds that
	\begin{equation}
		\frac{a^p}{p}+\frac{b^q}{q} \geq ab,
	\end{equation}
	by setting $a=D/x$ and $b=1$.
	
	A direct consequence of Eq.~\eqref{eq:ineq_for_moments} is
    \begin{equation}
          \Pr[\abs{\qty(\frac{1}{2}\norm{\rho-\sigma}_1)^p-D^p} > \varepsilon]
          \leq
          \Pr[\frac{1-D^p}{1-D}\abs{\frac{1}{2}\norm{\rho-\sigma}_1-D} > \varepsilon]
          \leq
          2 \exp(-\lambda d^2 \qty(\frac{1-D}{1-D^p})^2 \varepsilon^2),
    \end{equation}
    with $\lambda = 1/(4 \cdot 9\pi^3)$. Thus,$\qty(\frac{1}{2}\norm{\rho-\sigma}_1)^p$ is exponentially concentrated around $D^p$ and, as before, we obtain
    \begin{equation}
		\bE\qty[\qty(\frac{\frac{1}{2}\norm{T(\rho)-T(\sigma)}_1}{\frac{1}{2}\norm{\rho-\sigma}_1})^p]
		\leq
		\frac{1}{D^p}\bE\qty[\frac{1}{2^p}\norm{T(\rho)-T(\sigma)}_1^p] + \order{\frac{\sqrt{\ln d}}{d}}.
	\end{equation}
	We arrive at the desired result using $\norm{X}_1 \leq \sqrt{d'} \norm{X}_2$. The particular case $p=2$ can be explicitly computed following App.~\ref{sec:expvals_unit_invar}.
\end{proof}

To explore the tightness of Thm.~\ref{thm:HSaveragecontraction}, let us consider the $n$-qubit global depolarizing channel 
\begin{equation}
    \cD_p(\rho) =  (1-p) \rho + p \frac{\mathbb{I}^{\otimes n}}{2^n}.
\end{equation}
In terms of trace distance contraction this channel is trivial as $\norm{\cD_p(\rho)-\cD_p(\sigma)}_1=\abs{1-p}\norm{\rho-\sigma}_1$ for any input states $\rho,\sigma\in\cS(\cH_2^{\otimes n})$, which sets a convenient reference value. We can easily calculate the Choi state and its purity,
\begin{equation}
    \tau = (1-p)\ketbra{\Omega} + p \frac{\mathbb{I}^{\otimes 2n}}{2^{2n}}
    \quad \text{and} \quad
    \Tr\tau^2 =(1-p)^2+\frac{1-(1-p)^2}{2^{2n}}=\qty(1-2^{-2n})(1-p)^2+ 2^{-2n},
\end{equation}
which leads to the upper bound,
\begin{align}
    \eta_1(\cD_p,\norm{\cdot}_1,\nuHS\times\nuHS)
    =\abs{1-p}
    &\leq
    \alpha \sqrt{\qty(1-2^{-2n})(1-p)^2+\frac{1}{2^{2n}} - \frac{1}{2^{2n}}} +\order{n^{1/2} 2^{-n}} \\
    &=
    \alpha \abs{1-p} +\order{n^{1/2} 2^{-n}}.
\end{align}
We see that our bound is nearly tight for the case of the global depolarizing channel, as it only differs by a small constant multiplicative factor $\alpha\approx 1.24$ and an exponentially vanishing correction $\order{n^{1/2} 2^{-n}}$. Therefore, while the bound is not exactly optimal, any potential improvements would be constrained to refining these minor factors rather than achieving a qualitatively sharper scaling.

The bound becomes particularly illustrative when evaluated for product channels, showing exponential convergence:

\begin{corollary}\label{cor:travcontract_nproduct}
    Let $T = \Phi^{\otimes n}$ be a product channel, with local channel $\Phi:\cS(\cH_d)\to\cS(\cH_{d})$, and let $\tau$ be the Choi state of the local channel and $\pi = \Phi\qty(\frac{\mathbb{I}}{d})$ be the average output state. Then, the average trace distance contraction is bounded as follows:
    \begin{equation}\label{eq:travcontract_nproduct}
        \eta_1(T,\norm{\cdot}_1,\nuHS\times\nuHS)
        %\bEx{\rho,\sigma \sim \nuHS}\qty[\frac{\norm{T^{\otimes n}(\rho) -T^{\otimes n}(\sigma )}_1}{\norm{\rho -\sigma }_1}]
        \leq
        \alpha \sqrt{(\Tr \tau^2)^n - (\frac{1}{d}\Tr \pi^2)^n}
        + \order{\frac{\sqrt{n\ln d}}{d^n}},
    \end{equation}
     where $\alpha=\frac{2\sqrt{2}\pi}{4+\pi}\approx 1.24$.
    In particular,
    \begin{equation}\label{eq:travcontract_nproduct_simplified}
        \eta_1(T,\norm{\cdot}_1,\nuHS\times\nuHS)
        %\bEx{\rho,\sigma \sim \nuHS}\qty[\frac{\norm{T^{\otimes n}(\rho) -T^{\otimes n}(\sigma )}_1}{\norm{\rho -\sigma }_1}]
        \leq
        \alpha (\Tr \tau^2)^{n/2} +
        \order{\frac{\sqrt{n\ln d}}{d^n}}.
    \end{equation}
\end{corollary}

\begin{proof}
    The result follows from the multiplicativity of the Choi state, that is the Choi state for a product channel $T_1 \otimes T_2$ is the tensor product of the individual Choi states $\tau_1 \otimes \tau_2$, up to the ordering of subsystems.

    For the simplified bounds, we use $\Tr \pi^2 
    \geq \frac{1}{d}$ and that 
    \begin{equation}
        \frac{d^{2n} x - 1}{d^{2n} - 1} \leq x \quad \text{for} \quad 0 \leq x \leq 1.
    \end{equation}
\end{proof}

In this section, we have illustrated the proof strategy, which consists of leveraging the concentration of the trace distance to substitute the denominator with an effective lower bound and controlling the numerator with alternative techniques. We have found that when $\rho$ and $\sigma$ are independently distributed according to the Hilbert–Schmidt distribution, the $p$-moments of the trace distance can be controlled using those of the 2-norm, which are often more tractable. In particular, for the average contraction, we have identified two significant scenarios where it is strongly suppressed, namely, when the output dimension is much smaller than the input dimension, and for product channels. Indeed, for the latter, we established that the average contraction vanishes exponentially whenever the channel is non-unitary, i.e. when $\Tr \tau^2 < 1$. In the following section, we attempt a similar analysis for other distributions.

\subsubsection{Uniform 2-designs over pure states}\label{sec:upperbound_2designs}
Although the results of the last section, like Cor.~\ref{cor:travcontract_nproduct}, show that the average contraction will be exponential on average, one might rightfully argue that the ensemble we chose was already concentrated on highly mixed states. For instance, the average purity of the ensemble is already exponentially small before we apply the channel. In this section, we will study the average contraction w.r.t. pure states stemming from a 2-design \cite{Dankert2009}, establishing qualitatively similar bounds. Although we give the general definition for $t$-designs, in this Section we resort to uniform distributions for simplicity.
 
\begin{definition}[Unitary and state $t$-designs]\label{def:t_designs}
	Let $\nu$ be a probability distribution over the unitary group $U(d)$. Then $\nu$ is called a \emph{unitary $t$-design} if for every polynomial $P_{t,t}$ of degree at most $t$ in $U$ and $U^\dag$
	\begin{equation}
		\bEx{U \sim \nu} \qty[ P_{t,t}(U) ] = \bEx{U \sim \mu_H} \qty[ P_{t,t}(U) ]
	\end{equation}
	where $\mu_H$ denotes the Haar measure on $U(d)$.
	
	We say that a distribution $\tilde{\nu}$ over states $\cH_d$ is a state $t$-design if there is a reference state $\ket{\phi_0}$ and a unitary $t$-design $\nu$ such that $\tilde{\nu}$ is equivalent to the distribution of $U\ket{\phi_0}$ with $U$ distributed according to $\nu$.
\end{definition}

Directly defining the probability distribution over pairs of states $(\rho, \sigma) \sim \nu \times \nu$ by taking them independently and identically distributed comes with the risk of a non-zero collision probability $\xi = \Pr[\rho = \sigma] > 0$. To prevent this, we define the distribution $\nu^{(2)}$ as the product distribution $\nu \times \nu$ conditioned on the states being different; that is, the probability of a set $S \subseteq \cS(\cH_d) \times \cS(\cH_d)$ is given by
\begin{equation}
	\nu^{(2)}(S) = \frac{1}{1-\xi} (\nu\times\nu)(\qty{ (\rho, \sigma) \in S \mid \rho \neq \sigma }).
\end{equation}
In any case, the effect of collisions decreases rapidly with system size, as $\Pr[\rho=\sigma]=\order{d^{-4}}$ when $\rho$ and $\sigma$ have independent 2-design distributions (see Lemma~\ref{lemma:2design_collision}).

It is a standard fact that states pairs $\rho,\sigma$ distributed independently according to a $t$-design distribution satisfy a tail bound of the form required in Lemma~\ref{lemma:tailbound2ratiobound} with $t(\delta) = 1/(d \delta)$ (see Lemma~\ref{lemma:tailbound2ratiobound}). This leads to the following.

\begin{theorem}\label{thm:avc_from_operator_root_bound}
    Let $T:\cS(\cH_d)\to\cS(\cH_{d'})$ be a quantum channel, $\tau$ be its Choi state and $\pi=T(\frac{\mathbb{I}}{d})$. Let $\nu$ be a 2-design over pure states in $\cS(\cH_d)$ and denote by $\nu^{(2)}$ the distribution for $(\rho,\sigma)\sim\nu\times\nu$ conditioned on $\rho\neq \sigma$.
    Then, the average contraction of the trace distance is upper bounded by
    \begin{equation}
    \eta_1(T,\norm{\cdot}_1,\nu^{(2)})
        \leq
        \frac{1}{2}
        \Tr (
            \frac{2}{d+1}\qty[d\Tr_2 \tau^2 - \pi^2]
        )^{1/2} + \order{\frac{1}{\sqrt{d}}}.
    \end{equation}
\end{theorem}

\begin{proof}
    We use the fact that the distribution $\nu$ is a 1-design too, thus, in Lemma~\ref{lemma:tailbound2ratiobound}, we can use the tail bound on the trace distance in Lemma~\ref{lemma:concentrationoftrdist} to set $\mu_y=1$ and $t(\delta)=1/d\delta$ with collision probability $\xi=\order{d^{-4}}$ by Lemma~\ref{lemma:2design_collision}. This leads to
    \begin{align}
        \bEx{(\rho,\sigma)\sim\nu^{(2)}}\qty[\frac{\frac{1}{2}\norm{T(\rho)-T(\sigma)}_1}{\frac{1}{2}\norm{\rho-\sigma}_1}]
        \leq
        \frac{1}{1-\xi}\qty[\bEx{\rho,\sigma\sim\nu\times\nu}\qty[\frac{1}{2}\norm{T(\rho)-T(\sigma)}_1]\qty(1+\frac{\delta}{1-\delta}) + \frac{1}{d \delta} ]      
        .
    \end{align}
    To proceed, we find the optimal value for $\delta$. Given $A>0$ and imposing $0<\delta<1$
    \begin{equation}
    	\dv{\delta}\qty[A \qty(1+\frac{\delta}{1-\delta}) + \frac{1}{d \delta} ] (\delta_*) = 
    	A \frac{1}{(1-\delta_*)^2} - \frac{1}{d \delta_*^2} = 0,
    	\quad \delta_* = \frac{1}{1+\sqrt{d A}},
    \end{equation}
    which gives the optimal value
    \begin{equation}
    	A \frac{1}{1-\delta_*} + \frac{1}{d \delta_*}
    	=
    	A \frac{1+\sqrt{dA}}{\sqrt{dA}} + \frac{1+\sqrt{dA}}{d}
    	= (1+\sqrt{d A})\qty(\sqrt{\frac{A}{d}} + \frac{1}{d})
		= A + 2\sqrt{\frac{A}{d}} + \frac{1}{d}
    	.
    \end{equation}
    Introducing this result,    
    \begin{align}
        \bEx{(\rho,\sigma)\sim\nu^{(2)}}\qty[\frac{\frac{1}{2}\norm{T(\rho)-T(\sigma)}_1}{\frac{1}{2}\norm{\rho-\sigma}_1}]
        \leq
        \frac{1}{1-\xi}\qty(\bEx{\rho,\sigma\sim\nu\times\nu}\qty[\frac{1}{2}\norm{T(\rho)-T(\sigma)}_1]
        + 2\sqrt{\frac{1}{d}} + \frac{1}{d})    
        ,
    \end{align}
    where we substituted $A = \frac{1}{2}\norm{T(\rho)-T(\sigma)}_1\leq 1$.

    Now, using
    \begin{equation}
        \mathbb{E}[\norm{X}_1] = 
        \mathbb{E}[\Tr \sqrt{X X^{\dag})}] = 
        \Tr \mathbb{E}[ \sqrt{X X^{\dag}}] \leq 
        \Tr \sqrt{\mathbb{E}[X X^{\dag}]},
    \end{equation}
    where the last inequality follows from Jensen's inequality for the operator concave function $f(x)=\sqrt{x}$ \cite{Ahlswede2001}. Therefore, by the results in App.~\ref{sec:expvals_unit_invar} for states $\rho$ and $\sigma$ with independent 2-design distributions we have
    \begin{equation}\label{eq:operator_root_bound}
        \mathbb{E}[(T(\rho)-T(\sigma))^2]
        =
            \frac{2}{d+1}\qty[d\Tr_2 \tau^2 - \pi^2]
        .
    \end{equation}
    The final inequality follows from introducing this bound, gathering the error terms, and neglecting the correction due to collisions $1/(1-\xi)=1+\order{d^{-4}}$.
\end{proof}

\begin{remark}
    A pathway to generalizing the result is to observe that the key quantities involved are $\bE[T(\rho)]$, $\bE[T^2(\rho)]$, and a tail bound on the trace distance. Consequently, any distribution for which these quantities can be controlled admits a similar upper bound of the form $\frac{1}{2}\Tr\sqrt{\bE[T^2(\rho)] - (\bE[T(\rho)])^2}$, with an error term that depends on the strength of the tail bound.
\end{remark}

Our motivation is to study the asymptotic regime for this bound, especially for product channels.

\begin{corollary}\label{cor:travcontract_nproduct_2designs}
    Let $T = \Phi^{\otimes n}$ be a product channel, with local channel $\Phi:\cS(\cH_d)\to\cS(\cH_{d})$, and let $\tau$ be the Choi state of the local channel and $\pi = \Phi\qty(\frac{\mathbb{I}}{d})$ be the average output state. Let $\nu$ be a 2-design over pure states in $\cS(\cH_d^{\otimes n})$ and denote by $\nu^{(2)}$ the distribution for $(\rho,\sigma)\sim\nu\times\nu$ conditioned on $\rho\neq \sigma$.
    Then, the average trace distance contraction is bounded as follows:
    \begin{equation}\label{eq:travcontract_nproduct_2designs}
        \eta_1(\Phi^{\otimes n},\norm{\cdot}_1,\nu^{(2)})
        \leq
        \frac{1}{2}
        \Tr (
            \frac{2}{d^n+1}\qty[d^n (\Tr_2 \tau^2)^{\otimes n} - (\pi^2)^{\otimes n}]
        )^{1/2} + \order{\frac{1}{\sqrt{d^n}}}.
    \end{equation}
    In particular,
    \begin{equation}\label{eq:travcontract_nproduct_2designs_simple}
        \eta_1(\Phi^{\otimes n},\norm{\cdot}_1,\nu^{(2)})
        \leq
        \frac{1}{\sqrt{2}}
        \qty(\Tr \sqrt{\Tr_2 \tau^2})^n + \order{\frac{1}{\sqrt{d^n}}}.
    \end{equation}
\end{corollary}

\begin{proof}
    The first inequality is a direct consequence of Thm.~\ref{thm:avc_from_operator_root_bound}, noting that
    \begin{equation}
        \tau(T)=\tau(\Phi)^{\otimes n}, \quad \text{and} \quad
        \pi(T)=\Phi^{\otimes n}(\frac{\mathbb{I}}{d^n})=\pi(\Phi)^{\otimes n}.
    \end{equation}
    The second inequality requires noting that as the square root is operator monotone, $A\leq B$ implies $\sqrt{A}\leq\sqrt{B}$~\cite[Chapter 5]{MWolf2012Guidedtour}, thus
    \begin{equation}\label{eq:neglecting_pi2}
        \qty(
            \frac{2}{d^n+1}\qty[d^n (\Tr_2 \tau^2)^{\otimes n} - (\pi^2)^{\otimes n}]
        )^{1/2}
        \leq
        \qty(
            \frac{2}{d^n+1}\qty[d^n (\Tr_2 \tau^2)^{\otimes n}]
        )^{1/2}.
    \end{equation}
    Taking constants out of the trace and using $\Tr A^{\otimes n}=(\Tr A)^n$,
    \begin{align}
        \frac{1}{2}
        \Tr (
            \frac{2}{d^n+1}\qty[d^n (\Tr_2 \tau^2)^{\otimes n} - (\pi^2)^{\otimes n}]
        )^{1/2}
        &\leq 
        \frac{1}{2}
        \Tr (
            \frac{2}{d^n+1}\qty[d^n (\Tr_2 \tau^2)^{\otimes n}]
        )^{1/2} \\
        &\leq 
        \sqrt{\frac{d^n}{2(d^n+1)}}
        \qty(\Tr (
            \sqrt{\Tr_2 \tau^2}
        ))^{n}. \\
    \end{align}
    Finally, use that $\sqrt{d^n/(d^n+1)}=1-\frac{d^{-n}}{2}+\order{d^{-2n}}$ and that $\Tr \sqrt{\Tr_2 \tau^2}\leq \sqrt{d}\sqrt{\Tr\tau^2}\leq \sqrt{d}$, again because $\sqrt{x}$ is operator concave:
    \begin{align}
        \sqrt{\frac{d^n}{2(d^n+1)}}
        \qty(\Tr (
            \sqrt{\Tr_2 \tau^2}
        ))^{n}
        &\leq
        \frac{1}{\sqrt{2}}
        \qty(\Tr (
            \sqrt{\Tr_2 \tau^2}
        ))^{n}+\order{d^{-2n} (\Tr \sqrt{\Tr_2 \tau^2})^n}\\
        &\leq
        \frac{1}{\sqrt{2}}
        \qty(\Tr (
            \sqrt{\Tr_2 \tau^2}
        ))^{n}+\order{d^{-2n} d^{n/2}}.
    \end{align}
\end{proof}

\begin{remark}\label{remark:avc2design_nproduct_about_simplificaiton}
	The simplification in Eq.~\eqref{eq:travcontract_nproduct_2designs_simple} is in general weaker than the complete version in Eq.~\eqref{eq:travcontract_nproduct_2designs}. However, when $\Tr_2\tau^2$ and $\pi^2$ commute, the upper bound in Eq.~\eqref{eq:travcontract_nproduct_2designs} vanishes if and only if the bound in Eq.~\eqref{eq:travcontract_nproduct_2designs_simple} does. In these cases, verifying whether $\Tr\sqrt{\Tr_2\tau^2} < 1$ is enough to guarantee that the upper bound is vanishing.
	
	The key observation is that for $A$ and $B$ commuting matrices with $A\geq B \geq 0$ it holds that $\sqrt{A-B}\geq \sqrt{A} - \sqrt{B}$, therefore
	\begin{equation}
		\Tr\sqrt{(\Tr_2\tau^2)^{\otimes n} - \frac{1}{d^n} (\pi^2)^{\otimes n}}
		\geq
		\Tr\sqrt{(\Tr_2\tau^2)^{\otimes n}} - \frac{1}{d^{n/2}} \Tr\sqrt{(\pi^2)^{\otimes n}}
		=(\Tr\sqrt{\Tr_2\tau^2})^n - \frac{1}{d^{n/2}},
	\end{equation}
	which shows that the complete bound becomes trivial whenever the simplified bound does.
\end{remark}

So far, we have been fundamentally guided by the interest in characterizing the contraction of the distance between i.i.d. states, which has led us to find similar bounds for 2-design distributions as for the Hilbert-Schmidt distribution. However, in many practical settings, interest centers on how close a state is to the maximally mixed state~\cite{Quek2024}, which serves as a benchmark for complete loss of information. It is thus pertinent to establish a similar bound in this context:

\begin{corollary}\label{cor:avc_2design_mix}
    Let $T:\cS(\cH_d)\to\cS(\cH_{d'})$ be a quantum channel, $\tau$ be its Choi state and $\pi=T(\frac{\mathbb{I}}{d})$. Let $\nu$ be a 2-design over pure states in $\cS(\cH_d)$.
    Then, the average contraction of the trace distance is upper bounded by
    \begin{equation}
    \eta_1(T,\norm{\cdot}_1,\nu\times\delta_{\mathbb{I}/d})
        \leq
        \frac{1}{2(1-1/d)}
        \Tr (
            \frac{1}{d+1}\qty[d\Tr_2 \tau^2 - \pi^2]
        )^{1/2}.
    \end{equation}
\end{corollary}

\begin{proof}
    In this case, the denominator takes a constant value:
    \begin{align}
        \frac{1}{2}\norm{\rho -\frac{\mathbb{I}}{d}}_1 
        = \frac{1}{2} \qty(1-\frac{1}{d} + \frac{1}{d}(d-1)) = 1 - \frac{1}{d}.
    \end{align}
    We use $\bE[\norm{X}_1]\leq \Tr \sqrt{\bE[X X^{\dag}]}$ and the results in App.~\ref{sec:expvals_unit_invar}
    \begin{equation}
        \bE[(T(\rho)-T(\mathbb{I}/d))^2]
        =
           \frac{1}{d+1}\qty[d\Tr_2 \tau^2 - \pi^2]
        .
    \end{equation}
\end{proof}
\begin{remark}
	By a similar argument to the proof of Cor.~\ref{cor:travcontract_nproduct_2designs}, the upper bound can be upper bounded by $\frac{1}{2}(\Tr\sqrt{\Tr_2 \tau^2})^n + \order{d^{-n/2}}$ if we consider $T=\Phi^{\otimes n}$ and $\nu$ a 2-design distribution on $\cS(\cH_d^{\otimes n})$.
\end{remark}

Theorem~\ref{thm:avc_from_operator_root_bound}, Cor.~\ref{cor:travcontract_nproduct_2designs}, and Cor.~\ref{cor:avc_2design_mix}, all showcase bounds with a structure that qualitatively corresponds to the bound in Thm.~\ref{thm:HSaveragecontraction} in the previous section. Indeed, they all have a positive term that depends on the square of the Choi matrix and a negative term depending on the square of the image of the maximally mixed, namely $\Tr_2\tau^2$ corresponding to $\Tr\tau^2$ and $\pi^2$ to $\Tr \pi^2$. In contrast to the previous section, where exponentially vanishing average contraction was established for all non-unitary channels, we see that Cor.~\ref{cor:travcontract_nproduct_2designs} only demonstrates it with respect to 2-designs if the channel is sufficiently noisy, in the sense of $\Tr\sqrt{\Tr_2\tau^2} < 1$. In fact, the bound becomes trivial when this condition fails, leaving open the question of whether all non-unitary channels exhibit exponentially decaying average contraction. In the following section, we provide lower bounds that demonstrate a negative answer.

\subsection{Lower bounds on the contraction coefficient for the trace distance}\label{sec:lower_bounds_trace}

The bounds presented so far for the average contraction of product channels have relied heavily on the mixedness of the Choi state associated with the channel. Notably, when the Choi state is sufficiently pure, these bounds become trivial and fail to provide meaningful insights. This raises the question: Is this limitation merely a byproduct of the proof technique, or does there exist a critical purity level of the channel beyond which exponentially vanishing average contraction is fundamentally unattainable? In this section, we address this question and establish that the latter is indeed the case, at least for 1-design distributions on pure states. Specifically, we derive lower bounds on average contraction which demonstrate that, for sufficiently low entropy of the Choi state, the average contraction approaches unity. This result implies that, on average, pure states retain their essentially perfect distinguishability under such channels.

To illustrate this phenomenon, we begin by analyzing a simple yet insightful family of channels whose asymptotic average contraction can be fully characterized. Consider an $N$-qubit system and the partial trace channel $\mathcal{E}_M(\rho) = \Tr_{1\dots M}[\rho]$, defined by discarding the first $M$ qubits. Here, we examine the asymptotic average contraction of $\mathcal{E}_M^{\otimes n}$ over Haar-random pure states as a function of the fraction of discarded qubits, $M/N$. Note that the choice of the discarded qubits does not change the contraction due to the permutation invariance of the Haar distribution.

\begin{proposition}[Average contraction of partial trace in the asymptotic limit]\label{prop:avc_partial_trace_asymp_lim}
    Let $\mathcal{E}_M(\rho) = \Tr_{1\dots M}[\rho]$ denote the $N$-to-$(N-M)$-qubit channel obtained by discarding the first $M$ qubits. Let $\mu_{2^{nN}}$ represent the Haar distribution over pure states on $nN$ qubits. Then the average contraction satisfies: 
    \begin{align}
        \lim_{n \to \infty}
        \eta_1(\mathcal{E}_M^{\otimes n},\norm{\cdot}_1,\mu_{2^{nN}}\times\mu_{2^{nN}})
        %\bEx{\rho, \sigma \sim \mu_{2^{nN}}} \qty[ \frac{\norm{\mathcal{E}_M^{\otimes n}(\rho) - \mathcal{E}_M^{\otimes n}(\sigma)}_1}{\norm{\rho - \sigma}_1}] &= 
        \begin{cases}
            1 & \text{if} \quad M < N/2, \\
            \frac{1}{4} + \frac{1}{\pi} & \text{if} \quad M = N/2, \\
            0 & \text{if} \quad M > N/2.    
        \end{cases}
    \end{align}
\end{proposition}

\begin{proof}
    As shorthand, we use $d=2^N$ for the dimension of the whole system and $d'=2^{N-M}$ for the dimension of the output system.
    
    For the $M<N/2$ regime, we can find a lower bound that converges to one. We start by bounding the denominator with $\norm{\rho-\sigma}_1\leq 2$, and focusing on
    \begin{align}
        \bEx{\rho,\sigma\sim \mu_{d^n}}\qty[\frac{\norm{\mathcal{E}_M^{\otimes n} (\rho) - \mathcal{E}_M^{\otimes n}(\sigma)}_1}{\norm{\rho-\sigma}_1}]
        &\geq 
        \bEx{\rho,\sigma\sim \mu_{d^n}}\qty[\frac{1}{2}\norm{\mathcal{E}_M^{\otimes n} (\rho) - \mathcal{E}_M^{\otimes n}(\sigma)}_1]\\
        &\geq 
        \bEx{\rho,\sigma\sim \mu_{d^n}}\qty[\frac{1}{2}\norm{\mathcal{E}_{nM}(\rho) - \mathcal{E}_{nM}(\sigma)}_1]
        ,
    \end{align}
    where we used the permutation invariance of the Haar distribution and, with a little abuse of notation, denoted by $\mathcal{E}_{nM}$ the $nN$-to-$n(N-M)$ channel obtained by discarding the first $nM$ qubits of an $nN$-qubit system. The trace distance can be lower bounded by considering $P$ the orthogonal projector onto the support of $\mathcal{E}_{nM}(\rho)$,
    \begin{equation}
        \frac{1}{2}\norm{\mathcal{E}_{nM}(\rho) - \mathcal{E}_{nM}(\sigma)}_1 \geq \Tr[P(\mathcal{E}_{nM}(\rho) - \mathcal{E}_{nM}(\sigma))] = 1 - \Tr[P \mathcal{E}_{nM}(\sigma)].
    \end{equation}
    The rank of the projector equals the state's rank, bounded by the discarded subsystem's dimension $\rank(\mathcal{E}_{nM}(\rho))\leq (d/d')^n$, hence $\Tr P \leq (d/d')^n$ and
    \begin{equation}
        \bEx{\rho,\sigma\sim \mu_{d^n}}[\Tr[P \mathcal{E}_{nM}(\sigma)]] =
        \bEx{\rho\sim \mu_{d^n}}\qty[\Tr[P\; \bEx{\sigma\sim \mu_{d^n}}[\mathcal{E}_{nM}(\sigma)]]] =
        \bEx{\rho\sim \mu_{d^n}}\qty[\Tr[P \frac{\mathbb{I}_{N-M}^{\otimes n}}{d'^n}]] \leq \qty(\frac{(d/d')}{d'})^n.
    \end{equation}
    Consequently,
    \begin{equation}
        \bEx{\rho,\sigma\sim \mu_{d^n}}\qty[\frac{\norm{\mathcal{E}_M^{\otimes n} (\rho) - \mathcal{E}_M^{\otimes n}(\sigma)}_1}{\norm{\rho-\sigma}_1}]
        \geq
        1 - \qty(\frac{d}{d'^2})^n.
    \end{equation}
    Hence, if $d < d'^2$ the average contraction approaches one as $n\rightarrow\infty$. Equivalently, in terms of qubits, $M < N-M$ and $M < N/2$.

    For the $M>N/2$ regime we apply the upper bound in Thm.~\ref{thm:avc_from_operator_root_bound}, which shows that
    \begin{equation}\label{eq:ptr_upperbound}
        \eta_1(\mathcal{E}_M^{\otimes n},\norm{\cdot}_1,\mu_{2^{nN}}\times\mu_{2^{nN}})
        \leq
        \frac{1}{2}
        \Tr (
            \frac{2}{d^n+1}\qty[d^n(\Tr_2 \tau^2)^{\otimes n} - (\pi^2)^{\otimes n}]
        )^{1/2} + \order{\sqrt{\frac{n\ln d}{d^n}}}.
    \end{equation}
    In this case, the image of the maximally mixed state is
    \begin{equation}
        \pi = \mathcal{E}_M\qty(\frac{\mathbb{I}}{d}) = \frac{\mathbb{I}_{N-M}}{d'}.
    \end{equation}
    and the Choi state
    \begin{equation}
        \tau = (\mathcal{E}_M \otimes \operatorname{id}) \qty(\ketbra{\Omega}) = \ketbra{\Omega_{N-M}} \otimes \frac{\mathbb{I}_{M}}{(d/d')},
    \end{equation}
    where $\ketbra{\Omega_{N-M}}$ denotes the maximally entangled state between the last $N-M$ qubits of the main system and the dilated system and $\mathbb{I}_{M} / (d/d')$ is the maximally mixed state in the first $M$ qubits of the dilated system, recalling that $\Tr_2$ denotes discarding the dilation system
    \begin{equation}
        \Tr_2 \tau^2 =\Tr_2\qty[ \ketbra{\Omega_{N-M}} \otimes \frac{\mathbb{I}_{M}}{(d/d')^2}] = \frac{1}{(d/d')} \frac{\mathbb{I}_{N-M}}{d'}.
    \end{equation}
    Introducing this in Eq.~\eqref{eq:ptr_upperbound},
    \begin{align}
        \eta_1(\mathcal{E}_M^{\otimes n},\norm{\cdot}_1,\mu_{2^{nN}}\times\mu_{2^{nN}})
        &\leq
        \frac{1}{2}
        \Tr (
            \frac{2}{d^{n}+1}\qty[d^n \frac{1}{(d/d')^n} \frac{\mathbb{I}_{N-M}^{\otimes n}}{d'^n} - \frac{\mathbb{I}_{N-M}^{\otimes n}}{d'^{2n}}]
        )^{1/2} \\
        &=
        \frac{1}{2} d'^n
        \qty(
            \frac{2}{d^{n}+1}\qty[d^n \frac{1}{(d/d')^n} \frac{1}{d'^n} - \frac{1}{d'^{2n}}]
        )^{1/2} \\
        &=
        \frac{1}{2} d'^n
        \qty(
            \frac{2}{d^{n}+1}\qty[1 - \frac{1}{d'^{2n}}]
        )^{1/2} \\
        &=
        \frac{1}{2}
        \qty(
            \frac{2 }{d^{n}+1}\qty[d'^{2n} - 1]
        )^{1/2} \\
        & \leq 
        \frac{1}{2}
        \qty(
            \frac{2 d'^{n}}{d^{n}}
        )^{1/2} = \frac{1}{\sqrt{2}} \qty(\frac{d'}{\sqrt{d}})^n.
    \end{align}
    Hence, if $d' < \sqrt{d}$ the numerator is exponentially suppressed. In terms of discarded qubits $2^{N-M} < 2^{N/2}$ which is equivalent to $N/2 < M$.

    For the case $N=M$, note that the average distance for output states is equal to the average distance for mixed states over the Hilbert-Schmidt distribution $\nu_{d^{n/2}}$, where $d^{n/2}=2^{n(N-M)}=2^{nN/2}$ corresponds to the dimension of the remaining space. Precisely, let $\ket{\psi}\in\cH_{d^n}\cong\cH_{d^{n/2}}\otimes\cH_{d^{n/2}}$ be Haar distributed $\ket{\psi}\sim\mu_{d^n}$ and consider the state resulting from the partial trace $\rho=\Tr_2\ketbra{\psi} \in \cS(\cH_{d^{n/2}})$, then the induced probability distribution is the Hilbert-Schmidt distribution $\rho\sim\nu_{d^{n/2}}$~\cite{Puchala16}. The expected value for this distance in the asymptotic limit is known, for instance, as implied by Prop. 2 in Ref.~\cite{Puchala16}:
    \begin{equation}\label{eq:av_trdist_hs}
        \lim_{n\to\infty} \bEx{\rho,\sigma\sim\nu_{d^{n/2}}}\qty[\frac{1}{2}\norm{\rho-\sigma}_1] = \frac{1}{4}+\frac{1}{\pi},
    \end{equation}
    which allows us to find tight upper and lower bounds for the average contraction. For the lower bound,
    \begin{align}
        \lim_{n\to\infty}
        \eta_1(\mathcal{E}_M^{\otimes n},\norm{\cdot}_1,\mu_{2^{nN}}\times\mu_{2^{nN}})        
        &\geq
        \lim_{n\to\infty} \bEx{\rho,\sigma\sim \mu_{d^{n}}}\qty[\frac{1}{2} \norm{\mathcal{E}_M^{\otimes n} (\rho) - \mathcal{E}_M^{\otimes n}(\sigma)}_1]\\
        &=\lim_{n\to\infty} \bEx{\rho,\sigma\sim \nu_{d^{n/2}}}\qty[\frac{1}{2} \norm{\rho - \sigma}_1]\\
        &=\frac{1}{4}+\frac{1}{\pi},
    \end{align}
    where in the first inequality we used that $\norm{\rho-\sigma}_1\leq 2$ for any two states, the second is implied by our first observation and the last one by Eq.~\eqref{eq:av_trdist_hs}. The upper bound can be obtained as in Thm.~\ref{thm:avc_from_operator_root_bound}, relying on a combination of Lemma~\ref{lemma:tailbound2ratiobound}, which bounds the expected value of a fraction when the denominator satisfies a tail bound, and Lemma~\ref{lemma:concentrationoftrdist}, which claims that the trace distance satisfies the required concentration. It leads to
    \begin{align}
        \lim_{n\to\infty}
        \eta_1(\mathcal{E}_M^{\otimes n},\norm{\cdot}_1,\mu_{2^{nN}}\times\mu_{2^{nN}})
        &\leq
        \lim_{n\to\infty}\qty[
        \frac{1}{2}  \bEx{\rho,\sigma\sim \mu_{d^{n}}}\qty[\norm{\mathcal{E}_M^{\otimes n} (\rho) - \mathcal{E}_M^{\otimes n}(\sigma)}_1]
        + \frac{1}{d^n} + 2\sqrt{\frac{1}{d^n}}
        ] \\
        &=\lim_{n\to\infty} \bEx{\rho,\sigma\sim \nu_{d^{n/2}}}\qty[\frac{1}{2} \norm{\rho - \sigma}_1]\\
        &=\frac{1}{4}+\frac{1}{\pi},
    \end{align}
    where the first bound follows from the two lemmas referred above and $\frac{1}{2}\norm{\mathcal{E}_M^{\otimes n} (\rho) - \mathcal{E}_M^{\otimes n}(\sigma)}_1\leq 1$, the second one includes our first observation for the distribution together with neglecting $\order{d^{-n/2}}$ terms, and the last one is Eq.~\eqref{eq:av_trdist_hs}.
\end{proof}

The intuition behind this result is rooted in the typical properties of Haar-random states. These states are generally highly entangled, meaning that discarding more than half of the system leaves the remaining part in a highly mixed state almost surely. However, the dimension of the smallest subsystem limits the entanglement between the subsystems. Specifically, a system of dimension $(d/d')$ can only be entangled with a subspace of equal dimension. Thus, if $(d/d') < d'$, there exists a subspace of dimension approximately $d'^2/d$ that remains unentangled with the discarded subsystem. Alternatively, the $M$ discarded qubits can only be entangled with at most $M$ qubits, leaving the remaining $N - 2M$ qubits unaffected.

This example demonstrates that, for sufficiently low noise levels, the average contraction of a quantum channel can approach unity. The following result extends this observation to a more general setting.

\begin{theorem} \label{thm:avc_lower_bound}
    Let $T: \cS(\cH_d) \to \cS(\cH_{d'})$ be a quantum channel with Choi state $\tau$. For any $\varepsilon \in (0,1)$ and $\delta > 0$, there exists a sufficiently large $n$ such that for any continuous 1-design distribution $\mu$ on pure states $\cH_{d^n}$,
    \begin{equation}
        \eta_1(T^{\otimes n},\norm{\cdot}_1,\mu \times\mu)
        \geq 1 - \varepsilon - \qty( 2^{S(\tau) + \delta} \norm{T \qty(\frac{\mathbb{I}}{d})}_{\infty} )^n,
    \end{equation}
    where $S(\tau)=-\Tr[\tau\log\tau]$ is the von Neumann entropy of the Choi state.
\end{theorem}

\begin{proof}
    As above, the key idea is to construct a projector that we call $Q_{\psi}$ that has a high overlap with the first state and a low overlap with the second one and employ
     \begin{equation}
        \bEx{\psi, \phi \sim \mu} \qty[ \frac{\norm{T^{\otimes n}(\ketbra{\psi}) - T^{\otimes n}(\ketbra{\phi})}_1}{\norm{\ketbra{\psi} - \ketbra{\phi}}_1} ]
        \geq
        \bEx{\psi, \phi \sim \mu} \qty[ \Tr[Q_{\psi} (T^{\otimes n}(\ketbra{\psi}) - T^{\otimes n}(\ketbra{\phi})) ]].
    \end{equation}
    
    Consider the canonical Kraus decomposition
    \begin{equation}
        T(\rho)=\sum_{x\in\cX} E_x \rho E_x^{\dag}, \quad \text{with} \quad \Tr E_x^{\dag} E_y = \delta_{xy} \Tr E_x^{\dag} E_x,
    \end{equation}
    with the minimum number of terms, and define the random variable $X\in \cX$ with probability distribution $p(x)=\frac{1}{d}\Tr E_x^{\dag} E_x$. Then, define the $\delta$-typical set of sequences $x_1\cdots x_n \in \cX^n$ of i.i.d. random variables to be:
    \begin{equation}
        T_{\delta}^{X^n}=\qty{x^n\in \cX^n \mid \abs{-\frac{1}{n}\log p(x^n) - H(X)} < \delta},
    \end{equation}
    where $p(x^n)=p(x_1)\cdots p(x_n)$ and $H(X)=-\sum_{x\in\cX} p(x) \log p(x)$ is the entropy of the distribution, note that for this choice of the Kraus operators, this is equal to the entropy of the Choi state, $H(X)=S(\tau)=-\Tr[\tau \log \tau]$.

    $T_{\delta}^{X^n}$ is a widely studied object in elementary Shannon information theory \cite{Wilde2017}, and satisfies that for any $\varepsilon\in (0,1)$, $\delta > 0$ and there is a large enough $n$ such that the following properties hold:
    \begin{align}
        \text{Unit probability:} \quad & \quad \Pr[x \in T_{\delta}^n] > 1- \varepsilon , \\
        \text{Cardinality:} \quad & \quad (1-\varepsilon) 2^{n(H(X)-\delta)} \leq \abs{T_{\delta}^n} \leq  2^{n(H(X)+\delta)}, \\
        \text{Equipartition:} \quad & \quad 2^{-n(H(X)+\delta)} \leq p(x) \leq  2^{-n(H(X)-\delta)} \quad \text{for} \quad x\in T_{\delta}^n.
    \end{align}
    
    For a given state $\ket{\psi} \in (\mathbb{C}^d)^{\otimes n}$ define the operator $Q_{\psi}$ as the orthogonal projector onto the support of 
    \begin{equation}
        \sum_{x^n\in T_{\delta}^{X^n}} E_{x^n} \ketbra{\psi} E_{x^n}^{\dag},
    \end{equation}
    where we defined $E_{x^n}=E_{x_1}\otimes \cdots \otimes E_{x_n}$. This operator plays the role of a projector highly overlapping $T^{\otimes n}(\ketbra{\psi})$. Indeed,
    \begin{align}
        \bEx{\psi\sim\mu} \Tr[ Q_{\psi} T^{\otimes n}(\ketbra{\psi})]
        &=
        \sum_{x^n\in\cX^n}
        \bEx{\psi\sim\mu} \Tr[Q_{\psi} E_{x^n} \ketbra{\psi} E_{x^n}^{\dag} ] \\
        &\geq 
        \sum_{x^n\in T_{\delta}^{X^n}}
        \bEx{\psi\sim\mu} \Tr[E_{x^n} \ketbra{\psi} E_{x^n}^{\dag} ] \\
        &=
        \sum_{x^n\in T_{\delta}^{X^n}}
        \frac{1}{d^n} \Tr[E_{x^n} E_{x^n}^{\dag} ] = 
        \sum_{x^n\in T_{\delta}^{X^n}}
        p(x^n) \\
        &\geq 1-\varepsilon ,
    \end{align}
    where the first inequality comes from discarding positive terms corresponding to $x\notin T_{\delta}^{X^n}$ that may overlap $Q_{\psi}$, the next step uses that $\mu$ is a 1-design, and the last one follows from the unit probability of the $\delta$-typical set for sufficiently large $n$.

    Conversely, the operator $Q_{\psi}$ has a small overlap with the second state if the entropy is low enough,
    \begin{align}
        \bEx{\psi,\phi\sim\mu} \Tr[ Q_{\psi} T^{\otimes n}(\ketbra{\phi})]
        &=
        \bEx{\psi \sim\mu} \Tr[Q_{\psi}  \qty(T\qty(\frac{\mathbb{I}}{d}))^{\otimes n}] \\
        &\leq 
        \bEx{\psi \sim\mu} \Tr[Q_{\psi}] \cdot \norm{\qty(T\qty(\frac{\mathbb{I}}{d}))^{\otimes n}] }_{\infty} \\
        &\leq \qty(
        2^{H(X)+\delta}\cdot
        \norm{T\qty(\frac{\mathbb{I}}{d})}_{\infty}
        )^n,
    \end{align}
    where we used that $\Tr Q_{\psi} = \rank(Q_{\psi}) \leq \abs{T_{\delta}^{X^n}}\leq 2^{n(H(X)+\delta)}$.

    We arrive at the result by combining both bounds and introducing $H(X)=S(\tau)$.
\end{proof}

\begin{remark}
    A direct consequence of the lower bound is that for a channel satisfying $S(\tau) < \log \qty(\norm{T\qty(\frac{\mathbb{I}}{d})}_{\infty}^{-1})$, the average contraction in the asymptotic limit approaches unity.
\end{remark}

\begin{remark}
    The result can be applied to discrete distributions conditioning the expectation value on $\rho\neq\sigma$ and including a correction that depends on the collision probability. The deviation is often exponentially small in $n$ due to the large number of states where $\mu$ is supported. This is the case for arbitrary product distributions and for 2-design distributions (see Sec.~\ref{sec:upperbound_2designs}).
\end{remark}

\begin{remark}
    Note that in the proof, we can drop the condition $\Tr E_x^{\dag} E_y\propto \delta_{xy}$ and instead work with the probability distribution $p(x)= \Tr E_x^{\dag} E_x / d$, which satisfies $p(x)\geq 0$ and $\sum_x p(x)=1$. With this choice, the relevant entropy for the bound is $H(X)=-\sum_x p(x)\log p(x)$, rather than the entropy of the Choi state $S(\tau)$. Often, this approach is more convenient than determining the canonical Kraus operators, but Thm.~\ref{thm:avc_lower_bound} establishes the sharp bound as $S(\tau)\leq H(X)$ (see Fact \ref{fact:entropy_bound}).
\end{remark}

Theorem~\ref{thm:avc_lower_bound} extends the lower bound for the low-noise regime given in Prop.~\ref{prop:avc_partial_trace_asymp_lim} for the partial trace. The proof's physical intuition can be understood in terms of the paths generated by different products of Kraus operators applied to random states. The result shows that if the entropy of the Choi state is sufficiently low, and state 1-designs generate these paths, then their overlap is typically vanishingly small. Consequently, the output states preserve their distinguishability with high probability. Notably, for a qubit unital channel, this condition simplifies to $S(\tau)<1$, which is not particularly restrictive given that the entropy of the Choi state for a single-qubit channel lies in the range $0\leq S(\tau) \leq 2$.

In product channels, it is not surprising that certain 1-design distributions preserve distinguishability on average. For example, consider a 1-design distribution supported on i.i.d. product states. It is then not too difficult to see that the average trace distance between outputs will converge to $1$ by resorting to the quantum Stein lemma~\cite{Hiai1991}. The key insight here is that the lower bound holds for any 1-design distribution, including those concentrated on highly entangled states. Conversely, the fact that the upper bound in Thm.~\ref{thm:avc_from_operator_root_bound} requires at least a 2-design distribution suggests that concentration on highly entangled states is fundamental to achieving upper bounds implying asymptotically vanishing average contraction.

A natural question is whether the intuitive argument about the overlap of typical paths can be reversed. Specifically, could we find an upper bound as a function of the entropy of the Choi state, implying that the average distinguishability vanishes? One might assume that if the paths do not have vanishing overlap, they should mix, spoiling distinguishability. However, a simple counterexample reveals gaps in this intuition. Consider a dephasing channel on an $n$-qubit system and a uniform distribution over the computational basis. Although this distribution is a 1-design and the entropy of the Choi state can be as large as $n$ (see Eq.~\eqref{eq:choi_rep_deph}), these states are invariant under dephasing, so they do not suffer any contraction. 

In the context of product channels, the best upper bound complementing the lower bound in Thm.~\ref{thm:avc_lower_bound} is given by Cor.~\ref{cor:travcontract_nproduct_2designs}. In the case of the partial trace example from Prop.~\ref{prop:avc_partial_trace_asymp_lim}, these bounds are nearly complementary. However, they do not generally yield the same threshold. To illustrate the emerging picture of the asymptotic behavior of average contraction over $t$-designs, we consider the $n$-qubit local depolarizing channel.

\begin{proposition}[Average contraction of local depolarizing in the asymptotic limit]\label{prop:avc_local_depol}
    Let $T_{\mathrm{depol}}(\rho) = (1-p) \; \rho + p \; \frac{\mathbb{I}}{2}$ denote the single-qubit depolarizing channel. Let $\mu$ be a 2-design distribution over pure states on $n$ qubits. Then, the asymptotic average contraction satisfies: 
    \begin{align}
        \lim_{n \to \infty}
        \eta_1(T_{\mathrm{depol}}^{\otimes n},\norm{\cdot}_1,\mu \times\mu)
        &= 
        \begin{cases}
            1 & \text{if} \quad p < p_1 \approx 0.25, \\
            0 & \text{if} \quad p > p_2 \approx 0.42,
        \end{cases}
    \end{align}
    where $p_1$ is the solution to the equation $\left(\left(1-\frac{3}{4}p\right)\log\left(1-\frac{3}{4}p\right)+\frac{3p}{4}\log\left(\frac{p}{4}\right)\right) = -1$ and $p_2=1-\frac{1}{\sqrt{3}}$.
\end{proposition}

\begin{proof}
    For the small $p$ regime, we invoke Thm.~\ref{thm:avc_lower_bound}. The Choi state is
    \begin{equation}
        \tau_{\mathrm{depol}} = (1-p)\ketbra{\Omega}+p\frac{\mathbb{I}^{\otimes 2}}{4},
    \end{equation}
    and its entropy
    \begin{equation}
        S(\tau_{\mathrm{depol}}) = -\qty[(1-p+\frac{p}{4})\log(1-p+\frac{p}{4})+ 3\cdot \frac{p}{4}\log\frac{p}{4}].
    \end{equation}
    As the channel is unital, the condition for the average contraction of the product channel to converge to unity reads $S(\tau_{\mathrm{depol}})< 1$, which holds for $p<p_1$.

    For the large $p$ regime, we can resort to Cor.~\ref{cor:travcontract_nproduct_2designs} introducing
    \begin{align}
        \Tr_2 \tau^2_{\mathrm{depol}} &= \Tr_2\qty(
            (1-p)^2\ketbra{\Omega} + p^2\frac{\mathbb{I}^{\otimes 2}}{16}
            +2\frac{(1-p)p}{4} \ketbra{\Omega}
        ) \\
        &=   (1-p)^2 \frac{\mathbb{I}}{2}
             + p^2\frac{\mathbb{I}}{8}
            +2\frac{(1-p)p}{4} \frac{\mathbb{I}}{2}.
    \end{align}
    Discarding $\order{2^{-n/2}}$ terms, the upper bound takes the simplified form
    \begin{equation}
        \frac{1}{\sqrt{2}} \qty(\Tr
            \sqrt{(2 (1-p)^2 + \frac{p^2}{2} + (1-p)p) \frac{\mathbb{I}}{4}}
            )^n
            =
        \frac{1}{\sqrt{2}} \qty(2 (1-p)^2 + \frac{p^2}{2} + (1-p)p)^{n/2},
    \end{equation}
    the polynomial in the parentheses can be expressed as $(1+3(1-p)^2) / 2$, which is strictly smaller than one if $\abs{1-p}< 1/\sqrt{3}$.
\end{proof}

Thus, by combining the results of Sec.~\ref{sec:upper_bounds_tr} and Sec.~\ref{sec:lower_bounds_trace}, we conclude that for certain local noise strengths, the average distinguishability either goes to $1$ exponentially fast or $0$, i.e. it is maximal or minimal. The only example where we have a complete picture, namely the one where we discard certain subsystems in Prop.~\ref{prop:avc_partial_trace_asymp_lim} suggests there are intermediary regimes where we converge to a constant, but only at a critical point.
It would be interesting to investigate to what extent there are other possible intermediary regimes, i.e. open sets where we converge to a constant or the convergence can be polynomially fast.

The phase transitions observed here bear a superficial resemblance to effects studied in the context of quantum Darwinism \cite{BlumeWojciech06}, which describes the emergence of classicality in a quantum system coupled to a multipartite environment. In this framework, classicality arises from the redundancy of information about the system encoded in fragments of the environment, making it objectively accessible to multiple observers. For models involving random states, a transition in accessible information occurs when environment fragments reach about half the system size, the same threshold we observe in Prop.~\ref{prop:avc_partial_trace_asymp_lim}.

\subsection{Average contraction of random circuits under unital noise}\label{sec:quantum_circuits}
Another class of quantum channels that is of interest in quantum information and computation that goes beyond the models considered so far is that of noisy quantum circuits. 
In a growing body of literature~\cite{StilckFrana2021,Quek2024,Gonzalez-Garcia2022,Preskill2018}, several authors try to answer the question of to what extent quantum circuits are susceptible to noise and/or capable of performing interesting quantum computations even in the absence of fault-tolerance.

We will now leverage the tools developed in the last sections to establish that for shallow enough (constant) depths, noisy quantum circuits have an average contraction coefficient of the trace distance that converges to $1$ in the limit of many qubits under local unital noise of low enough, but constant strength. Interestingly, this happens for \emph{any} ensemble of random unitaries, even arbitrarily non-local ones, as long as they form 1-designs. The same holds for the ensemble of input states: as long as they form a 1-design, we observe this phenomenon. However, at this stage we will only state results for unital noise channels and leave generalizations to non-unital models to future work.

More precisely, our models consider the case where we have an $n$-qubit circuit $\cC$ implemented as a sequence of unitaries $U_i$ such that $U_{\cC}=U_D\cdots U_1$, and an $n$-qubit quantum channel $\Phi$ with Kraus decomposition $\qty{E_x}_{x\in\cX}$. Assume $\cC'$, the noisy implementation of $\cC$ is given by intercalating $\Phi$ in each unitary layer resulting in the quantum channel
\begin{equation}\label{equ:circuit_model}
    \cE_{\cC'} = \Phi \circ \cU_D \circ \Phi \circ \cdots \circ \Phi \circ \cU_2 \circ \Phi \circ \cU_1,
\end{equation}
where $\cU(\cdot) = U \cdot U^{\dag}$ denotes unitary conjugation. Let $\ket{\psi}$ and $\ket{\phi}$ be two independent random states drawn from a 1-design and let the unitary layers be drawn from a 1-design. Note again that we do not make any assumptions on the locality of the unitaries or the states. Let us assume the channel is unital, $\Phi(\mathbb{I})=\mathbb{I}$, which also implies $\cE_{\cC'}(\mathbb{I})=\mathbb{I}$. We then have: 

\begin{theorem}[Average contraction for random circuits]\label{thm:avgcirc}
Let $\cC'$ be a noisy random unitary as in Eq.~\eqref{equ:circuit_model}, where $\phi,\psi$ and $U_i$ are independently drawn from 1-designs and $\Phi$ is unital.
For any $\varepsilon\in (0, 1)$ and $\delta > 0$, there is a large enough $D$ (but independent of $n$) such that
\begin{equation}
    \bE_{\psi,\phi, \cC}\qty[
    \frac{\norm{\cE_{\cC'}(\ketbra{\psi}) - \cE_{\cC'}(\ketbra{\phi})}_1}{\norm{\ketbra{\psi} - \ketbra{\phi}}_1}
    ]
    \geq 1 - \varepsilon - 
      2^{D (S(\tau)+\delta)-n}
    ,
\end{equation}
where $S(\tau)$ is the entropy of the Choi state of $\Phi$.    
\end{theorem}
\begin{proof}
The proof of this result is analogous to Thm.~\ref{thm:avc_lower_bound}, so we leave it to App.~\ref{sec:random_circuits}. 
\end{proof}
Note that in the case where $\Phi=T^{\otimes n}$, then $S(\tau)=nS(\tau_1)$, where $\tau_1$ is the Choi matrix of $T$. Thus, we see that for $D<1/S(\tau_1)$, the average contraction will converge to $1$ regardless of the underlying ensemble of initial random states or unitaries.

The result above complements an interesting picture that emerges about the contractive properties of random circuits. To more easily compare our results to others in the literature, we will now focus on the case of local depolarizing noise with strength $p$. In~\cite{Deshpande2021TightBO}, the authors show that at depth $D$, for noisy random brickwork circuits with local $2$-designs, we have that the average trace distance to the maximally mixed state is at least $\Omega(p^{cD})$ for some constant $c$. Thm.~\ref{thm:avgcirc} implies that this result is not tight in the $D=\mathcal{O}(1)$ regime for certain noise levels, as the trace distance actually converges to $1$ as we increase the number of qubits. Interestingly, the proof of~\cite{Deshpande2021TightBO} depends on having product initial states, whereas our conclusions even hold with arbitrarily entangled pure input states. Our assumptions on the unitaries are also significantly more general, as we just assume $1$ designs.
It is also interesting to contrast this with the results of~\cite{Quek2024}, where the authors construct a family of quantum circuits that are a 1-design at each layer s.t. at depth $\mathcal{O}(p^{-1}\textrm{poly}\log\log(n))$ the trace distance between outputs becomes superpolynomially small on average. Put together with Thm.~\ref{thm:avgcirc}, these results paint an interesting picture of average contraction of noisy, random quantum circuits: for depths $D<1/S(\tau_1)$, the trace distance will converge to a maximal value and typical noisy outputs will be perfectly distinguishable. This behavior is particularly striking if we consider the regime where $p$ depends on the system size, as e.g. for $S(\tau_1)^{-1}= \Omega(\log(n))$ and all-to-all connectivity, our result still holds for $D=\Omega(\log(n))$ and that the lightcone of each qubit is the whole system. That is, even though each qubit will have undergone an expected number of errors $n/\log(n)$, the contraction coefficient will still converge to $1$.
And for depths a whisker away from constant (i.e. $p^{-1}\textrm{poly}\log \log(n)$), it can already be superpolynomially small for certain ensembles.

One interesting consequence of these observations is that current techniques to show limitations on error mitigation protocols~\cite{Quek2024,Takagi2022,Takagi2023} by showing exponential lower bounds on the sample complexity do not extend to the regime where $D=\mathcal{O}(1)$ for sufficiently low $p$, even if we allow for non-local unitaries. This is because, to the best of our knowledge, all known results rely on upper-bounding the distinguishability of (potentially average) quantum states and reducing error mitigation to distinguishing various outputs of noisy quantum circuits. However, our bounds show that, at least on average, in the regime $D=\mathcal{O}(1)$, we can only hope for constant lower bounds by resorting to the majority of random ensembles of quantum states considered in the literature so far.

\section{Bounds between moments of contraction}
\label{sec:bound_for_dfs}

This section will discuss relations between the moments of contraction of the trace distance and a family of divergences known as quantum $f$-divergences by integral formulation~\cite{Hirche23}. These are a quantum generalization of Csisz\'ar $f$-divergences \cite{Csiszar1963, AliSilvey1966}, which for a convex function $f:(0,+\infty)\to \mathbb{R}$ with $f(1)=0$ and two probability distributions $P$ and $Q$ over a finite set $\cX$ is given by 
\begin{equation}\label{eq:df_classical}
	D_f(P \| Q) = \sum_{x\in \cX} Q(x) f\qty(\frac{P(x)}{Q(x)}).
\end{equation}
The definition can be directly evaluated when the distributions are positive $P(x),Q(x) > 0$ for any $x\in\cX$, to extend it for any two distributions, we take the convention that 
\begin{equation}
	f(0)=\lim_{u\to 0} f(u), \quad
	0 \cdot f\qty(\frac{0}{0}) = 0,\quad
	\text{and} \quad
	0 \cdot f\qty(\frac{a}{0}) = a \lim_{u\to \infty} \frac{f(u)}{u},
\end{equation}
for $0< a< +\infty$. We set that $D_f(P \| Q)=+\infty$ whenever we come across $f(0)$ or $0\cdot f(a/0)$ and their value is equal to $+\infty$. This family of divergences can be seen as a generalization of relative entropy, which retains some of the most relevant properties such as DPI \cite{CsizarShields2004}.  As special cases, we find the Kullback-Leibler divergence, $\chi^2$-distance, the Hellinger distance, and the total variation distance, which have wide applicability in information theory and statistics. Recently, it was shown that $f$-divergences admit an integral representation when $f$ is twice differentiable in terms of the $E_\gamma$ divergence (see Prop. 3 in \cite{Sason2016}), which is a special case of $f$-divergence defined by $E_\gamma (P \| Q) = \sum_{x\in\cX} \max\{ P(x) - \gamma Q(x), 0\}$. This representation enables to characterize other $f$-divergences by first establishing properties of the $E_\gamma$ divergence, which is often easier to analyze.

Originating from Petz's original work on quantum $f$-divergences \cite{Petz1985, Petz1986}, numerous quantum generalizations were defined, as discussed in~\cite{Hiai2010, Sharma2012, Matsumoto2018, Wilde2018}. Our interest is in relating the moments of contraction for $f$-divergences with those for the trace distance, extending the bounds found in the preceding sections. Some quantum generalizations have been demonstrated to satisfy inequalities on maximal contraction coefficients analogous to their classical counterparts \cite{PetzRuskai1998, LesniewskiRuskai1999}. More recently, the integral representation of the quantum relative entropy has been shown to imply data processing inequality \cite{Frenkel2023}. This result motivated a definition of quantum $f$-divergences via integral representation and proved useful in deriving properties related to data processing inequalities. This is the reason to consider this quantum generalization in this section:

\begin{definition}[Definition 2.4. in \cite{Hirche23}]\label{def:fdivergences}
Let $f:(0,\infty)\rightarrow \mathbb{R}$ be a twice differentiable convex function  with $f(1)=0$. Then, define
\begin{equation}\label{eq:Df_integral}
D_f(\rho\| \sigma) = \int_1^{\infty} f''(\gamma)E_{\gamma}(\rho\| \sigma)+\gamma^{-3}f''(\gamma^{-1})E_{\gamma}(\sigma\| \rho) \dd \gamma,
\end{equation}
whenever the integral is finite and $D_f(\rho\| \sigma)=+\infty$ otherwise. Here, $E_{\gamma}(\rho\| \sigma)=\Tr (\rho -\gamma \sigma)_+$ denotes the quantum hockey-stick divergence.
\end{definition}

Indeed, this integral form has been shown to guarantee DPI and joint convexity. Naturally, for commuting states it reduces to the classical integral representation, which justifies using the same notation for the classical and quantum $f$-divergences in this context. Remarkably, it is known that for any channel $T$, the contraction coefficient of a $f$-divergence $D_f$ lies between the one of the $\chi^2$-divergence, $f(x)=x^2-1$, and the trace distance~\cite{Hirche23}:
\begin{equation}\label{eq:maxcontract_relations}
    \eta_{\infty}\qty(T, D_{x^2}) \leq \eta_{\infty}\qty(T, D_f) \leq \eta_{\infty}\qty(T, \norm{\cdot}_1),
\end{equation}
where the lower bound saturates for any operator convex $f$, and thus, the maximal contraction coefficient is independent of $f$ in that case~\cite{Hirche23}. Similar results are known for other quantum generalizations~\cite{PetzRuskai1998, LesniewskiRuskai1999}. This section explores whether similar bounds hold for the average case, in particular, if the moments of contraction of the trace distance dominate the moments of $f$-divergences.

To guarantee that the denominator is finite, $D_f(\rho \| \sigma) < \infty$, we will focus on distributions that keep the second state fixed to the maximally mixed states, which for short we will denote $\sigma_*=\mathbb{I}/d$, and functions with a finite limiting value at the origin, $f(0)<\infty$. Additionally, we focus on distributions $\nu$ concentrated on pure states. As mentioned before, distributions of this kind have a natural application, as $\nu$ can be thought of as a distribution of operationally relevant states and taking $\sigma_*$ as a reference state we quantify how close we are to losing the information of the input state. Conveniently, $f$-divergences are unitarily invariant, so we obtain the same value for any pure state, and as the states commute we can compute the $f$-divergence by consistency with the classical:
\begin{equation}\label{eq:df_pure_vs_maxmix}
	D_f(\ketbra{\psi} \| \sigma_* ) = \frac{1}{d} f(d) + \qty(1 - \frac{1}{d}) f(0).
\end{equation}
We have a similar situation for the trace distance
\begin{equation}\label{eq:tr_pure_vs_maxmix}
	\frac{1}{2}\norm{\ketbra{\psi} - \sigma_*}_1 = 1 - \frac{1}{d},
\end{equation}
Therefore, we are considering distributions where the denominator of the contraction ratio is fixed. For the denominator, we can rely on Pinsker-type inequalities shown below to relate $f$-divergences and the trace distance. This is the main reason to use the integral formulation, as to our knowledge this kind of inequalities have not been established for other quantum $f$-divergences to the same level of generality. It is worth noting that proving these inequalities would be sufficient to apply the results in this section to other quantum divergences.

Pinsker's inequality generally refers to the bound
\begin{equation}
	2 \qty(\frac{1}{2}\norm{\rho - \sigma}_1)^2 \leq D_{\text{RE}}(\rho \| \sigma),
\end{equation}
where $D_{\text{RE}}(\rho \| \sigma)$ denotes the relative entropy. Similar expressions are well known for classical $f$-divergences, and reversed inequalities, upper bounding $D_f$ with an expression depending on the trace distance, have also been established~\cite{Fedotov03, Gilardoni10}. In App.~\ref{sec:pinsker_ineqs} we see that a result holds for quantum $f$-divergences that are consistent with the classical definition; however, to the best of our knowledge, reverse Pinsker-like inequalities have only been extended to the quantum $f$-divergences considered in this section~\cite{Hirche23}. Fixing the second state to the maximally mixed, this reverse direction is necessary and sufficient to prove that a vanishing contraction of the trace distance implies a vanishing contraction for any $f$-divergence. Because the trace distance is always finite, but $f$-divergences can diverge, it is clear that the reversed Pinsker-type inequalities should generally depend on the states. Fortunately, this dependence can be expressed universally in terms of the max relative entropy, defined as
\begin{equation}
    D_{\max}(\sigma\|\rho)=\inf\qty{\lambda\,:\,\rho\leq e^{\lambda}\sigma},
\end{equation}
where we set it to $+ \infty$ when the infimum does not exist, and its symmetrized version, the Thompson distance:
\begin{equation}
    \Xi (\rho, \sigma) = \max\qty{D_{\max}(\rho\|\sigma),D_{\max}(\sigma\|\rho)}.
\end{equation}
We gather the direct and reverse Pinsker-type inequalities in the following two propositions.

\begin{proposition}[Reverse Pinsker-type inequality, Prop. 5.2. in \cite{Hirche23}]\label{prop:reverse_pinsker}
    Let $f:(0,\infty)\rightarrow \mathbb{R}$ be a twice differentiable convex function with $f(1)=0$ and $D_f$ be its associated $f$-divergence, we have
    \begin{equation} \label{eq:pinsker_like}
            \begin{split}
            D_f(\rho \| \sigma) &\leq 
            \qty(\frac{f(e^{D_{\max}(\rho\|\sigma)})}{e^{D_{\max}(\rho\|\sigma)}-1} + \frac{e^{D_{\max}(\sigma\|\rho)} f(e^{-D_{\max}(\sigma\|\rho)})}{e^{D_{\max}(\sigma\|\rho)}-1})
             \, \frac{1}{2}\norm{\rho - \sigma}_1 \\
            &\leq 
            K_f\qty(e^{\Xi(\rho\|\sigma)}) \, \frac{1}{2}\norm{\rho - \sigma}_1,
        \quad
        \text{with}
        \quad
        K_f (x) = \frac{f \qty(x) + x f \qty(x^{-1})}{x - 1}.
        \end{split}
    \end{equation}
\end{proposition}

\begin{proof}
    See App.~\ref{sec:pinsker_ineqs}.
\end{proof}

The following generalized Pinsker bound is a quantum generalization of the main result in \cite{Gilardoni10}. The result holds for any generalization of $f$-divergences that is consistent with the classical and satisfies DPI. However, note that it requires an extra condition on $f$, namely, that it is differentiable up to third order. 

\begin{proposition}[Pinsker-type inequality]\label{prop:pinsker}
    Let $f:(0,\infty)\rightarrow \mathbb{R}$ be a convex function differentiable up to order 3 at $u=1$ with $f''(1)>0$, and let $D_f$ be its associated $f$-divergence. Then, we have
    \begin{equation}
        D_f(\rho\|\sigma) \, \geq \frac{f''(1)}{2}\, \qty(\frac{1}{2}\norm{\rho-\sigma}_1)^2.
    \end{equation}
\end{proposition}

\begin{proof}
    See App.~\ref{sec:pinsker_ineqs}
\end{proof}

The following lemma is a straightforward consequence of the preceding discussion:

\begin{proposition}\label{prop:pure_to_mix_ratio_bound}
	Let $\rho$ be a pure state and assume $D_{\max}(T(\rho) \| T(\sigma_*)) \leq \varepsilon$. Let $f:(0,\infty)\to \mathbb{R}$ be a twice differentiable convex function with $f(1)=0$ and $f(0)$ finite. Then,
	\begin{equation}
		\frac{D_f( T(\rho) \| T(\sigma_*))}{ D_f( \rho \| \sigma_*)}
		\leq
		\frac{d}{e^\varepsilon - 1} \frac{f(e^\varepsilon)  + (e^\varepsilon -1) f(0)}{f(d) + (d-1) f(0)}
		\frac{\norm{T(\rho) - T(\sigma_*)}_1}{ \norm{\rho - \sigma_*}_1}.
	\end{equation}
	In particular,
	\begin{equation}
		\frac{D_f( T(\rho) \| T(\sigma_*))}{ D_f( \rho \| \sigma_*)}
		\leq
		\frac{d}{d - 1}
		\frac{\norm{T(\rho) - T(\sigma_*)}_1}{ \norm{\rho - \sigma_*}_1}.
	\end{equation}
	
	If $f$ is differentiable up to order 3 at $1$ with $f''(1)>0$,
    \begin{equation}
    	\qty(\frac{\norm{T(\rho) - T(\sigma_*)}_1}{ \norm{\rho - \sigma_*}_1})^2
    	\leq
    	\frac{2}{f''(1)} \frac{1}{(1-1/d)^2} \qty(\frac{f(d)}{d-1} + f(0))
    	\frac{D_f( T(\rho) \| T(\sigma_*))}{ D_f( \rho \| \sigma_*)}.
    \end{equation}
\end{proposition}

\begin{proof}
	We consider the first upper bound on Prop.~\ref{prop:reverse_pinsker} applied to $T(\rho)$ and $T(\sigma_*)$. As the coefficient in the bound is non-decreasing with $D_{\max}(T(\rho) \| T(\sigma_*))$ and $D_{\max}(T(\sigma_*) \| T(\rho))$, then introducing that $D_{\max}(T(\rho) \| T(\sigma_*)) \leq \varepsilon$,
    \begin{equation}
        \frac{f(e^{D_{\max}(T(\rho)\| T(\sigma_*))})}{e^{D_{\max}(T(\rho) \| T(\sigma_*))}-1} \leq \frac{f(e^\varepsilon)}{e^\varepsilon-1},
    \end{equation}
    whereas the second term can be bounded by its value in the limit $D_{\max}(T(\sigma_*) \| T(\rho) )\to\infty$,
    \begin{equation}
        \frac{e^{D_{\max}(T(\sigma_*) \| T(\rho))} f(e^{-D_{\max}(T(\sigma_*) \| T(\rho))})}{e^{D_{\max}(T(\sigma_*) \| T(\rho))}-1}
        \leq 
        \lim_{x\to \infty} \frac{x f(x^{-1})}{x-1} = \lim_{x\to \infty} f(x^{-1}) \equiv f(0).
    \end{equation}
    For the particular case we used that we always have $D_{\max}(T(\rho) \| T(\sigma_*)) \leq \ln d$.

    The second inequality is equivalent to Prop.~\ref{prop:pinsker} using the exact values for the denominators.
\end{proof}

As a direct consequence of the previous proposition, we have that for any distribution over pure states $\nu$ and any moment $p$,
\begin{equation}\label{eq:avc_bounds}
    \eta_p(T,D_f,\nu \times \delta_{\sigma_*})
    \leq \frac{d}{d - 1 }
    \,
    \eta_p(T,\norm{\cdot}_1,\nu \times \delta_{\sigma_*})
  	\leq
  	\,\sqrt{
   	\frac{2}{f''(1)} \qty(\frac{f(d)}{d-1} + f(0))
   	\eta_p(T,D_f,\nu \times \delta_{\sigma_*})}.
\end{equation}
In particular, we observe that the asymptotic behavior of the moments of contraction for the trace distance dominates over those of general $f$-divergences:

\begin{corollary}\label{cor:asympt_bounds_tr2Df_avc_wrt_max_mix}
    Let $\Phi_n:\cS(\cH_d^{\otimes n}) \to \cS(\cH_{d}^{\otimes n})$ be a family of quantum channels and let $\nu_n$ probability distributions over pure states, $\rho\sim\nu_n$ with $\rho\in\cS(\cH_d^{\otimes n})$. Let $f:(0,\infty)\to \mathbb{R}$ be a twice differentiable convex function with $f(1)=0$ and $f(0)$ finite.
    
    Then,
    \begin{equation}
        \eta_p(\Phi_n,D_f,\nu_n \times \delta_{\sigma_*^{\otimes n}})
        \leq
        \,
        \eta_p(\Phi_n,\norm{\cdot}_1,\nu_n \times \delta_{\sigma_*^{\otimes n}}) + \order{d^{-n}}.
    \end{equation}
    Additionally, if $f$ is differentiable up to order 3 at $1$ with $f''(1)>0$,
    \begin{equation}
    	\eta_p(\Phi_n, \norm{\cdot}_1,\nu_n \times \delta_{\sigma_*^{\otimes n}}))^2
    	\leq
    	\frac{2}{f''(1)} \qty(\frac{f(d^n)}{d^n} + f(0))
    	\eta_p(\Phi_n, D_f,\nu_n \times \delta_{\sigma_*^{\otimes n}})
    	+ \order{d^{-n}}.
    \end{equation}
\end{corollary}

The converse relation does not hold in general. That is, vanishing moments for the trace distance imply vanishing moments for any divergence; in contrast, the converse bound is conditioned by the limiting behavior of $f(x)/x$ as $x\to\infty$.

To show this aspect, let us consider some common instances of $f$. For the relative entropy $f(x)=x \ln x$ and $f(x)/x = \ln x$, 
\begin{equation}
   	\eta_p(\Phi_n, \norm{\cdot}_1,\nu_n \times \delta_{\sigma_*^{\otimes n}}))^2
   	\leq
   	2 n \ln d \;
   	\eta_p(\Phi_n, D_{x \ln x},\nu_n \times \delta_{\sigma_*^{\otimes n}})
  	+ \order{d^{-n}}.
\end{equation}
Thus, a vanishing moment of contraction for the relative entropy only suppresses that of the trace distance if it decreases stronger than inverse-linearly with system size, such that $n \; \eta_p(\Phi_n, D_{x \ln x},\nu_n \times \delta_{\sigma_*^{\otimes n}}) \to 0$. Hellinger divergences are given by $f(x)= \frac{x^{\alpha}-1}{\alpha-1}$ with $\alpha\in (0,1)\cup (1,\infty)$, we obtain that
\begin{equation}
    \eta_p(\Phi_n, \norm{\cdot}_1,\nu_n \times \delta_{\sigma_*^{\otimes n}}))^2
    \leq
    \frac{2}{\alpha (\alpha -1)} \qty(d^{(\alpha - 1) n} - 1)
    \eta_p(\Phi_n, D_{x^\alpha},\nu_n \times \delta_{\sigma_*^{\otimes n}})
    + \order{d^{-n}},
\end{equation}
for $\alpha > 1$ the bound only vanishes as $n\to\infty$ if the moment for the Hellinger distance is exponentially decreasing, strong enough to ensure $2^{(\alpha -1) n} \eta_1(\Phi_n, D_f,\nu_n \times \delta_{\sigma_*^{\otimes n}}) \to 0$. This includes the special case $\alpha = 2$, which corresponds to the $\chi^2$-divergence. We include a few additional observations on this case in App.~\ref{sec:supmat_chi2}. Whereas, for $\alpha < 1$ it takes the form
\begin{equation}
    \eta_p(\Phi_n, \norm{\cdot}_1,\nu_n \times \delta_{\sigma_*^{\otimes n}}))^2
    \leq
    \frac{2}{\alpha (1 - \alpha)}
    \eta_p(\Phi_n, D_{x^\alpha},\nu_n \times \delta_{\sigma_*^{\otimes n}})
    + \order{d^{-(1-\alpha)n}},
\end{equation}
which shows that an asymptotically vanishing moment of contraction for the Hellinger divergence implies asymptotically vanishing moment of contraction for the trace distance and, hence, for all $f$-divergence.

To summarize, we have shown that for distributions that concentrate on random pure states compared to the mixed state, averaging does not influence the bound. This insight provides a useful baseline for more refined analyses. Notably, averaging is potentially useful when one aims to tighten the bounds, for example, in scenarios where $\rho\sim \nu$ satisfies $\Pr[D_{\max}(T(\rho) \| T(\sigma_*)) < \varepsilon] \geq 1 - \delta$. In the next section, we explore quantum channels that exhibit such concentration properties. Additionally, several directions for further work emerge, including extending known classical $f$-divergence bounds (e.g., Thm.~1 in \cite{Sason2016}) to the quantum setting, or establishing Pinsker-type inequalities for alternative quantum generalizations to broaden the applicability of our results.

\section{Applications to Local Differential Privacy}\label{sec:ldp}

Differential privacy is a condition on stochastic processes that constrains the distinguishability of outputs between different, but similar, initial states, thereby preserving privacy. This concept finds wide applicability in safeguarding shared data, securing training data in machine learning, and providing resilience against adversarial examples~\cite{Dwork2014}. Similar privacy guarantees have been studied in the quantum domain \cite{Angrisani2023DPamplificationquantum, Angrisani2023QDP, Quek2021PrivateLearning, Du2022, Du2021, Hirche2023QuantumDP}. Here, we examine a stringent form of differential privacy, namely local differential privacy (LDP)~\cite{Angrisani2023QDP}, and derive bounds on average contraction.

Classically, a channel $\mathcal{M}$ is said to satisfy $(\varepsilon, \delta)$-local differential privacy if, for any pair of inputs $x$ and $x'$ and for any measurable subset $S$ of outputs, it holds that $\Pr[\mathcal{M}(x)\in S] \leq e^{\varepsilon} \Pr[\mathcal{M}(x')\in S] + \delta$. Beyond the local framework, this is only required for pairs $x,x'$ determined by some neighboring relation, which is a less stringent requirement~\cite{Zhou2017}. Roughly speaking, the parameter $\varepsilon$ quantifies the worst-case multiplicative increase in the likelihood of any output when changing the input, while $\delta$ allows for a small probability of greater deviation. Smaller values of $\varepsilon$ and $\delta$ correspond to stronger privacy guarantees, as they limit the information that can be inferred about any individual input. This classical notion motivates the quantum analogue introduced below.

\begin{definition}[Local Differential Privacy (LDP)]\label{def:ldp}
    A quantum channel $T: \cS(\cH_d)\to\cS(\cH_{d'})$ satisfies $(\varepsilon, \delta)$-LDP if for any positive operator $0\leq M \leq \mathbb{I}$ and any pair of states $\rho,\sigma\in\cS(\cH_d)$ it holds that
    \begin{equation}\label{eq:ldp_condition_measurements}
        \Tr[M T(\rho)]\leq e^{\varepsilon} \Tr[M T(\sigma)]+\delta.
    \end{equation}
\end{definition}

For simplicity, we restrict our discussion to the pure privacy scenario characterized by $\delta = 0$ \cite{Angrisani2023QDP}, and denote it by $\varepsilon$-LDP. While more general forms, such as $\delta\neq 0$ and non-local differential privacy conditions are important in practice, they introduce technical complexities without altering the fundamental insights of our discussion. Additionally, the core limitations of our results are still present in those scenarios, so we choose to keep the discussion elementary. However, extending it to these broader settings is conceptually straightforward, and we refer the reader to \cite{Hirche2023QuantumDP, Du2021, Angrisani2023QDP} for a thorough treatment of quantum differential privacy frameworks.

To build intuition for $\varepsilon$-LDP in quantum settings, we recall that it can also be expressed using the max relative entropy $D_{\max}$~\cite{Hirche2023QuantumDP}:

\begin{lemma}[From Lemma III.2. in \cite{Hirche23}]\label{lemma:ldp_to_dmax}
    A quantum channel $T: \cS(\cH_d)\to\cS(\cH_{d'})$ is $\varepsilon$-LDP if and only if $\sup_{\rho,\sigma} D_{\max}(T(\rho)\|T(\sigma))\leq \varepsilon$ for arbitrary input states $\rho$ and $\sigma$.
\end{lemma}

A direct consequence of this characterization is that for any input state $\rho$, the privacy constraint also applies when comparing $T(\rho)$ to the image of the maximally mixed state, yielding $D_{\max}(T(\rho)| T(\frac{\mathbb{I}}{d}))\leq \varepsilon$. In this case, we can use the tighter version of Prop.~\ref{prop:pure_to_mix_ratio_bound}, which makes $\varepsilon$-LDP a natural setting to apply the results from Sec.~\ref{sec:bound_for_dfs}.

Certainly, the LDP condition inherently imposes a closeness relation between output states, suggesting an influence on the contraction properties of the quantum channel. This connection has been explored before \cite{Nuradha25}, with one of the most direct consequences being that their contraction coefficient is upper bounded by $(e^\varepsilon-1)/(e^\varepsilon+1)$. In this section, we explore an analogous bound for the average contraction coefficient. Our bounds will show that many $\varepsilon$-LDP channels have average contraction that vanishes exponentially fast. This questions the utility of the outputs of these channels for computational tasks if the noise is not desgined carefully, although our results come with caveats we discuss below. 
As a warm-up, we first examine the depolarizing channel, considering its form as a product of local operations and a global action. This example illustrates our general approach, demonstrating how the LDP condition enables us to derive bounds through careful choices of the operators involved in Eq.~\eqref{eq:ldp_condition_measurements}.

\begin{example}[Global and local depolarizer]
    For both $n$-qubit local and global depolarizing channels, satisfying $\varepsilon$-LDP requires that the noise parameter $p$ approaches $1$ as $n$ increases, which implies exponentially vanishing average contraction. Here we consider input states $\rho,\sigma\sim\nu$ independently and identically distributed with $\nu$ a 2-design distribution over pure states.
    
    In Eq.~\eqref{eq:ldp_condition_measurements}, set $\rho=\ketbra{\psi}^{\otimes n}$, $\sigma=\ketbra{\phi}^{\otimes n}$ and $M=\ketbra{\psi}^{\otimes n}$, where $\ket{\psi}$ and $\ket{\phi}$ are single-qubit orthogonal states, $\braket{\phi}{\psi} = 0$. In the local and global cases, the mappings are
    \begin{align}
        \ketbra{\psi}^{\otimes n} \quad &\mapsto \quad \qty((1-p_l)\ketbra{\psi}+p_l\frac{\mathbb{I}}{2})^{\otimes n}, \\
        \ketbra{\psi}^{\otimes n} \quad &\mapsto \quad  (1-p_g) \ketbra{\psi}^{\otimes n} + p_g \frac{\mathbb{I}^{\otimes n}}{2^n},
    \end{align}
    where $p_l$ and $p_g$ are the noise parameter for each channel and $\mathbb{I}$ denotes the single-qubit identity matrix. With this choice, Eq.~\eqref{eq:ldp_condition_measurements} takes the form of a bound on the noise parameter, for each case:
    \begin{equation}
        \qty(1-\frac{p_l}{2})^n \leq e^{\varepsilon} \qty(\frac{p_l}{2})^n,
        \quad \text{and} \quad
        1-p_g+\frac{p_g}{2^n} \leq e^{\varepsilon}\frac{p_g}{2^n}.
    \end{equation}
    Hence,
    \begin{align}
        p_l &\geq \frac{2}{1+e^{\varepsilon /n}} = 
         1-(e^{\varepsilon /n}-1)/2 + \order{(e^{\varepsilon /n}-1)^2)}
         = 1 - \frac{\varepsilon}{2 n} + \order{n^{-2}}
        \quad \text{and} \quad \\
        p_g &\geq \frac{1}{1+2^{-n}(e^{\varepsilon}-1)} = 1-2^{-n} (e^{\varepsilon}-1)+\order{2^{-2n}},
    \end{align}
    which shows that $p\rightarrow 1$ as $n\rightarrow \infty$.

    For the global depolarizing channel, the contraction of the trace distance is exactly $(1-p_g)\leq 2^{-n} (e^{\varepsilon}-1)+\order{2^{-2n}}$, visibly exponentially decreasing. For the local depolarizing channel, we can resort to Prop.~\ref{prop:avc_local_depol} and note that, as for large enough $n$ we have $p_l\geq 1-1/\sqrt{3}$, the average contraction vanishes as $n\to\infty$. More precisely, following the analysis on Prop.~\ref{prop:avc_local_depol},
    \begin{equation}
        \bEx{\rho,\sigma\sim \nu}\qty[\frac{\norm{\mathcal{D}^{\otimes n}(\rho) - \mathcal{D}^{\otimes n}(\rho)}_1}{\norm{\rho-\sigma}_1}]
        \leq 
        \frac{1}{\sqrt{2}} \qty(\frac{3(1-p_l)^2+1}{2})^n+\order{2^{-n/2}}
        \leq
        \frac{1}{\sqrt{2}} \qty(\frac{3(\frac{\varepsilon}{2 n})^2+1}{2})^n+\order{2^{-n/2}},
    \end{equation}
    which is exponentially suppressed.
\end{example}

More generally, for an arbitrary $\varepsilon$-LDP channel, the following is a direct consequence of Thm.~\ref{thm:HSaveragecontraction}.

\begin{corollary}\label{cor:ldp_on_avc}
    Let us consider $T:\cS(\cH_d)\to\cS(\cH_{d'})$ an $\varepsilon$-LDP channel and $\nuHS$ the Hilbert-Schmidt distribution on $\cS(\cH_d)$. Then,
    \begin{equation}
        \eta_1(T,\norm{\cdot}_1,\nuHS\times \nuHS)
        \leq
        \alpha \sqrt{\frac{d'}{d-1}\qty(e^{\varepsilon} - 1)\; \Tr \pi^2}
    +\order{\sqrt{\frac{d' }{d^2}\sqrt{\ln d}}}
    ,
    \end{equation}
    where $\alpha=\frac{2\sqrt{2}\pi}{4+\pi}\approx 1.24$.
\end{corollary}

\begin{proof}
    This is a consequence of Thm.~\ref{thm:HSaveragecontraction} together with the inequality for the purity of the Choi state of an $\varepsilon$-LDP channel obtained below,
    \begin{equation}\label{eq:ldp_choi_purity}
        \Tr \tau^2 \leq ((1+d^{-1})e^{\varepsilon} - 1) \; \Tr \pi^2.
    \end{equation}
    Then,
    \begin{equation}
        \frac{d'}{d^2-1}\qty[d \Tr\tau^2 - \Tr \pi^2]
        \leq 
        \frac{d'}{d^2-1}\qty[(d+1)e^{\varepsilon} - d - 1]\; \Tr \pi^2
        \leq 
        \frac{d'}{d-1}\qty(e^{\varepsilon} - 1)\; \Tr \pi^2.
    \end{equation}

    To prove Eq.~\eqref{eq:ldp_choi_purity}, let us take inspiration from the previous example and set $M=T(\rho)$ in the LDP condition, such that
    \begin{equation}
        \Tr[T(\rho)^2]\leq e^{\varepsilon} \Tr[T(\rho) T(\sigma)].
    \end{equation}
    As the inequality holds for every state, we can choose $\rho, \sigma\sim\mu$ to be independently distributed with Haar distribution (any 2-design would work), and observe that the inequality must hold for the expectation values of both sides of the inequality:
    \begin{equation}
        \frac{d}{d +1}\qty[\Tr \pi^2 + \Tr \tau^2]
        \leq e^{\varepsilon} \Tr\pi^2,
    \end{equation}
    where the expectation values are computed as in Sec.~\ref{sec:expvals_unit_invar}, rearranging these terms leads to Eq.~\eqref{eq:ldp_choi_purity}.
\end{proof}

A compelling and illustrative example is an $n$-qubit unital channel $T:\cS(\cH_2^{\otimes n})\to\cS(\cH_2^{\otimes n})$. Unitality immediately implies $\pi=\mathbb{I}/2^n$ and hence that the average contraction with $\rho,\sigma\sim\nuHS$ is exponentially decreasing, as the upper bound becomes
\begin{equation}
    \alpha \sqrt{\frac{e^{\varepsilon} - 1}{2^n}}+\order{n^{1/4} 2^{-n}}.
\end{equation}
This demonstrates that for a fixed privacy parameter $\varepsilon$, increasing the system size leads to an exponentially decreasing average contraction of a unital channel under the Hilbert–Schmidt distribution. This seems to suggest that the channel essentially approaches a replacer channel, but the conclusion strongly depends on the random distribution. Alternatively, operationally motivated distributions may exhibit different average behavior and further investigation is needed to understand to what extent this result can be generalized. However, we will now argue that the average contraction alone might not be a good proxy for the performance of a classifier.

In particular, let us reflect on the context of differential privacy in quantum machine learning, which is a paradigmatic application \cite{Du2021, Watkins2023}. More precisely, we focus on quantum classification with quantum data, meaning that the objects to be classified are encoded as quantum states \cite{Gambs2008QuantumClassification, LaRose2020, Grant2018, Farhi2018}. Quantum classifiers can be structured in three steps \cite{Grant2018, Schuld2018SupervisedLearningQuantumComputers, Farhi2018, Du2021, Schuld2020}: data encoding, quantum processing and measurement, and a classical postprocessing. Assuming that our input data is quantum allows us to focus on the quantum processing and measurement. Additionally, as our interest is to find fundamental limitations imposed by LDP, we do not consider the complexities related to learning from a sample dataset \cite{Gambs2008QuantumClassification}. Therefore, we base our analysis on the following definition for the classification problem.

\begin{definition}
    A \emph{classification task} is defined by $K$ disjoint classes of quantum states $\cC_x\subset \cS(\cH_d)$ labeled by $x\in\cX$ and the goal of implementing a mapping from classifiable states to their corresponding labels.
\end{definition}

The requirement $\cC_x\cap\cC_y = \emptyset$ for $x\neq y$ is necessary for the task to be well-defined. However, the classes do not need to cover all possible states; in general, the inclusion $\bigcup_{x\in \cX} \cC_x \subset \cS(\cH_d)$ is strict. This is the reason for referring to \emph{classifiable states}, which are states $\rho$ that belong to some class, that is, $\rho\in\cC_x$ for some $x\in\cX$.

The central step in a quantum classifier is normally implemented with parametrized quantum circuits and measurements in the computational basis \cite{Schuld2020}. The map from quantum state $\rho$ to the probability distribution of the measurements can naturally be interpreted as the probability for the classifier to assign each of the labels in $\cX$ \cite{Grant2018, PerezSalinas2020, Schuld2018SupervisedLearningQuantumComputers}. With this perspective, the classical postprocessing is a selection protocol based on those probabilities, such as choosing the label with the highest probability. However, including decision rules would obscure our point, as it requires analyzing finite sample statistics to draw the impact of LDP \cite{Du2021}. With this framework, we define the probabilistic classifiers as the quantum channels implementing that central processing task:

\begin{definition}
    A \emph{probabilistic classifier} is a quantum-to-classical channel $\cK:\cS(\cH_d)\to \cS(\cH_{K})$ that maps quantum states to probability distributions over classes in $\cX$. It can be associated with a POVM $\{P_x\}_{x\in\cX}$ such that
    \begin{equation}
        \cK(\rho) = \sum_{x\in\cX} \Tr[P_x \rho] \ketbra{x}.
    \end{equation}

    $\cK$ is an \emph{ideal classifier} if for any label $x$ and classifiable state $\rho\in\cC_y$ we have that $\expval{\cK(\rho)}{x}=\delta_{x,y}$.
\end{definition}

There are various ways to guarantee $\varepsilon$-LDP for a quantum classifier \cite{Watkins2023, Du2021}. Here we consider the probabilistic classifier $\cK$ has been perturbed to be $\varepsilon$-LDP satisfying Eq.~\eqref{eq:ldp_condition_measurements}. As a direct application of Cor.~\ref{cor:ldp_on_avc}, consider $\rho,\sigma$ independent states with Hilbert-Schmidt distribution, then the average contraction obeys
\begin{equation}
    \eta_1(\cK, \norm{\cdot}_1, \nuHS\times\nuHS)
    %\bEx{\rho,\sigma\sim\nuHS}\qty[\frac{\norm{\cK(\rho) - \cK(\sigma)}_1}{\norm{\rho-\sigma}_1}]
    \leq
    \alpha \sqrt{\frac{K}{d-1}\qty(e^{\varepsilon} - 1)\; \Tr \cK\qty(\frac{\mathbb{I}}{d})^2}
+\order{\sqrt{\frac{K }{d^2}\sqrt{\ln d}}}
,
\end{equation}
which is visibly much smaller than one as the number of labels is expected to be small compared to the input dimension $K\ll d$.

We can understand this result as an implication of Cor.~\ref{cor:dtodprime}, which establishes that channels mapping high-dimensional states to lower-dimensional ones exhibit strong average contraction under the Hilbert-Schmidt distribution. This reveals that a strong average contraction is not necessarily indicative of a poor classifier, as the goal of any classifier is to ``reduce'' states to their class. Therefore, even an ideal classifier discards information that helps distinguish input states. A meaningful bound would need to demonstrate a significantly stronger contraction in a noisy classifier compared to an ideal one.

Another shortcoming of this bound is that Hilbert-Schmidt random states do not represent typical inputs for a classifier; instead, we expect states that are close to classifiable ones. Consequently, a more relevant analysis would focus on distributions over approximately classifiable states rather than arbitrary ones. This approach aligns with operational motivations and suggests the construction of distributions concentrated on state pairs that remain distinguishable under an ideal classifier while becoming closer under a noisy one, effectively isolating the effect of noise-induced contraction.

However, this analysis introduces challenges and subtleties. The main results of this work rely on distributions with some degree of uniformity, such as the Hilbert-Schmidt distribution for mixed states or 1- and 2-designs for pure states. These are well-established in quantum information and have well-understood properties \cite{Aubrun2017}. While using such distributions simplifies the task of proving strong contraction for noisy classifiers, as a side effect, it implies strong contraction for an ideal classifier. Conversely, distributions that maximize the average contraction of an ideal classifier make it difficult to find strong average contraction bounds solely using the $\varepsilon$-LDP condition and our current knowledge of average contraction. Despite these challenges, we believe this is a promising application of the moments of contraction. We present partial progress in this direction, focusing on expectation values of operationally relevant magnitudes, and leave the development of a consistent framework for future research.

There are many standard metrics to evaluate the performance of a classifier in solving a classification task \cite{Duda2000, Murphy2012, Gareth2014}. These metrics usually correspond to statistics computed on a test dataset, which in our case would be formed by instances of classifiable states $\{\rho_m\}_m$. For our discussion, we assume the states of the dataset are independent and identically distributed $\rho_m\sim\nu$ and compute the metrics as expectation values.

A common practice is to form a confusion matrix $C$ \cite{Ting2010a}, where the elements $C_{x,y}$ are the number of instances known to be in category $\cC_y$ that are labeled $x$ by the classifier $\cK$. We use the following analogous definition in terms of a probability distribution:

\begin{definition}[Confusion matrix]\label{def:confusion_matrix}
    Given a classification task with state classes $\{\cC_x\}_{x\in\cX}$ and a probability distribution over classifiable states $\nu$, the \emph{confusion matrix} associated with the probabilistic classifier $\cK$ is defined by the probability of a random state $\rho\sim\nu$ belonging to class $\cC_y$ and being assigned the label $x$,
    \begin{equation}
        C_{x,y} := \int_{\cC_y} \expval{\cK(\rho)}{x} \dd \nu.
    \end{equation}
\end{definition}

Ideally, all states are correctly labeled, and the confusion matrix is diagonal. However, we will show that if $\cK$ is a $\varepsilon$-LDP classifier, then there are necessarily non-diagonal confusion terms that obstruct perfect classification.

To give further insight, we introduce the precision, recall, and accuracy, which are standard performance metrics for classifiers \cite{Ting2010b, SammutWebb2010}. Precision for the label $x$ is the fraction of states that were correctly labeled $x$ out of all states that were labeled $x$. Conversely, recall for the label $x$ is the fraction of states correctly labeled out of all states belonging to class $\cC_x$. The three metrics can be described in terms of the confusion matrix:
\begin{equation}
    x\text{-precision} = \frac{C_{x,x}}{\sum_{y\in\cX} C_{x,y}},
    \quad \text{and} \quad
    x\text{-recall} = \frac{C_{x,x}}{\sum_{y\in\cX} C_{y,x}}.
\end{equation}
The accuracy is the fraction of correctly labeled states out of all states
\begin{equation}
    \text{accuracy} = \frac{\sum_{x\in\cX} C_{x,x}}{\sum_{x,y\in\cX} C_{x,y}}.
\end{equation}
These metrics can be expressed in terms of the confusion matrix, so we can directly adapt these definitions to the metrics depending on the probability distribution using Def.~\ref{def:confusion_matrix}.

\begin{proposition}\label{prop:ldp_on_classifier_metrics}
    Let us consider a classification task with state classes $\{\cC_x\}_{x\in\cX}$ and a probability distribution $\nu$ over classifiable states. Let $C$ be the confusion matrix of a probabilistic classifier $\cK$ that is $\varepsilon$-LDP, and define $p_x = \sum_{y\in\cX} C_{y,x}$ and $q_x=\sum_{y\in\cX} C_{x,y}$. Then, $C$ satisfies the constraints:
    \begin{equation}\label{eq:ldp_on_confusion_matrix}
        p_y C_{x,x} \leq e^{\varepsilon} p_x C_{x,y},
    \end{equation}
    as a consequence,
    \begin{align}
        x\text{-precision} &= \frac{C_{x,x}}{\sum_{y\in\cX} C_{x,y}} \leq \frac{e^{\varepsilon} p_x }{1-p_x + e^{\varepsilon} p_x},\\
        x\text{-recall} &= \frac{C_{x,x}}{\sum_{y\in\cX} C_{y,x}} \leq \frac{e^{\varepsilon} q_x}{1-p_x + e^{\varepsilon} p_x}, \\
        \text{accuracy} &= \frac{\sum_{x\in\cX} C_{x,x}}{\sum_{x,y\in\cX} C_{x,y}} \leq 
        e^{\varepsilon} \sum_{x\in\cX} p_x q_x.
    \end{align}
\end{proposition}

\begin{proof}
    Consider $\rho$ is a random state with distribution $\nu$, we will see that $p_x$ coincides with the probability that $\rho$ belongs to class $\cC_x$,
    \begin{equation}
        p_x = \Pr[\rho\in\cC_x].
    \end{equation}
    As $\nu$ is distributed over classifiable states, we have $\sum_{x\in\cX} p_x = 1$. If $p_x>0$ we can define $\sigma_x$ as the average state conditioned on states in class $\cC_x$,
    \begin{equation}
        \sigma_x = \frac{1}{p_x} \; \int_{\cC_x} \rho \; \dd \nu.
    \end{equation}
    Note that in general $\sigma_x\notin \cC_x$, but this is not necessary for the proof. With this, the confusion matrix takes the form
    \begin{equation}
        C_{x,y} =
        \int_{\cC_y} \expval{\cK(\rho)}{x} \dd \nu
        =  \expval{\cK\qty(\int_{\cC_y}\; \rho\; \dd \nu)}{x}
        = p_y \expval{\cK(\sigma_y)}{x},
    \end{equation}
    that is, it is the probability to measure label $x$ in $\cK(\sigma_y)$ multiplied by the probability of picking a state in class $\cC_y$. If $p_y=0$ we take $C_{x,y}=0$. This shows that the definition of $p_x$ is consistent with the interpretation, and shows that $q_x$ can be regarded as the probability that a state $\rho\sim\nu$ is labeled $x$:
    \begin{equation}
        q_x = \sum_{y\in\cX} C_{x,y} =  \expval{\cK\qty(\int \; \rho\; \dd \nu)}{x}.
    \end{equation}

    As $\cK$ is $\varepsilon$-LDP we have that
    \begin{equation}
        \expval{\cK(\sigma_x)}{x} \leq e^{\varepsilon} \expval{\cK(\sigma_y)}{x},
    \end{equation}
    which implies
    \begin{equation}
        p_y C_{x,x} \leq e^{\varepsilon} p_x C_{x,y}
        \quad\Rightarrow\quad
        (1-p_x) C_{x,x} \leq e^{\varepsilon} p_x (q_x-C_{x,x})
        \quad\Rightarrow\quad
        C_{x,x} \leq \frac{e^{\varepsilon} p_x q_x}{1-p_x + e^{\varepsilon} p_x}
        ,
    \end{equation}
    where in the first step we summed over all $y\neq x$. The bounds for the precision and recall are obtained by combining the last inequality with the definitions of $p_x$ and $q_x$. The accuracy is obtained summing over all $x$ and $y$ in both sides of Eq.~\eqref{eq:ldp_on_confusion_matrix} and using $\sum_{x,y\in\cX} C_{x,y}=1$.
\end{proof}

Proposition~\ref{prop:ldp_on_classifier_metrics} is less informative when the $\varepsilon$-LDP condition is weakened, i.e. for $\varepsilon\to\infty$. On the contrary, in the regime where the condition is tightened, $\varepsilon\approx 0$, the performance metrics are strongly limited:
\begin{align}
    x\text{-precision} & \leq p_x ( 1 + (1-p_x) \varepsilon ), \\
    x\text{-recall} &\leq q_x (1 + (1-p_x) \varepsilon ), \\
    \text{accuracy} &\leq (1+\varepsilon) \sum_{x\in\cX} p_x q_x.
\end{align}
As expected, the probabilistic classifier $\cK$ approaches a random classifier, one that randomly assigns label $x$ with probability $q_x$ for any input state, which corresponds to a confusion matrix $C_{x,y}=q_x p_y$.

\section{Examples and numerical results}\label{sec:numerics}

In this section, we analyze $n$-fold products of various families of single-qubit channels. We examine the asymptotic regime $n \to \infty$ using the bounds established in previous sections, and the finite-size regime through numerical simulations. For each channel family, we consider four divergences: the trace distance, the relative entropy, the $\chi^2$-divergence, and the max relative entropy. In all cases, the first state is sampled from the Haar measure, $\rho \sim \mu_{2^n}$, while the second state is fixed to the maximally mixed state $\sigma_* = \mathbb{I}^{\otimes n} / 2^n$. For the three single-parameter families, we swept across the full parameter range. For each system size and parameter value, we sampled 2100, 600, and 100 Haar-random states for $n\in\{1,2,3\}$, $n\in\{4,5,6\}$, and $n=7$, respectively; for larger dimensions, fewer samples were sufficient due to the increasing concentration properties of the random variable.

To predict the asymptotic behavior, we invoke Cor.~\ref{cor:avc_2design_mix}, which states that if $\Tr\sqrt{\Tr_2 \tau^2} < 1$, then the upper bound for the average contraction vanishes as $n \to \infty$. As noted in Remark~\ref{remark:avc2design_nproduct_about_simplificaiton}, this condition fully characterizes our asymptotic bound because, in all considered cases, $\Tr_2 \tau^2$ and $\pi^2$ commute. We complement this result with Thm.~\ref{thm:avc_lower_bound}, which shows that the average contraction approaches $1$ as $n \to \infty$ whenever $S(\tau) < \log(\|\pi\|_\infty^{-1})$. The full calculations are presented in App.~\ref{sec:supmat_numerics}.

\paragraph{Depolarizing.}
By the results in Prop.~\ref{prop:avc_local_depol}, we have that the average trace distance admits an asymptotically vanishing upper bound for $p \gtrsim 0.42$, and a lower bound that tends to one for $p \lesssim 0.25$. These limiting values are shown as solid black lines in Fig.~\ref{fig:depolarizing_param_tr}. For small values of $n$, numerical simulations reveal an initial decrease in the average contraction for all $p$, followed by a reversal of this trend at low values of $p$. That is, for weak noise and sufficiently large $n$, the average contraction increases, consistent with its asymptotic convergence to one. Conversely, for higher noise levels (larger $p$), the average contraction consistently decreases as the system size grows. In Fig.~\ref{fig:depolarizing_param_tr}, the transition between these regimes occurs between $p \approx 0.25$ and $p \approx 0.42$, in agreement with the asymptotic predictions.

\begin{figure}[htb]
	\subcaptionsetup[figure]{position=top,slc=off, margin={0pt,0pt}, skip=-10pt}
		\subcaptionbox{\label{fig:depolarizing_param_tr}}
		{\includegraphics[width=.45\textwidth]{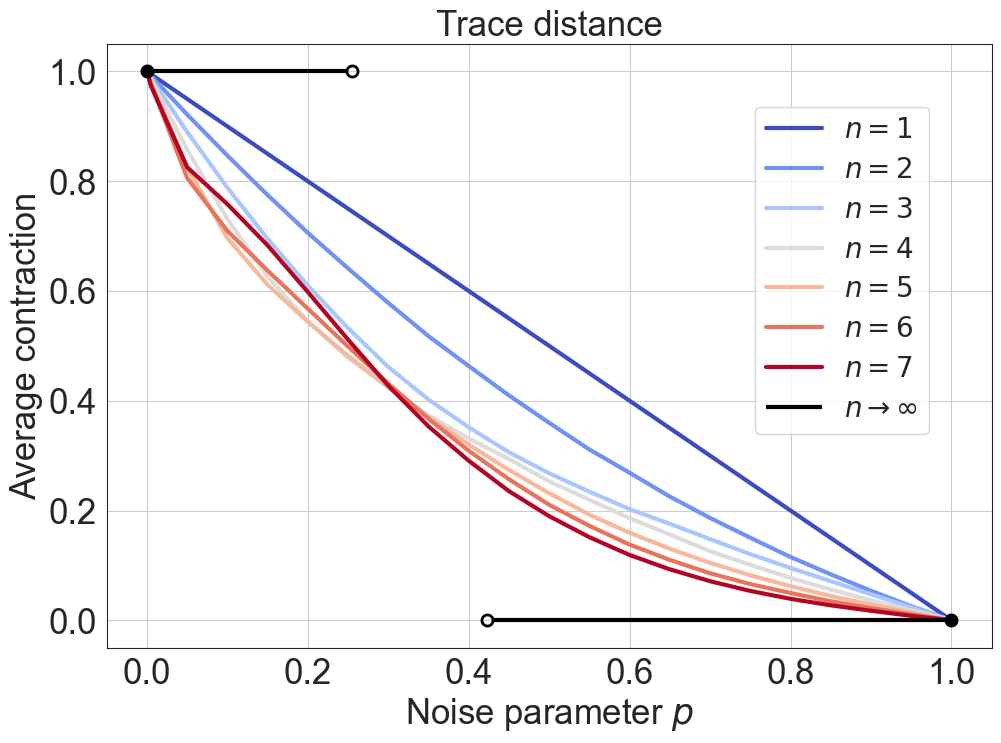}}
	\hfill
		\subcaptionbox{\label{fig:depolarizing_param_me}}
		{\includegraphics[width=.45\textwidth]{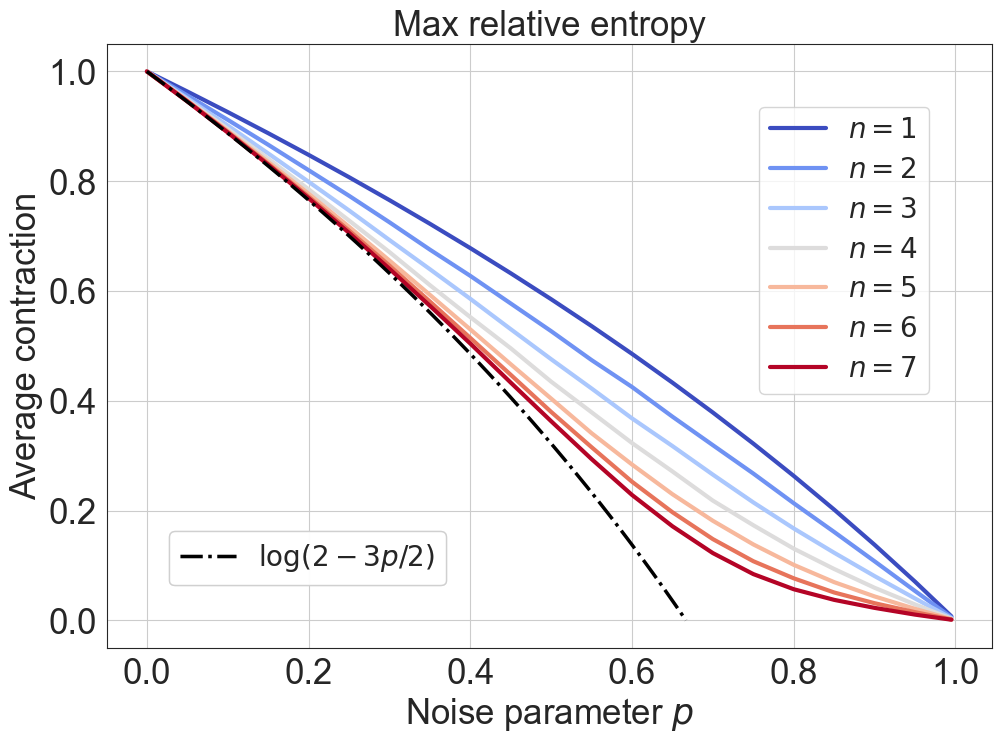}}
		\subcaptionbox{\label{fig:depolarizing_param_re}}
		{\includegraphics[width=.45\textwidth]{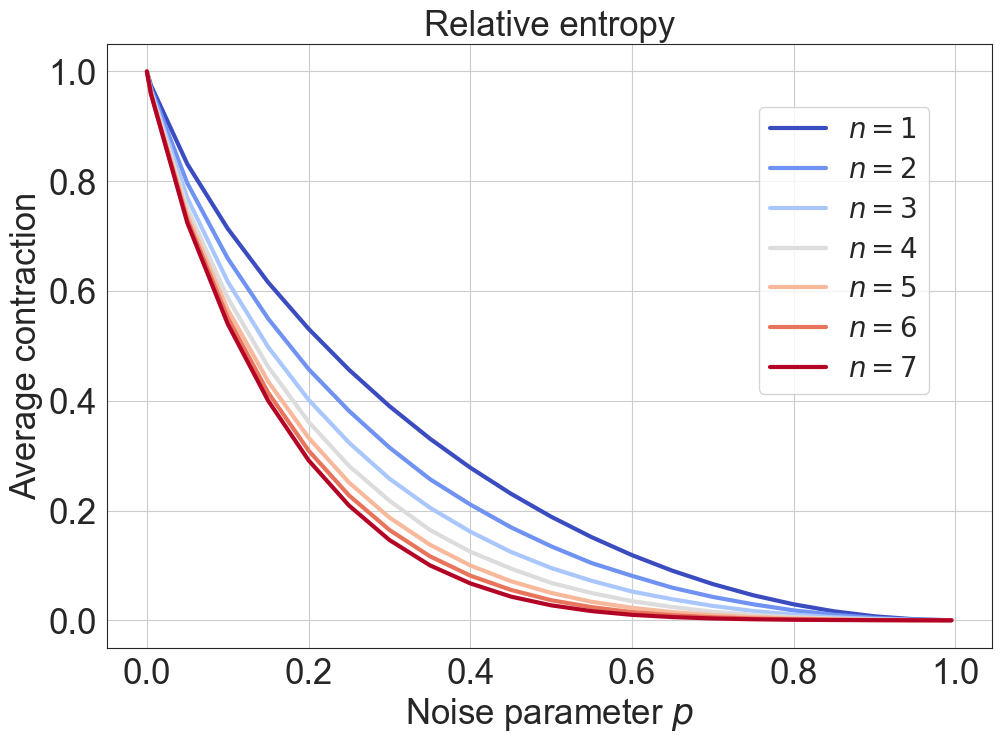}}
	\hfill
		\subcaptionbox{\label{fig:depolarizing_param_c2}}
		{\includegraphics[width=.45\textwidth]{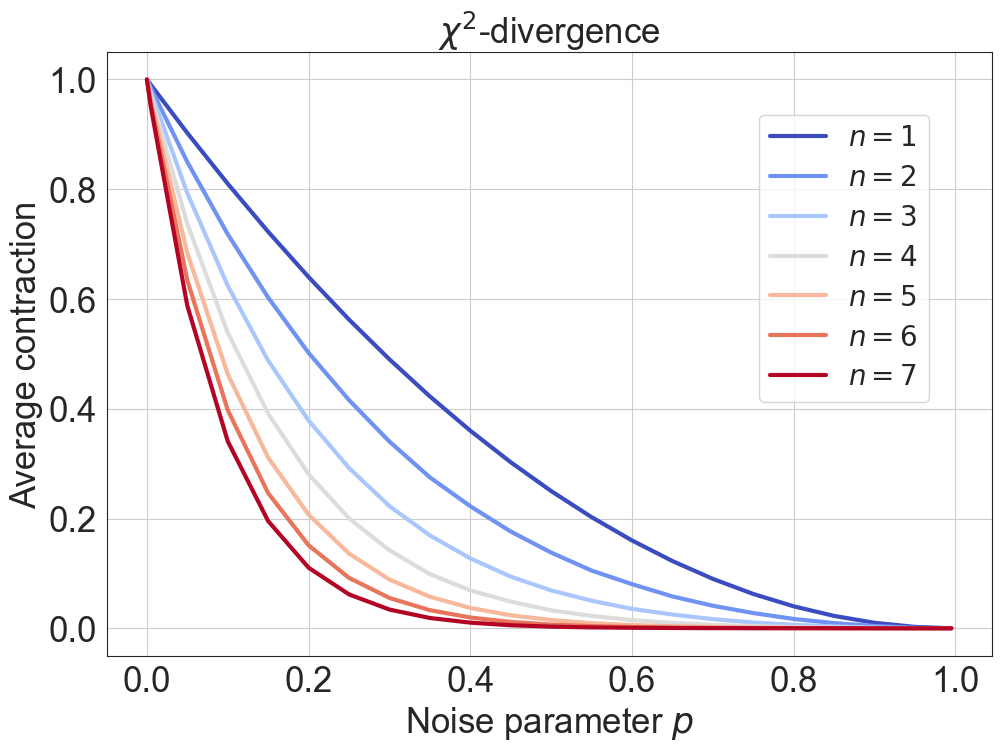}}
	\caption{
	Average contraction as a function of the depolarizing noise parameter $p$ for different system sizes $n$, across four divergences. (a) Trace distance: the average contraction initially decreases for all $p$, but increases at low noise levels ($p\lesssim 0.25$) as $n$ grows, consistent with the asymptotic lower bound (solid black line) and vanishing upper bound (black dot) for $p\gtrsim 0.42$. (b) Max relative entropy: numerical values show a $p$-dependent trend; a lower bound based on $\log(2-3p/2)$ is plotted as a dashed line. (c) Relative entropy and (d) $\chi^2$-divergence: both display a monotonic decrease in average contraction, dominated by the trace distance trend and consistent with Cor.~\ref{cor:asympt_bounds_tr2Df_avc_wrt_max_mix}.}
	\label{fig:depolarizing_param}
\end{figure}

In contrast, we observe monotonic in $n$ behavior for the relative entropy and the $\chi^2$-divergence, as shown in Figs.~\ref{fig:depolarizing_param_re} and \ref{fig:depolarizing_param_c2}, respectively. In both cases, the average contraction decreases rapidly and is largely dominated by the trend of the trace distance, following Cor.~\ref{cor:asympt_bounds_tr2Df_avc_wrt_max_mix}. This additionally agrees with the prediction that the purity of the Choi state effectively gives the average contraction of the $\chi^2$-divergence.

The numerical results for the max relative entropy, shown in Fig.~\ref{fig:depolarizing_param_me}, exhibit a seemingly $p$-dependent trend: for large $p$, the average contraction decreases, while for small $p$, the values appear to collapse onto a single curve. Proving this asymptotic behavior lies beyond the scope of the bounds developed in earlier sections. Nevertheless, in App.~\ref{sec:supmat_numerics}, we derive a lower bound that closely matches the numerical data. This bound is represented by a dashed line in Fig.~\ref{fig:depolarizing_param_me}, it corresponds to
\begin{equation}
	D_{\max}(T_{\text{depol}}^{\otimes n}(\rho) \| \mathbb{I}/2^n) \geq n \cdot \ln(2 - \frac{3}{2} p),
\end{equation}
which remains informative as long as $p < 2/3$.

\paragraph{Dephasing.} Consider the single-qubit dephasing channel with dephasing strength $0\leq\gamma \leq 1$, expressed as a phase flip:
\begin{equation}
	T_{\text{deph}}(\rho) = (1 - \frac{\gamma}{2}) \rho + \frac{\gamma}{2} Z \rho Z,
	\quad Z = \mqty(1 & 0 \\ 0 & -1).
\end{equation}
In App.~\ref{sec:supmat_numerics}, we show that $\Tr\sqrt{\Tr_2 \tau_{\text{deph}}^2}\geq 1$ for all $\gamma$, while $S(\tau_{\text{deph}}) < 1$ for any $\gamma < 1$. Therefore, the average contraction converges to one as $n\to\infty$ for all $\gamma < 1$, except for the extremal case $\gamma=1$. The latter case corresponds to an absolute dephasing channel and can be related to known results to prove that as $n\to\infty$ the average contraction tends to $e^{-1}$, which is shown as a dark dot in Fig.~\ref{fig:dephasing_param_tr}.

\begin{figure}[h!]
	\subcaptionsetup[figure]{position=top,slc=off, margin={0pt,0pt}, skip=-10pt}
		\subcaptionbox{\label{fig:dephasing_param_tr}}
		{\includegraphics[width=.45\textwidth]{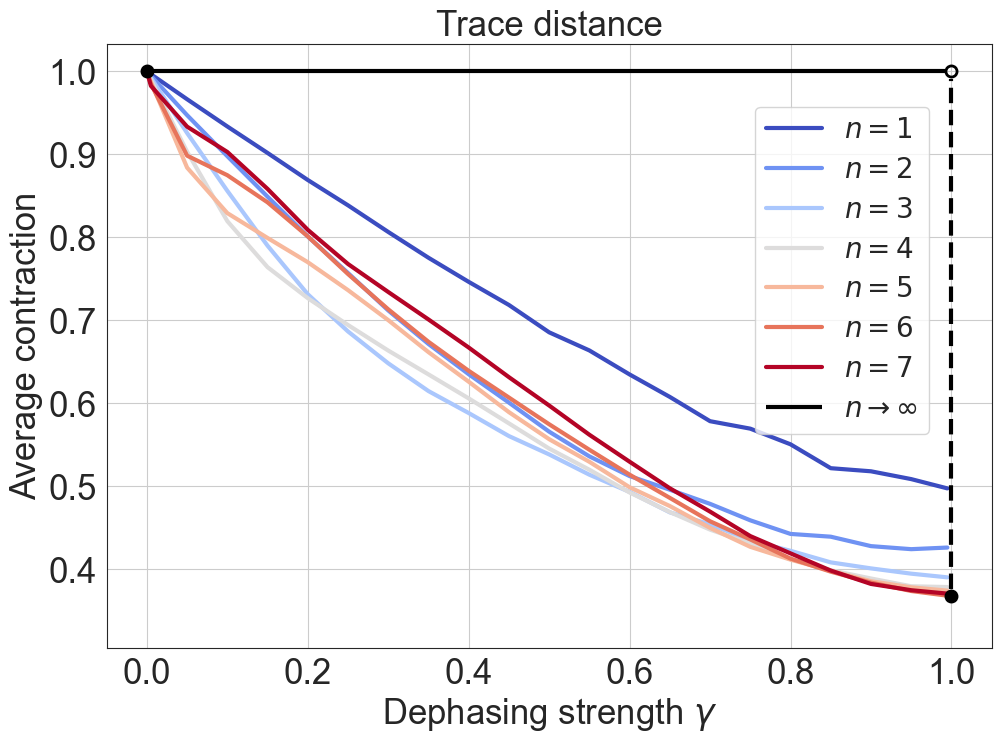}}
	\hfill
		\subcaptionbox{\label{fig:dephasing_param_me}}
		{\includegraphics[width=.45\textwidth]{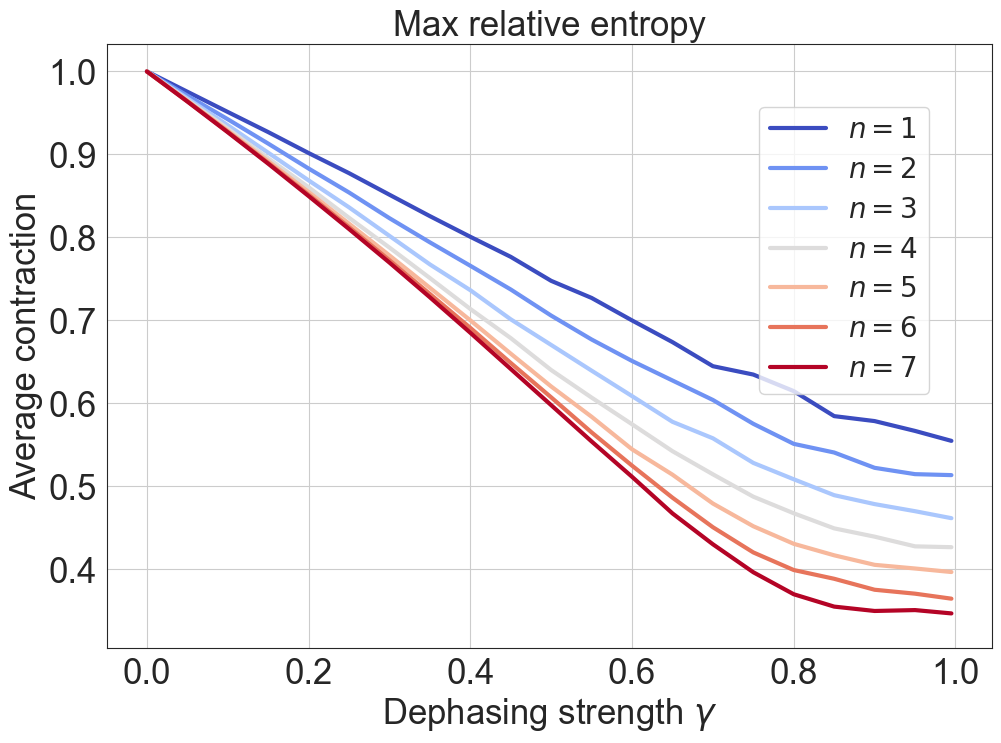}}
		\subcaptionbox{\label{fig:dephasing_param_re}}
		{\includegraphics[width=.45\textwidth]{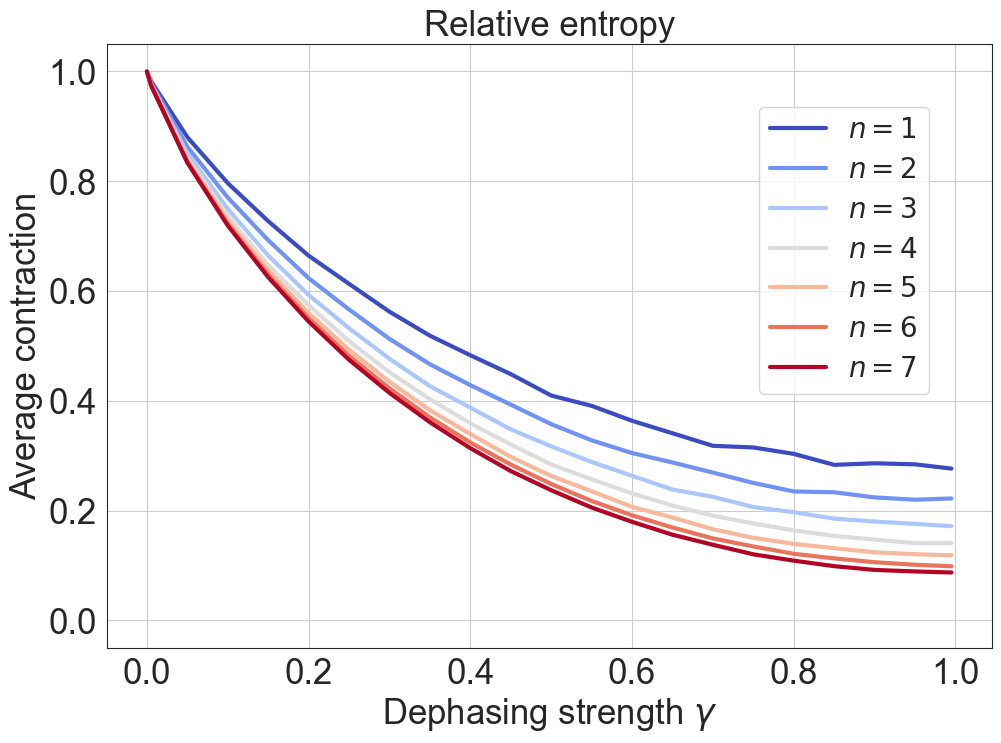}}
	\hfill
		\subcaptionbox{\label{fig:dephasing_param_c2}}
		{\includegraphics[width=.45\textwidth]{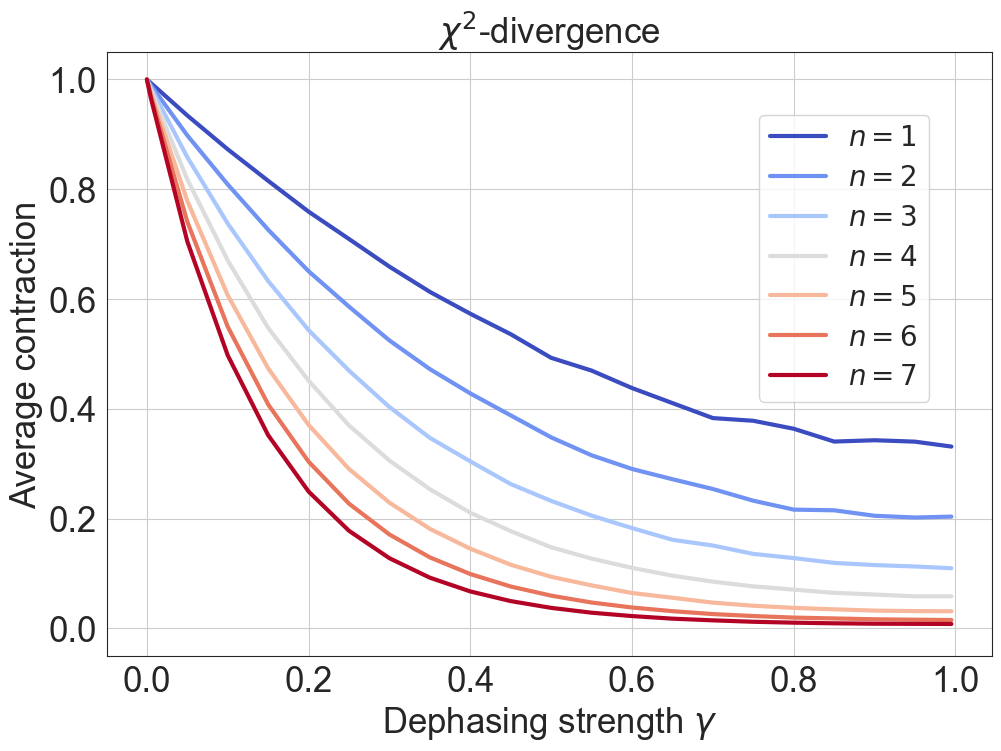}}
	\caption{
	Average contraction as a function of the dephasing strength $\gamma$ for products of the single-qubit dephasing channel, shown for different system sizes $n$ across four divergences. (a) Trace distance: the average contraction converges to $1$ as $n\to\infty$, except at $\gamma=1$, where it saturates to a known value $e^{-1}$. (b) Max relative entropy: exhibits a similar trend to the depolarizing case. (c) Relative entropy and (d) $\chi^2$-divergence: both show monotonic decay with increasing system size; while the $\chi^2$-divergence appears to vanish as $n\to\infty$, the relative entropy seems to converge to a nonzero value. Note that for $n\in\{1,2\}$ we sometimes observe a small non-monotone behavior of the average contraction in terms of the noise parameter. This is an effect of shot noise and the fact that the average contraction exhibits weaker concentration properties in small dimensions.
	}
	\label{fig:dephasing_param}
\end{figure}

As for the other divergences depicted in Fig.~\ref{fig:dephasing_param}, the max relative entropy displays behavior similar to that of the depolarizing case, although we could not find a simple matching lower bound. The $\chi^2$-divergence appears to decay monotonically for all nonzero values of $\gamma$. In contrast, the relative entropy seems to saturate at a positive value for $\gamma \approx 1$.

\paragraph{Amplitude damping.}
The single-qubit amplitude damping channel, with damping strength $0\leq \lambda \leq 1$, is defined as
\begin{equation}
	T_{\text{damp}}(\rho) = K \rho K + \lambda \ketbra{0}{1} \rho \ketbra{1}{0},
	\quad K = \mqty(1 & 0 \\ 0 & \sqrt{1-\lambda}).
\end{equation}
The quantity $\Tr\sqrt{\Tr_2 \tau_{\text{damp}}^2}$ falls below $1$ for $\lambda \lesssim 0.46$, which implies asymptotically vanishing upper bounds on average contraction for those values. Since the channel is non-unital, we need to compute $\norm{\pi}_{\infty} = (1+\lambda)/2$. The condition $S(\tau_{\text{damp}}) < \log \norm{\pi}_{\infty}^{-1} = -\log[ (1+\lambda) /2]$ holds for $\lambda \lesssim 0.20$, indicating that the lower bound in Thm.~\ref{thm:avc_lower_bound} predicts asymptotic convergence to $1$ in that regime.

\begin{figure}[h!]
	\subcaptionsetup[figure]{position=top,slc=off, skip=-10pt, margin={0pt,0pt}}
		\subcaptionbox{\label{fig:amplitude_damping_param_tr}}
		{\includegraphics[width=.45\textwidth]{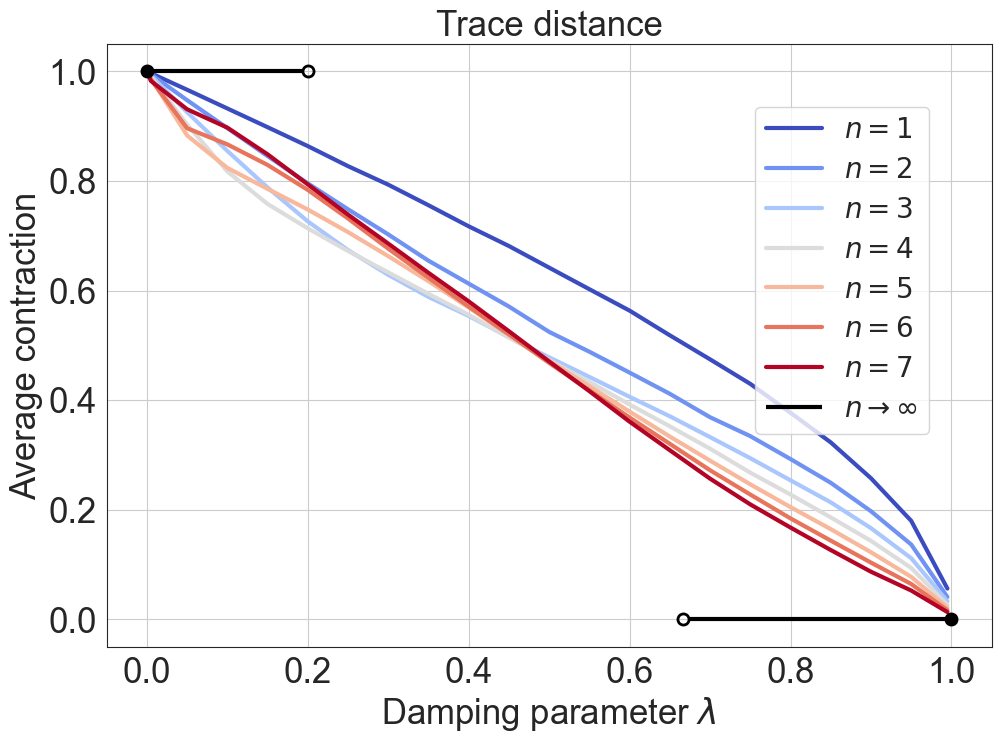}}
	\hfill
		\subcaptionbox{\label{fig:amplitude_damping_param_me}}
		{\includegraphics[width=.45\textwidth]{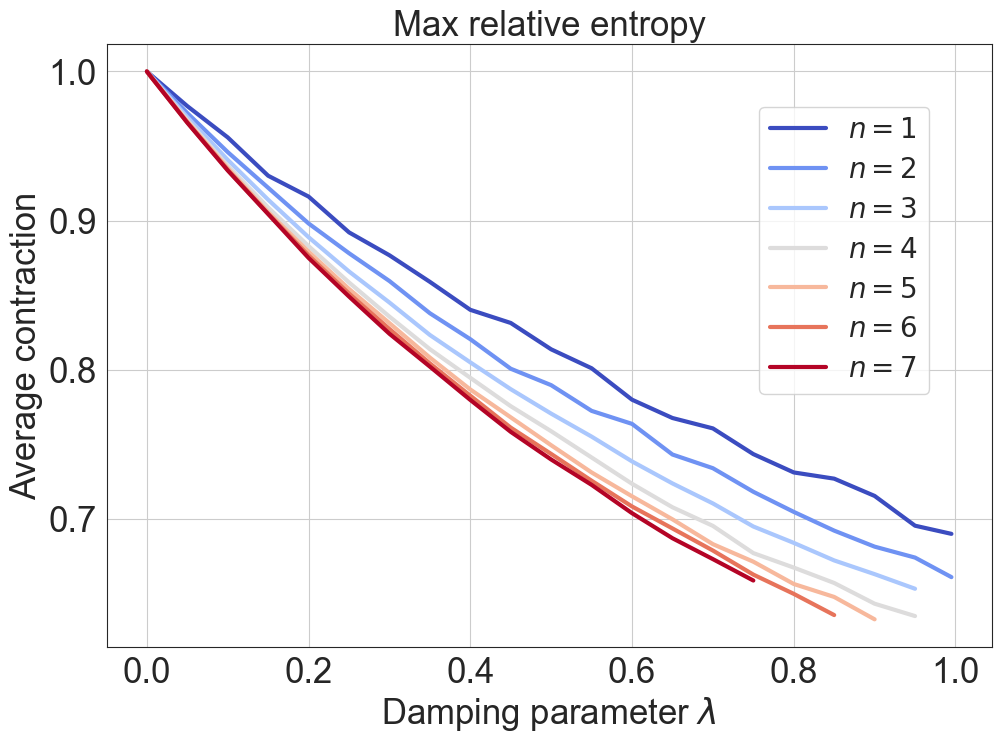}}
		\subcaptionbox{\label{fig:amplitude_damping_param_re}}
		{\includegraphics[width=.45\textwidth]{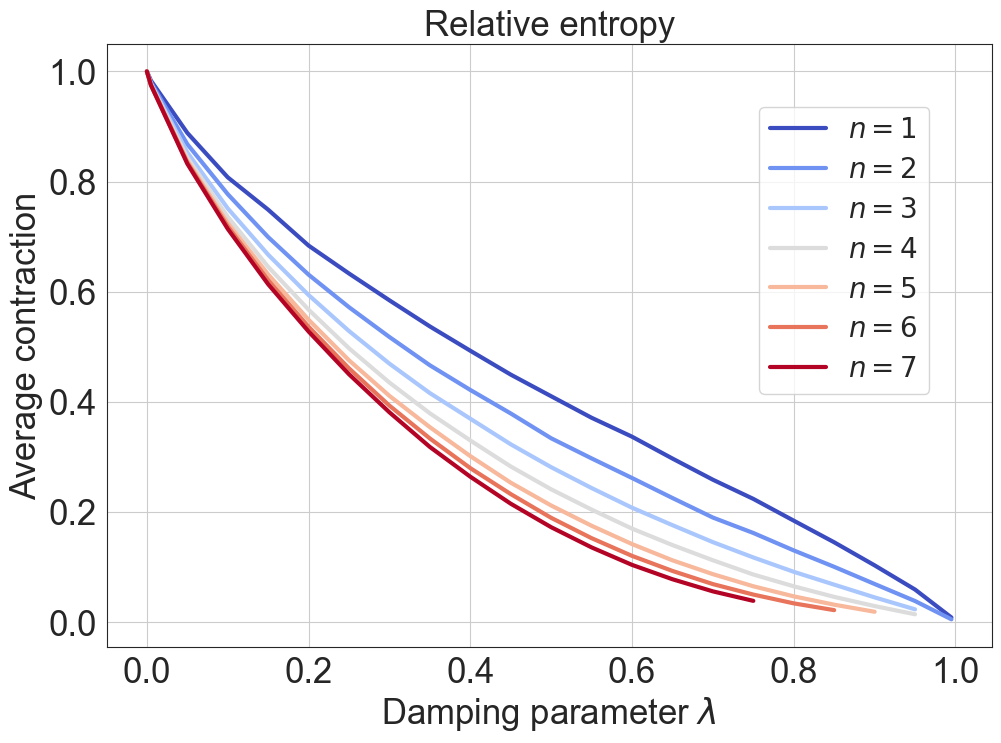}}
	\hfill
		\subcaptionbox{\label{fig:amplitude_damping_param_c2}}
		{\includegraphics[width=.45\textwidth]{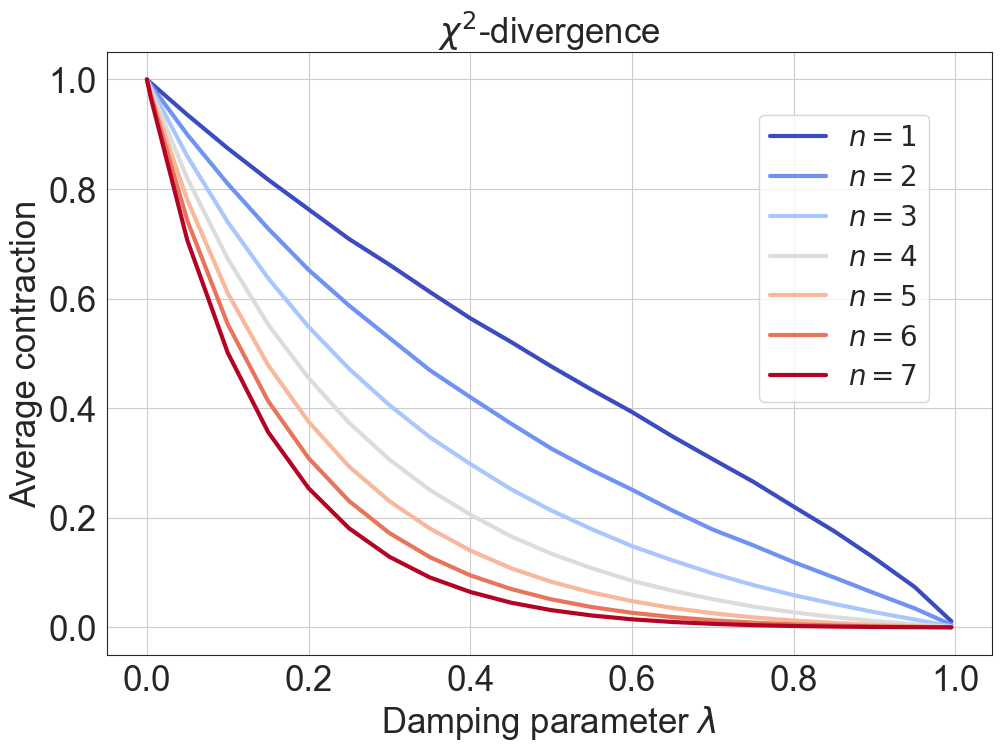}}
	\caption{
	Average contraction as a function of the damping strength $\lambda$ for products of the single-qubit amplitude damping channel, shown for different system sizes $n$ across four divergences. (a) Trace distance: the average contraction respectively increases or decreases with system size for both small and large values of $\lambda$, consistent with the asymptotic bounds predicting convergence to $1$ for $\lambda\lesssim 0.20$ and vanishing for $\lambda \gtrsim 0.46$. (b) Max relative entropy: displays a non-vanishing trend for all values of $\lambda$, remarkably far from $0$ even for large values. (c) Relative entropy and (d) $\chi^2$-divergence: both decrease monotonically with increasing $\lambda$, following a trend consistent with vanishing as $n\to\infty$.	For (b) and (c), we omitted the points where the state $T_{\text{damp}}(\sigma_*)$ was not invertible due to lack of numerical precision. Note that for $n\in\{1,2\}$ we sometimes observe a small non-monotone behavior of the average contraction in terms of the noise parameter, particularly for the max-relative entropy. This is an effect of shot noise and the fact that the average contraction exhibits weaker concentration properties in small dimensions.
	}
	\label{fig:amplitude_damping_param}
\end{figure}

These observations are illustrated in Fig.~\ref{fig:amplitude_damping_param_tr}. Even for small system sizes, the trace distance shows a trend compatible with the asymptotic predictions, with the average contraction increasing with $n$ for small values of $\lambda$ and decreasing for larger values. In contrast, the max relative entropy exhibits a markedly different behavior, with a non-vanishing trend visible in Fig.~\ref{fig:amplitude_damping_param_me}. In the plots, we omitted certain points where the state $T_{\text{damp}}(\sigma_*)$ was not invertible due to insufficient numerical precision. The numerics suggest that this sensitivity is sufficient to produce a non-vanishing average contraction even for values of $\lambda$ close to $1$. Both the relative entropy and the $\chi^2$-divergence (Fig.~\ref{fig:amplitude_damping_param_re} and \ref{fig:amplitude_damping_param_c2}) decrease monotonically and follow a trend consistent with vanishing in the limit $n\to\infty$.

\paragraph{Mixture of orthogonal unitaries.} Given a set of unitaries $\{U_i\}_{i=1}^M$ with $\Tr[U_i^\dag U_j] = \delta_{ij} d$, and a probability distribution $\mathbf{p}$, we consider the  quantum channel defined by applying unitary $U_i$ with probability $p_i$, that is $T_{\text{MU}}(\rho) = \sum_{i=1}^M p_i U_i \rho U_i^{\dag}$. Then, the asymptotic behavior of the average contraction of the trace distance:
\begin{equation}
	\lim_{n\to\infty}\eta_1(T_{\text{MU}}^{\otimes n}, \norm{\cdot}_1, \mu_{d^{n}}\times \delta_{\sigma_*})
	=
	\begin{cases}\label{equ:cases_orthogonal}
		1 \quad &\text{if } H(\mathbf{p}) < \log d, \\
		0 \quad &\text{if } \norm{\mathbf{p}}_2 < d^{-1/2}.
	\end{cases}
\end{equation}
Note that if $M \leq d$, the first condition is always satisfied except for the uniform mixture and the second is always violated. Also note that the cases in~\eqref{equ:cases_orthogonal} do not cover all parameter space.

As a concrete example, consider a single-qubit Pauli channel with three components,
\begin{equation}\label{eq:simple_pauli_channel}
	\rho \mapsto (1-p-q) \rho + p X \rho X + q Z \rho Z,
\end{equation}
here $U_i\in\{\mathbb{I}, X, Z\}$ and $\mathbf{p} = (1-p-q, p, q)$. Fig.~\ref{fig:mix_unit_regimes} shows the different regions in the $(p, q)$ plane where our bounds predict definite asymptotic behavior.  Within the red triangular frame the channel is equivalent to a dephasing channel, thus all points converge to one, except for the special points $(q,p)\in\{(0, 1/2), (1/2, 0), (1/2, 1/2)\}$, which correspond to absolute dephasing, and the average contraction tends to $e^{-1}$ as established above. For points on the boundary of the blue ellipse, on the inner boundary of the red region, and within the three unshaded sectors, the asymptotic behavior remains undetermined by our current results.

\section{Conclusion and open questions}
In this work, we introduced the concept of the moments of contraction of a quantum channel under a divergence as a way to obtain a more fine-grained understanding of ``how noisy" a quantum channel is on an ensemble of inputs. We showed some structural properties of this quantity and derived upper and lower bounds for certain divergences and classes of channels, paying particular attention to product channels and the trace distance.
We uncovered that the average contraction coefficient exhibits interesting properties, such as sharp phase transitions in terms of the strength of the noise and exponential separations from the maximal contraction coefficient. We applied these findings to study noisy quantum circuits and quantum differential privacy, showing how the average input state perspective is useful and important for such applications of contraction coefficients.

Our work leaves open various questions to develop a better understanding of when moments of contraction for one divergence dominate the moments of contraction for another divergence. Various results in this direction are known for the worst-case contraction~\cite{Hirche23}. In addition, it is intriguing to explore to what extent the phase transitions we uncovered for the trace distance can happen for other divergences and investigate if there are intermediate regimes. That is, we found examples where the average contraction coefficient either converges exponentially fast to $1$ or $0$, or converges to a constant $0<c<1$ for a set of noise parameters of measure $0$. It would be interesting to investigate if there are examples where it converges to a nontrivial constant for sets of positive measure, or regimes where it converges at a rate that is not exponential. As can be seen in Fig.~\ref{fig:mix_unit_regimes}, even for simple models like Pauli channels, there are parameter regions for which we do not know the behavior of the average contraction.

\begin{figure}[h!]
	\centering
	\includegraphics[width=0.45\textwidth]{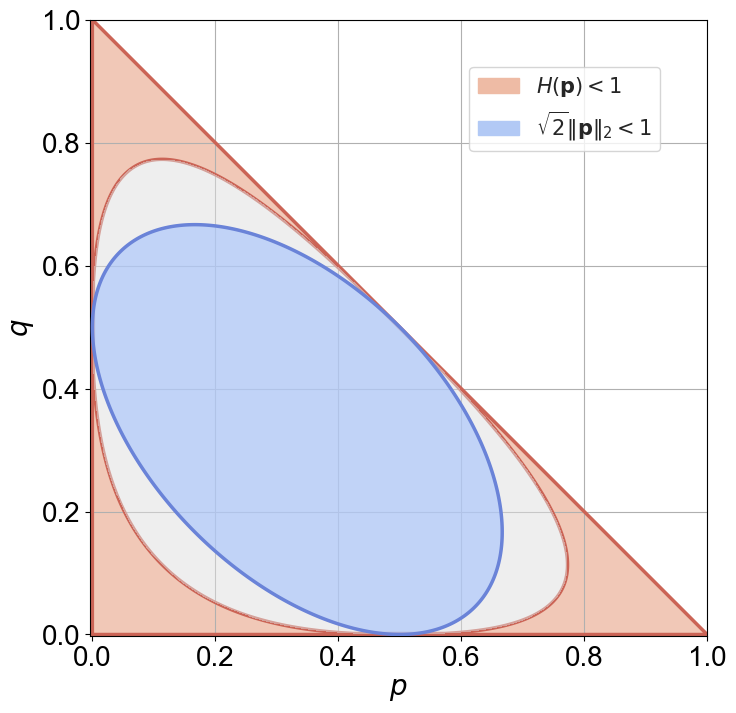}
	\caption{
	Asymptotic average contraction for $n$-fold products of the Pauli channel in Eq.~\eqref{eq:simple_pauli_channel}. In the red-shaded region where $H(\mathbf{p}) < 1$, the average contraction converges to one as $n\to\infty$. In the blue-shaded ellipse where $\sqrt{2}\norm{\mathbf{p}}_2 < 1$, it vanishes asymptotically. These bounds guarantee the limiting behavior only within the shaded areas; the boundaries may exhibit distinct behavior.
	}
		\label{fig:mix_unit_regimes}
\end{figure}

Finally, it would be interesting to understand if the framework of average-case contraction can be applied to the study of mixing times of quantum Markov chains to improve or simplify existing results and to generalize the framework of reverse Pinsker inequalities for general $f$-divergences to simplify and generalize our results for $f$-divergences.

\section{Acknowledgments}
We thank Mikel Sanz, David P\'erez-Garc\'ia, Mario Berta, and Christoph Hirche for helpful discussions. 
R.I. acknowledges the support of the Basque Government Ph.D. Grant No. PRE 2021-1-0102. R.I. also acknowledges support from HORIZON-CL4-2022-QUANTUM01-SGA project 101113946 OpenSuperQ-Plus100 of the EU Flagship on Quantum Technologies, the Spanish Ram\'on y Cajal Grant RYC-2020-030503-I, and the “Generaci\'on de Conocimiento” project Grant No. PID2021-125823NA-I00 funded by MICIU/\allowbreak AEI/\allowbreak 10.13039/\allowbreak 501100011033, by “ERDF Invest in your Future” and by FEDER EU. R.I. also acknowledges support from the Basque Government through Grants No. IT1470-22, the Elkartek project KUBIT KK-2024/00105, and from the IKUR Strategy under the collaboration agreement between Ikerbasque Foundation and BCAM on behalf of the Department of Education of the Basque Government. This work has also been partially supported by the Ministry for Digital Transformation and the Civil Service of the Spanish Government through the QUANTUM ENIA project call – Quantum Spain project, and by the European Union through the Recovery, Transformation and Resilience Plan – NextGenerationEU within the framework of the Digital Spain 2026 Agenda.
D.S.F. acknowledges financial support from the Novo Nordisk Foundation (Grant No. NNF20OC0059939 Quantum for Life) and by the ERC grant GIFNEQ 101163938.

This work is funded by the European Union. Views and opinions expressed are however those of the author(s) only and do not necessarily reflect those of the European Union or the European Research Council Executive Agency. Neither the European Union nor the granting authority can be held responsible for them.
\bibliographystyle{alphaurl}
\bibliography{avcontract_updated,library_updated}

\newcommand{\etalchar}[1]{$^{#1}$}
\begin{thebibliography}{PSCLGFL20}

\bibitem[ABO08]{Aharonov2008}
Dorit Aharonov and Michael Ben-Or.
\newblock Fault-tolerant quantum computation with constant error rate.
\newblock {\em SIAM Journal on Computing}, 38(4):1207–1282, January 2008.
\newblock \href {https://doi.org/10.1137/s0097539799359385}
  {\path{doi:10.1137/s0097539799359385}}.

\bibitem[ABOIN96]{aharonov_depo}
D.~Aharonov, M.~Ben-Or, R.~Impagliazzo, and N.~Nisan.
\newblock Limitations of noisy reversible computation, 1996.
\newblock \href {https://doi.org/10.48550/arxiv.quant-ph/9611028}
  {\path{doi:10.48550/arxiv.quant-ph/9611028}}.

\bibitem[ADK23]{Angrisani2023DPamplificationquantum}
Armando Angrisani, Mina Doosti, and Elham Kashefi.
\newblock Differential privacy amplification in quantum and quantum-inspired
  algorithms, 2023.
\newblock \href {https://arxiv.org/abs/2203.03604} {\path{arXiv:2203.03604}}.

\bibitem[AE07]{Ambainis2007}
Andris Ambainis and Joseph Emerson.
\newblock Quantum t-designs: t-wise independence in the quantum world.
\newblock In {\em Twenty-Second Annual IEEE Conference on Computational
  Complexity (CCC’07)}, page 129–140. IEEE, June 2007.
\newblock \href {https://doi.org/10.1109/ccc.2007.26}
  {\path{doi:10.1109/ccc.2007.26}}.

\bibitem[AK23]{Angrisani2023QDP}
Armando Angrisani and Elham Kashefi.
\newblock Quantum local differential privacy and quantum statistical query
  model, 2023.
\newblock \href {https://arxiv.org/abs/2203.03591} {\path{arXiv:2203.03591}}.

\bibitem[AK25]{Angrisani2025}
Armando Angrisani and Elham Kashefi.
\newblock Quantum differential privacy in the local model.
\newblock {\em IEEE Transactions on Information Theory}, 71(5):3675–3692, May
  2025.
\newblock \href {https://doi.org/10.1109/tit.2025.3552671}
  {\path{doi:10.1109/tit.2025.3552671}}.

\bibitem[AS66]{AliSilvey1966}
S.~M. Ali and S.~D. Silvey.
\newblock A general class of coefficients of divergence of one distribution
  from another.
\newblock {\em Journal of the Royal Statistical Society. Series B
  (Methodological)}, 28(1):131--142, 1966.

\bibitem[AS17]{Aubrun2017}
Guillaume Aubrun and Stanisław Szarek.
\newblock {\em Alice and Bob Meet Banach}.
\newblock American Mathematical Society, August 2017.
\newblock \href {https://doi.org/10.1090/surv/223}
  {\path{doi:10.1090/surv/223}}.

\bibitem[AW01]{Ahlswede2001}
R.~Ahlswede and A.~Winter.
\newblock Strong converse for identification via quantum channels, 2001.
\newblock \href {https://arxiv.org/abs/quant-ph/0012127}
  {\path{arXiv:quant-ph/0012127}}.

\bibitem[Bha96]{Bhatia1996}
R.~Bhatia.
\newblock {\em Matrix Analysis}.
\newblock Graduate Texts in Mathematics. Springer New York, 1996.

\bibitem[BKZ06]{BlumeWojciech06}
Robin Blume-Kohout and Wojciech~H. Zurek.
\newblock Quantum darwinism: Entanglement, branches, and the emergent
  classicality of redundantly stored quantum information.
\newblock {\em Phys. Rev. A}, 73:062310, Jun 2006.
\newblock \href {https://doi.org/10.1103/PhysRevA.73.062310}
  {\path{doi:10.1103/PhysRevA.73.062310}}.

\bibitem[BSW23]{Berta2023}
Mario Berta, David Sutter, and Michael Walter.
\newblock {Quantum Brascamp–Lieb dualities}.
\newblock {\em Communications in Mathematical Physics}, 401(2):1807–1830,
  March 2023.
\newblock \href {https://doi.org/10.1007/s00220-023-04678-w}
  {\path{doi:10.1007/s00220-023-04678-w}}.

\bibitem[CBB{\etalchar{+}}23]{Cai2023}
Zhenyu Cai, Ryan Babbush, Simon~C. Benjamin, Suguru Endo, William~J. Huggins,
  Ying Li, Jarrod~R. McClean, and Thomas~E. O’Brien.
\newblock Quantum error mitigation.
\newblock {\em Reviews of Modern Physics}, 95(4), December 2023.
\newblock \href {https://doi.org/10.1103/revmodphys.95.045005}
  {\path{doi:10.1103/revmodphys.95.045005}}.

\bibitem[CHR24]{hirche_privateht}
Hao–Chung Cheng, Christoph Hirche, and Cambyse Rouzé.
\newblock Sample complexity of locally differentially private quantum
  hypothesis testing.
\newblock In {\em 2024 IEEE International Symposium on Information Theory
  (ISIT)}, pages 2921--2926, 2024.
\newblock \href {https://doi.org/10.1109/ISIT57864.2024.10619433}
  {\path{doi:10.1109/ISIT57864.2024.10619433}}.

\bibitem[CS04]{CsizarShields2004}
I.~Csiszár and P.C. Shields.
\newblock Information theory and statistics: A tutorial.
\newblock {\em Foundations and Trends® in Communications and Information
  Theory}, 1(4):417--528, 2004.
\newblock \href {https://doi.org/10.1561/0100000004}
  {\path{doi:10.1561/0100000004}}.

\bibitem[C{\'{S}}06]{Collins2006}
Beno{\^i}t Collins and Piotr {\'{S}}niady.
\newblock Integration with respect to the {Haar} measure on unitary, orthogonal
  and symplectic group.
\newblock {\em Communications in Mathematical Physics}, 264(3):773--795, Jun
  2006.
\newblock \href {https://doi.org/10.1007/s00220-006-1554-3}
  {\path{doi:10.1007/s00220-006-1554-3}}.

\bibitem[Csi63]{Csiszar1963}
Imre Csisz\'ar.
\newblock Eine informationstheoretische ungleichung und ihre anwendung auf den
  beweis der ergodizitat von markoffschen ketten.
\newblock {\em Magyar. Tud. Akad. Mat. Kutato Int. Kozl}, pages 85--108, 1963.

\bibitem[DCEL09]{Dankert2009}
Christoph Dankert, Richard Cleve, Joseph Emerson, and Etera Livine.
\newblock Exact and approximate unitary 2-designs and their application to
  fidelity estimation.
\newblock {\em Phys. Rev. A}, 80:012304, Jul 2009.
\newblock \href {https://doi.org/10.1103/PhysRevA.80.012304}
  {\path{doi:10.1103/PhysRevA.80.012304}}.

\bibitem[DHL{\etalchar{+}}21]{Du2021}
Yuxuan Du, Min-Hsiu Hsieh, Tongliang Liu, Dacheng Tao, and Nana Liu.
\newblock Quantum noise protects quantum classifiers against adversaries.
\newblock {\em Physical Review Research}, 3(2), May 2021.
\newblock \href {https://doi.org/10.1103/physrevresearch.3.023153}
  {\path{doi:10.1103/physrevresearch.3.023153}}.

\bibitem[DHL{\etalchar{+}}22]{Du2022}
Yuxuan Du, Min-Hsiu Hsieh, Tongliang Liu, Shan You, and Dacheng Tao.
\newblock Quantum differentially private sparse regression learning.
\newblock {\em IEEE Transactions on Information Theory}, 68(8):5217–5233,
  August 2022.
\newblock \href {https://doi.org/10.1109/tit.2022.3164726}
  {\path{doi:10.1109/tit.2022.3164726}}.

\bibitem[DHS00]{Duda2000}
Richard~O. Duda, Peter~E. Hart, and David~G. Stork.
\newblock {\em Pattern Classification (2nd Edition)}.
\newblock Wiley-Interscience, USA, 2000.

\bibitem[DNS{\etalchar{+}}21]{Deshpande2021TightBO}
A.~Deshpande, Pradeep Niroula, O.~Shtanko, A.~Gorshkov, Bill Fefferman, and
  M.~Gullans.
\newblock Tight bounds on the convergence of noisy random circuits to the
  uniform distribution.
\newblock {\em PRX Quantum}, 2021.

\bibitem[DR13]{Dwork2014}
Cynthia Dwork and Aaron Roth.
\newblock The algorithmic foundations of differential privacy.
\newblock {\em Foundations and Trends{\textregistered} in Theoretical Computer
  Science}, 9(3–4):211–407, 2013.
\newblock \href {https://doi.org/10.1561/0400000042}
  {\path{doi:10.1561/0400000042}}.

\bibitem[FHT03]{Fedotov03}
A.A. Fedotov, P.~Harremoes, and F.~Topsoe.
\newblock Refinements of pinsker's inequality.
\newblock {\em IEEE Transactions on Information Theory}, 49(6):1491--1498,
  2003.
\newblock \href {https://doi.org/10.1109/TIT.2003.811927}
  {\path{doi:10.1109/TIT.2003.811927}}.

\bibitem[FN18]{Farhi2018}
Edward Farhi and Hartmut Neven.
\newblock Classification with quantum neural networks on near term processors,
  2018.
\newblock \href {https://arxiv.org/abs/1802.06002} {\path{arXiv:1802.06002}}.

\bibitem[Fre23]{Frenkel2023}
Péter~E. Frenkel.
\newblock Integral formula for quantum relative entropy implies data processing
  inequality.
\newblock {\em Quantum}, 7:1102, September 2023.
\newblock \href {https://doi.org/10.22331/q-2023-09-07-1102}
  {\path{doi:10.22331/q-2023-09-07-1102}}.

\bibitem[GAE07]{Gross2007}
D.~Gross, K.~Audenaert, and J.~Eisert.
\newblock Evenly distributed unitaries: On the structure of unitary designs.
\newblock {\em Journal of Mathematical Physics}, 48(5), May 2007.
\newblock \href {https://doi.org/10.1063/1.2716992}
  {\path{doi:10.1063/1.2716992}}.

\bibitem[Gam08]{Gambs2008QuantumClassification}
Sébastien Gambs.
\newblock Quantum classification, 2008.
\newblock \href {https://arxiv.org/abs/0809.0444} {\path{arXiv:0809.0444}}.

\bibitem[GBC{\etalchar{+}}18]{Grant2018}
Edward Grant, Marcello Benedetti, Shuxiang Cao, Andrew Hallam, Joshua Lockhart,
  Vid Stojevic, Andrew~G. Green, and Simone Severini.
\newblock Hierarchical quantum classifiers.
\newblock {\em npj Quantum Information}, 4(1), December 2018.
\newblock \href {https://doi.org/10.1038/s41534-018-0116-9}
  {\path{doi:10.1038/s41534-018-0116-9}}.

\bibitem[GGTC22]{Gonzalez-Garcia2022}
Guillermo Gonz\'alez-Garc\'{\i}a, Rahul Trivedi, and J.~Ignacio Cirac.
\newblock Error propagation in {NISQ} devices for solving classical
  optimization problems.
\newblock {\em PRX Quantum}, 3:040326, Dec 2022.
\newblock \href {https://doi.org/10.1103/PRXQuantum.3.040326}
  {\path{doi:10.1103/PRXQuantum.3.040326}}.

\bibitem[Gil10]{Gilardoni10}
Gustavo~L. Gilardoni.
\newblock {On Pinsker's and Vajda's Type Inequalities for Csiszár's $f$
  -Divergences}.
\newblock {\em IEEE Transactions on Information Theory}, 56(11):5377--5386,
  2010.
\newblock \href {https://doi.org/10.1109/TIT.2010.2068710}
  {\path{doi:10.1109/TIT.2010.2068710}}.

\bibitem[Hel69]{Helstrom1969}
Carl~W. Helstrom.
\newblock Quantum detection and estimation theory.
\newblock {\em Journal of Statistical Physics}, 1(2):231--252, Jun 1969.
\newblock \href {https://doi.org/10.1007/BF01007479}
  {\path{doi:10.1007/BF01007479}}.

\bibitem[Hir24]{Hirche2024}
Christoph Hirche.
\newblock Quantum doeblin coefficients: A simple upper bound on contraction
  coefficients.
\newblock In {\em 2024 IEEE International Symposium on Information Theory
  (ISIT)}, page 557–562. IEEE, July 2024.
\newblock \href {https://doi.org/10.1109/isit57864.2024.10619667}
  {\path{doi:10.1109/isit57864.2024.10619667}}.

\bibitem[HLP34]{Hardy1934}
G.~H. Hardy, John~E. Littlewood, and George P\'olya.
\newblock {\em Inequalities}.
\newblock Cambridge University Press, 1934.
\newblock Chapter 4.8.

\bibitem[HMPB11]{Hiai2010}
Fumio Hiai, Mil\'{a}n Mosonyi, D\'{e}nes Petz, and C\'{e}dric B\'{e}ny.
\newblock Quantum $f$-divergences and error correction.
\newblock {\em Reviews in Mathematical Physics}, 23(07):691--747, 2011.
\newblock \href {https://doi.org/10.1142/S0129055X11004412}
  {\path{doi:10.1142/S0129055X11004412}}.

\bibitem[HP91]{Hiai1991}
Fumio Hiai and Dénes Petz.
\newblock The proper formula for relative entropy and its asymptotics in
  quantum probability.
\newblock {\em Communications in Mathematical Physics}, 143(1):99–114,
  December 1991.
\newblock \href {https://doi.org/10.1007/bf02100287}
  {\path{doi:10.1007/bf02100287}}.

\bibitem[HR15]{Hiai15}
Fumio Hiai and Mary~Beth Ruskai.
\newblock {Contraction coefficients for noisy quantum channels}.
\newblock {\em Journal of Mathematical Physics}, 57(1):015211, 12 2015.
\newblock \href {https://doi.org/10.1063/1.4936215}
  {\path{doi:10.1063/1.4936215}}.

\bibitem[HRF23]{Hirche2023QuantumDP}
Christoph Hirche, Cambyse Rouzé, and Daniel~Stilck Fran\c{c}a.
\newblock Quantum differential privacy: An information theory perspective.
\newblock {\em IEEE Transactions on Information Theory}, 69(9):5771–5787,
  September 2023.
\newblock \href {https://doi.org/10.1109/tit.2023.3272904}
  {\path{doi:10.1109/tit.2023.3272904}}.

\bibitem[HRSF22]{Hirche2022}
Christoph Hirche, Cambyse Rouzé, and Daniel Stilck~Fran\c{c}a.
\newblock On contraction coefficients, partial orders and approximation of
  capacities for quantum channels.
\newblock {\em Quantum}, 6:862, November 2022.
\newblock \href {https://doi.org/10.22331/q-2022-11-28-862}
  {\path{doi:10.22331/q-2022-11-28-862}}.

\bibitem[HT24]{Hirche23}
Christoph Hirche and Marco Tomamichel.
\newblock Quantum {R\'enyi} and f-divergences from integral representations.
\newblock {\em Communications in Mathematical Physics}, 405(9), 2024.
\newblock \href {https://doi.org/10.1007/s00220-024-05087-3}
  {\path{doi:10.1007/s00220-024-05087-3}}.

\bibitem[HZ08]{Heinosaari2008}
Teiko Heinosaari and Mario Ziman.
\newblock Guide to mathematical concepts of quantum theory.
\newblock {\em Acta Physica Slovaca. Reviews and Tutorials}, 58(4), August
  2008.
\newblock \href {https://doi.org/10.2478/v10155-010-0091-y}
  {\path{doi:10.2478/v10155-010-0091-y}}.

\bibitem[JWHT14]{Gareth2014}
Gareth James, Daniela Witten, Trevor Hastie, and Robert Tibshirani.
\newblock {\em An Introduction to Statistical Learning: with Applications in
  R}.
\newblock Springer Publishing Company, Incorporated, 2014.

\bibitem[KRUdW08]{Kempe2008}
Julia Kempe, Oded Regev, Falk Unger, and Ronald de~Wolf.
\newblock {\em Upper Bounds on the Noise Threshold for Fault-Tolerant Quantum
  Computing}, page 845–856.
\newblock Springer Berlin Heidelberg, 2008.
\newblock \href {https://doi.org/10.1007/978-3-540-70575-8_69}
  {\path{doi:10.1007/978-3-540-70575-8_69}}.

\bibitem[LC20]{LaRose2020}
Ryan LaRose and Brian Coyle.
\newblock Robust data encodings for quantum classifiers.
\newblock {\em Physical Review A}, 102(3), September 2020.
\newblock \href {https://doi.org/10.1103/physreva.102.032420}
  {\path{doi:10.1103/physreva.102.032420}}.

\bibitem[Lin75]{Lindblad1975}
G\"{o}ran Lindblad.
\newblock Completely positive maps and entropy inequalities.
\newblock {\em Communications in Mathematical Physics}, 40(2):147–151, June
  1975.
\newblock \href {https://doi.org/10.1007/bf01609396}
  {\path{doi:10.1007/bf01609396}}.

\bibitem[LR73]{Lieb1973}
Elliott~H. Lieb and Mary~Beth Ruskai.
\newblock Proof of the strong subadditivity of quantum-mechanical entropy.
\newblock {\em Journal of Mathematical Physics}, 14(12):1938–1941, December
  1973.
\newblock \href {https://doi.org/10.1063/1.1666274}
  {\path{doi:10.1063/1.1666274}}.

\bibitem[LR99]{LesniewskiRuskai1999}
Andrew Lesniewski and Mary~Beth Ruskai.
\newblock Monotone {Riemannian} metrics and relative entropy on noncommutative
  probability spaces.
\newblock {\em Journal of Mathematical Physics}, 40(11):5702–5724, November
  1999.
\newblock \href {https://doi.org/10.1063/1.533053}
  {\path{doi:10.1063/1.533053}}.

\bibitem[Mat18]{Matsumoto2018}
Keiji Matsumoto.
\newblock A new quantum version of f-divergence.
\newblock In Masanao Ozawa, Jeremy Butterfield, Hans Halvorson, Mikl{\'o}s
  R{\'e}dei, Yuichiro Kitajima, and Francesco Buscemi, editors, {\em Reality
  and Measurement in Algebraic Quantum Theory}, pages 229--273, Singapore,
  2018. Springer Singapore.

\bibitem[Mel24]{Mele2024}
Antonio~Anna Mele.
\newblock Introduction to {Haar} measure tools in quantum information: A
  beginner's tutorial.
\newblock {\em Quantum}, 8:1340, May 2024.
\newblock \href {https://doi.org/10.22331/q-2024-05-08-1340}
  {\path{doi:10.22331/q-2024-05-08-1340}}.

\bibitem[MHF18]{MllerHermes2018}
Alexander M\"{u}ller-Hermes and Daniel~Stilck Franca.
\newblock Sandwiched {R}ényi convergence for quantum evolutions.
\newblock {\em Quantum}, 2:55, February 2018.
\newblock \href {https://doi.org/10.22331/q-2018-02-27-55}
  {\path{doi:10.22331/q-2018-02-27-55}}.

\bibitem[MHSFW16]{MllerHermes2016}
Alexander M\"{u}ller-Hermes, Daniel Stilck~Fran\c{c}a, and Michael~M. Wolf.
\newblock Relative entropy convergence for depolarizing channels.
\newblock {\em Journal of Mathematical Physics}, 57(2), January 2016.
\newblock \href {https://doi.org/10.1063/1.4939560}
  {\path{doi:10.1063/1.4939560}}.

\bibitem[Mur12]{Murphy2012}
Kevin~P. Murphy.
\newblock {\em Machine Learning: A Probabilistic Perspective}.
\newblock The MIT Press, 2012.

\bibitem[NC10]{NielsenChuang2010}
Michael~A. Nielsen and Isaac~L. Chuang.
\newblock {\em Quantum Computation and Quantum Information: 10th Anniversary
  Edition}.
\newblock Cambridge University Press, 2010.

\bibitem[NGW24]{Nuradha2024}
Theshani Nuradha, Ziv Goldfeld, and Mark~M. Wilde.
\newblock Quantum pufferfish privacy: A flexible privacy framework for quantum
  systems.
\newblock {\em IEEE Transactions on Information Theory}, 70(8):5731–5762,
  August 2024.
\newblock \href {https://doi.org/10.1109/tit.2024.3404927}
  {\path{doi:10.1109/tit.2024.3404927}}.

\bibitem[NW25]{Nuradha25}
Theshani Nuradha and Mark~M. Wilde.
\newblock Contraction of private quantum channels and private quantum
  hypothesis testing.
\newblock {\em IEEE Transactions on Information Theory}, 71(3):1851--1873,
  2025.
\newblock \href {https://doi.org/10.1109/TIT.2025.3527859}
  {\path{doi:10.1109/TIT.2025.3527859}}.

\bibitem[Pau03]{Paulsen2003}
Vern Paulsen.
\newblock {\em Completely Bounded Maps and Operator Algebras}.
\newblock Cambridge Studies in Advanced Mathematics. Cambridge University
  Press, 2003.

\bibitem[Pet85]{Petz1985}
D{\'e}nes Petz.
\newblock Quasi-entropies for states of a von {Neumann} algebra.
\newblock {\em Publications of the Research Institute for Mathematical
  Sciences}, 21(4):787--800, 1985.

\bibitem[Pet86]{Petz1986}
Dénes Petz.
\newblock Quasi-entropies for finite quantum systems.
\newblock {\em Reports on Mathematical Physics}, 23(1):57--65, 1986.
\newblock \href {https://doi.org/10.1016/0034-4877(86)90067-4}
  {\path{doi:10.1016/0034-4877(86)90067-4}}.

\bibitem[PGWPR06]{PerezGarcia2006}
David Pérez-García, Michael~M. Wolf, Denes Petz, and Mary~Beth Ruskai.
\newblock {Contractivity of positive and trace-preserving maps under $L_p$
  norms}.
\newblock {\em Journal of Mathematical Physics}, 47(8), August 2006.
\newblock \href {https://doi.org/10.1063/1.2218675}
  {\path{doi:10.1063/1.2218675}}.

\bibitem[PPZ16]{Puchala16}
Zbigniew Puchala, Lukasz Pawela, and Karol \.{Z}yczkowski.
\newblock Distinguishability of generic quantum states.
\newblock {\em Physical Review A}, 93:062112, 6 2016.
\newblock \href {https://doi.org/10.1103/PhysRevA.93.062112}
  {\path{doi:10.1103/PhysRevA.93.062112}}.

\bibitem[PR98]{PetzRuskai1998}
D\'{e}nes Petz and Mary~Beth Ruskai.
\newblock Contraction of generalized relative entropy under stochastic mappings
  on matrices.
\newblock {\em Infinite Dimensional Analysis, Quantum Probability and Related
  Topics}, 01(01):83--89, 1998.
\newblock \href {https://doi.org/10.1142/S0219025798000077}
  {\path{doi:10.1142/S0219025798000077}}.

\bibitem[Pre18]{Preskill2018}
John Preskill.
\newblock {Quantum Computing in the NISQ era and beyond}.
\newblock {\em Quantum}, 2:79, August 2018.
\newblock \href {https://doi.org/10.22331/q-2018-08-06-79}
  {\path{doi:10.22331/q-2018-08-06-79}}.

\bibitem[PSCLGFL20]{PerezSalinas2020}
Adrián Pérez-Salinas, Alba Cervera-Lierta, Elies Gil-Fuster, and José~I.
  Latorre.
\newblock Data re-uploading for a universal quantum classifier.
\newblock {\em Quantum}, 4:226, February 2020.
\newblock \href {https://doi.org/10.22331/q-2020-02-06-226}
  {\path{doi:10.22331/q-2020-02-06-226}}.

\bibitem[QAS21]{Quek2021PrivateLearning}
Yihui Quek, Srinivasan Arunachalam, and John~A Smolin.
\newblock Private learning implies quantum stability.
\newblock {\em Advances in Neural Information Processing Systems},
  34:20503--20515, 2021.

\bibitem[QFK{\etalchar{+}}24]{Quek2024}
Yihui Quek, Daniel~Stilck França, Sumeet Khatri, Johannes~Jakob Meyer, and
  Jens Eisert.
\newblock Exponentially tighter bounds on limitations of quantum error
  mitigation.
\newblock {\em Nature Physics}, 20(10):1648--1658, 2024.
\newblock \href {https://doi.org/10.1038/s41567-024-02536-7}
  {\path{doi:10.1038/s41567-024-02536-7}}.

\bibitem[RFA25]{gibbs_samplers}
Cambyse Rouz\'{e}, Daniel~Stilck Fran\c{c}a, and \'{A}lvaro~M. Alhambra.
\newblock Efficient thermalization and universal quantum computing with quantum
  gibbs samplers.
\newblock In {\em Proceedings of the 57th Annual ACM Symposium on Theory of
  Computing}, STOC '25, page 1488–1495, New York, NY, USA, 2025. Association
  for Computing Machinery.
\newblock \href {https://doi.org/10.1145/3717823.3718268}
  {\path{doi:10.1145/3717823.3718268}}.

\bibitem[RS09]{Roy2009}
Aidan Roy and A.~J. Scott.
\newblock Unitary designs and codes.
\newblock {\em Designs, Codes and Cryptography}, 53(1):13–31, April 2009.
\newblock \href {https://doi.org/10.1007/s10623-009-9290-2}
  {\path{doi:10.1007/s10623-009-9290-2}}.

\bibitem[SBSW20]{Schuld2020}
Maria Schuld, Alex Bocharov, Krysta~M. Svore, and Nathan Wiebe.
\newblock Circuit-centric quantum classifiers.
\newblock {\em Physical Review A}, 101(3), March 2020.
\newblock \href {https://doi.org/10.1103/physreva.101.032308}
  {\path{doi:10.1103/physreva.101.032308}}.

\bibitem[SFGP21]{StilckFrana2021}
Daniel Stilck~Fran\c{c}a and Raul García-Patrón.
\newblock Limitations of optimization algorithms on noisy quantum devices.
\newblock {\em Nature Physics}, 17(11):1221–1227, October 2021.
\newblock \href {https://doi.org/10.1038/s41567-021-01356-3}
  {\path{doi:10.1038/s41567-021-01356-3}}.

\bibitem[Sha48]{Shannon1948}
C.~E. Shannon.
\newblock A mathematical theory of communication.
\newblock {\em Bell System Technical Journal}, 27(3):379–423, July 1948.
\newblock \href {https://doi.org/10.1002/j.1538-7305.1948.tb01338.x}
  {\path{doi:10.1002/j.1538-7305.1948.tb01338.x}}.

\bibitem[Sha12]{Sharma2012}
Naresh Sharma.
\newblock Equality conditions for the quantum f-relative entropy and
  generalized data processing inequalities.
\newblock {\em Quantum Information Processing}, 11(1):137--160, Feb 2012.
\newblock \href {https://doi.org/10.1007/s11128-011-0238-x}
  {\path{doi:10.1007/s11128-011-0238-x}}.

\bibitem[SP18]{Schuld2018SupervisedLearningQuantumComputers}
Maria Schuld and Francesco Petruccione.
\newblock {\em Supervised learning with quantum computers}.
\newblock Quantum Science and Technology (QST). Springer, 2018.

\bibitem[SV16]{Sason2016}
Igal Sason and Sergio Verdu.
\newblock $f$ -divergence inequalities.
\newblock {\em IEEE Transactions on Information Theory}, 62(11):5973–6006,
  November 2016.
\newblock \href {https://doi.org/10.1109/tit.2016.2603151}
  {\path{doi:10.1109/tit.2016.2603151}}.

\bibitem[SW10]{SammutWebb2010}
Claude Sammut and Geoffrey~I. Webb, editors.
\newblock {\em Accuracy}, pages 9--10.
\newblock Springer US, Boston, MA, 2010.
\newblock \href {https://doi.org/10.1007/978-0-387-30164-8_3}
  {\path{doi:10.1007/978-0-387-30164-8_3}}.

\bibitem[SZP16]{Singh2016}
Uttam Singh, Lin Zhang, and Arun~Kumar Pati.
\newblock Average coherence and its typicality for random pure states.
\newblock {\em Phys. Rev. A}, 93:032125, Mar 2016.
\newblock \href {https://doi.org/10.1103/PhysRevA.93.032125}
  {\path{doi:10.1103/PhysRevA.93.032125}}.

\bibitem[TEMG22]{Takagi2022}
Ryuji Takagi, Suguru Endo, Shintaro Minagawa, and Mile Gu.
\newblock Fundamental limits of quantum error mitigation.
\newblock {\em npj Quantum Information}, 8(1), September 2022.
\newblock \href {https://doi.org/10.1038/s41534-022-00618-z}
  {\path{doi:10.1038/s41534-022-00618-z}}.

\bibitem[Tin10a]{Ting2010a}
Kai~Ming Ting.
\newblock Confusion matrix.
\newblock In Claude Sammut and Geoffrey~I. Webb, editors, {\em Encyclopedia of
  Machine Learning}, pages 209--209. Springer US, Boston, MA, 2010.
\newblock \href {https://doi.org/10.1007/978-0-387-30164-8_157}
  {\path{doi:10.1007/978-0-387-30164-8_157}}.

\bibitem[Tin10b]{Ting2010b}
Kai~Ming Ting.
\newblock Precision and recall.
\newblock In Claude Sammut and Geoffrey~I. Webb, editors, {\em Encyclopedia of
  Machine Learning}, pages 781--781. Springer US, Boston, MA, 2010.
\newblock \href {https://doi.org/10.1007/978-0-387-30164-8_652}
  {\path{doi:10.1007/978-0-387-30164-8_652}}.

\bibitem[TKR{\etalchar{+}}10]{Temme10}
K.~Temme, M.~J. Kastoryano, M.~B. Ruskai, M.~M. Wolf, and F.~Verstraete.
\newblock {{The $\chi^2$-divergence and mixing times of quantum Markov
  processes}}.
\newblock {\em Journal of Mathematical Physics}, 51(12):122201, 12 2010.
\newblock \href {https://doi.org/10.1063/1.3511335}
  {\path{doi:10.1063/1.3511335}}.

\bibitem[TTG23]{Takagi2023}
Ryuji Takagi, Hiroyasu Tajima, and Mile Gu.
\newblock Universal sampling lower bounds for quantum error mitigation.
\newblock {\em Physical Review Letters}, 131(21), November 2023.
\newblock \href {https://doi.org/10.1103/physrevlett.131.210602}
  {\path{doi:10.1103/physrevlett.131.210602}}.

\bibitem[Ume62]{Umegaki1962}
Hisaharu Umegaki.
\newblock {Conditional expectation in an operator algebra. IV. Entropy and
  information}.
\newblock {\em Kodai Mathematical Journal}, 14(2), January 1962.
\newblock \href {https://doi.org/10.2996/kmj/1138844604}
  {\path{doi:10.2996/kmj/1138844604}}.

\bibitem[Wat18]{Watrous2018}
John Watrous.
\newblock {\em The Theory of Quantum Information}.
\newblock Cambridge University Press, 2018.

\bibitem[WCY23]{Watkins2023}
William~M. Watkins, Samuel Yen-Chi Chen, and Shinjae Yoo.
\newblock Quantum machine learning with differential privacy.
\newblock {\em Scientific Reports}, 13(1):2453, Feb 2023.
\newblock \href {https://doi.org/10.1038/s41598-022-24082-z}
  {\path{doi:10.1038/s41598-022-24082-z}}.

\bibitem[Wil17]{Wilde2017}
Mark~M. Wilde.
\newblock {\em Quantum Information Theory}.
\newblock Cambridge University Press, 2 edition, 2017.

\bibitem[Wil18]{Wilde2018}
Mark~M Wilde.
\newblock Optimized quantum f-divergences and data processing.
\newblock {\em Journal of Physics A: Mathematical and Theoretical},
  51(37):374002, aug 2018.
\newblock \href {https://doi.org/10.1088/1751-8121/aad5a1}
  {\path{doi:10.1088/1751-8121/aad5a1}}.

\bibitem[Wol12]{MWolf2012Guidedtour}
Michael~M. Wolf.
\newblock Quantum channels and operations - guided tour, 2012.

\bibitem[ZHSL98]{Zyczkowski1998}
Karol \.{Z}yczkowski, Pawe\l{} Horodecki, Anna Sanpera, and Maciej Lewenstein.
\newblock Volume of the set of separable states.
\newblock {\em Phys. Rev. A}, 58:883--892, Aug 1998.
\newblock \href {https://doi.org/10.1103/PhysRevA.58.883}
  {\path{doi:10.1103/PhysRevA.58.883}}.

\bibitem[ZS01]{Zyczkowski2001}
Karol \.{Z}yczkowski and Hans-J\"urgen Sommers.
\newblock Induced measures in the space of mixed quantum states.
\newblock {\em Journal of Physics A: Mathematical and General}, 34(35):7111,
  aug 2001.
\newblock \href {https://doi.org/10.1088/0305-4470/34/35/335}
  {\path{doi:10.1088/0305-4470/34/35/335}}.

\bibitem[ZY17]{Zhou2017}
Li~Zhou and Mingsheng Ying.
\newblock Differential privacy in quantum computation.
\newblock In {\em 2017 IEEE 30th Computer Security Foundations Symposium
  (CSF)}, page 249–262. IEEE, August 2017.
\newblock \href {https://doi.org/10.1109/csf.2017.23}
  {\path{doi:10.1109/csf.2017.23}}.

\end{thebibliography}

\pagebreak
\appendix

\section{Expectation values for \texorpdfstring{$t$}{t}-invariant distributions}\label{sec:expvals_unit_invar}

A distribution $\nu$ over Hermitian matrices is unitarily invariant if it satisfies $U \Lambda U^{\dag} \sim \nu$ for $\Lambda \sim \nu$ and any unitary $U$. When such a distribution is defined over quantum states, it reflects the absence of a preferred basis and ensures statistical invariance under unitary transformations. Any random variable $\Lambda \sim \nu$ with this property can be written as $\Lambda = U \Sigma U^{\dagger}$, where $U$ is Haar-distributed and $\Sigma$ follows an independent matrix-valued distribution $\Delta$ \cite{Collins2006}. In the study of random quantum states, it is common to assume that $\Sigma$ is distributed over diagonal matrices~\cite{Zyczkowski1998, Zyczkowski2001}, leading to independent distributions for eigenvalues and eigenvectors.  

However, exact unitary invariance is often a strong condition, and in many practical scenarios, one encounters distributions that are only approximately invariant. Analogously to $t$-designs, which serve as approximate substitutes for Haar-random unitaries, we define $t$-invariant distributions, which satisfy a weaker form of unitary invariance. Specifically, a distribution $\nu$ is $t$-invariant if, for any polynomial $f$ of degree at most $t$ and any unitary $U$,  
\begin{equation}
    \bEx{\Lambda\sim\nu}\qty[f(\Lambda)] =
    \bEx{\Lambda\sim\nu}\qty[f(U \Lambda U^{\dag})],
\end{equation} 
in particular, this implies that we can compute
\begin{equation}
    \bEx{\Lambda\sim\nu}\qty[f(\Lambda)] =
    \bEx{\Lambda\sim\nu}\qty[\bEx{U\sim\mu_H}\qty[f(U \Lambda U^{\dag})]].
\end{equation}  
For pure states, $t$-invariant distributions correspond to those induced by unitary $t$-designs. Moreover, if $\nu_1$ and $\nu_2$ are $t$-invariant distributions over quantum states, then the induced distribution $\nu_3$ of their difference $\Lambda = \rho - \sigma$, where $\rho \sim \nu_1$ and $\sigma \sim \nu_2$, is also $t$-invariant.

In this section, using standard Haar integration techniques \cite{Mele2024}, we compute expectation values of the form $\bE[T(\Lambda)]$, $\bE[T(\Lambda)^2]$, $\operatorname{Var}[T(\Lambda)] = \bE[T(\Lambda)^2] - \bE[T(\Lambda)]^2$, and $\bE[\norm{T(\Lambda)}_2^2]$, where the result for the first quantity is valid for $\Lambda$ with a 1-invariant distribution and the others for $\Lambda$ with a 2-invariant distribution.

For $T(\Lambda)$, computing first the expectation value for the unitary part we obtain:
\begin{equation}
    \mathbb{E}_{\Lambda\sim\nu}[\mathbb{E}_{U\sim\mu_H}[T(U\Lambda U^{\dag})]] 
    = \mathbb{E}_{\Lambda\sim\nu}[\Tr\Lambda]\; T\qty(\frac{\mathbb{I}}{d}) 
    = \mathbb{E}_{\Lambda\sim\nu}[\Tr\Lambda]\; \pi,
\end{equation}
where we relied on linearity of $T$ and $\mathbb{E}_{U\sim\mu_H}[U \Lambda U^{\dag}] = \Tr \Lambda \; \frac{\mathbb{I}}{d}$, and named $\pi = T\qty(\frac{\mathbb{I}}{d})$.

To compute the expectation value of $T(\Lambda)^2$ we use the swap trick $T(\Lambda)^2 = \Tr_2 (T(\Lambda)\otimes T(\Lambda) \mathbb{F})$, where $\mathbb{F}$ is the flip operator between both terms of the tensor product and $\Tr_2$ is the partial trace over the second term. Thus, we first obtain
\begin{align}
\mathbb{E}_{U\sim \mu_H} \qty[T(U\Lambda U^{\dag})\otimes T(U\Lambda U^{\dag})]
&= 
(T\otimes T) \qty(\mathbb{E}_{U\sim \mu_H} \qty[(U\otimes U) \Lambda \otimes \Lambda (U^{\dag}\otimes U^{\dag})]) \\
&= (T\otimes T) \qty( 
\frac{1}{d^2-1}\qty{
\qty[(\Tr\Lambda)^2 - d^{-1} \Tr\Lambda^2] \mathbb{I}
+\qty[\Tr\Lambda^2 - d^{-1} (\Tr\Lambda)^2] \mathbb{F}
}
) \\
&=
\frac{d^2}{d^2-1}\left\lbrace
\qty[(\Tr\Lambda)^2 - d^{-1} \Tr\Lambda^2] \; \pi\otimes\pi \right.\\
&\left.+\qty[\Tr\Lambda^2 - d^{-1} (\Tr\Lambda)^2] (T\otimes T)\qty(\frac{\mathbb{F}}{d^2})
\right\rbrace,
\end{align}
where we used the linearity of $T\otimes T$ and the second moments of the Haar distribution
\begin{equation}
\mathbb{E}_{U\sim \mu_H} \qty[ U\otimes U A  U^{\dag}\otimes U^{\dag}] = c_{\mathbb{I}_{d},A} \mathbb{I}_{d}+ c_{\mathbb{F}_{d},A} \mathbb{F}_{d},
\end{equation}
with
\begin{align}
c_{\mathbb{I}_{d},A} &= \frac{1}{d^2-1}\qty[\Tr A - d^{-1} \Tr A\mathbb{F}_{d}],\\
c_{\mathbb{F}_{d},A} &= \frac{1}{d^2-1}\qty[\Tr A\mathbb{F}_{d} - d^{-1} \Tr A].
\end{align}
Calculating $\Tr_2[\pi\otimes\pi \mathbb{F}] = \pi^2$ and $\Tr_2[(T\otimes T)\qty(\frac{\mathbb{F}}{d^2}) \mathbb{F}] = \Tr_2 \tau^2$, we eventually obtain
\begin{equation}
    \mathbb{E}_{U\sim\mu_H} [T(U \Lambda U^{\dag})^2] = 
\frac{d^2}{d^2-1}\qty(
\qty[(\Tr\Lambda)^2 - d^{-1} \Tr\Lambda^2] \; \pi^2
+\qty[\Tr\Lambda^2 - d^{-1} (\Tr\Lambda)^2] \Tr_2 \tau^2
).
\end{equation}

Lastly, as $\norm{A}_2^2= \Tr A^2$ and the trace is a linear mapping, we have
\begin{equation}
    \mathbb{E}[\norm{T(\Lambda)}_2^2] = \Tr \mathbb{E}[T(\Lambda)^2].
\end{equation}

Let us consider specific distributions for the variable $\Lambda$. Below, we assume that the states $\rho$ and $\sigma$ have 2-invariant distributions on $\cS(\cH)$. Additionally, we assume $\mathbb{E}[\Tr\rho^2] = \frac{d + r}{dr + 1}$, which can be used to compute the terms in $\mathbb{E}[\Tr\Lambda^2]$. The integer parameter $r$ characterizes the rank of the states. Specifically, $r = 1$ means the distributions are concentrated on pure states, while $r = d$ corresponds to full-rank states. This choice for the average purity provides a convenient parametrization of mixedness, where $r$ controls the transition from pure to full-rank states. Moreover, this choice makes the moments computed below coincide with those of distributions induced by taking Haar distributions over $\cH_d \otimes \cH_r$ and discarding the second subsystem. In particular, when $r = d$, the distribution coincides with the Hilbert-Schmidt distribution.

We now apply the general formula to three common cases considered throughout this work, noting that the results for the second and third cases can be directly expressed in terms of the first.

First, let $\Lambda = \rho$, where $\rho$ follows a 2-invariant distribution with an average purity of $\mathbb{E}[\Tr\rho^2] = \omega$. In this case, we obtain the following expectation values:
\begin{align}
    \mathbb{E}[T(\rho)] &= \pi, \\
    \mathbb{E}[T(\rho)^2] &= \frac{d^2}{d^2-1}\qty[\qty(1-\frac{\omega}{d}) \pi^2 + \qty(\omega-\frac{1}{d})\Tr_2 \tau^2], \\
    \operatorname{Var}[T(\rho)] &= \frac{d^2}{d^2-1}\qty(\omega-\frac{1}{d})\qty[\Tr_2 \tau^2 - \frac{1}{d} \pi^2],\\
    \mathbb{E}[\norm{T(\rho)}_2^2] 
    =  \mathbb{E}[\Tr T(\rho)^2]
    &= \frac{d^2}{d^2-1}\qty[\qty(1-\frac{\omega}{d}) \Tr \pi^2 + \qty(\omega-\frac{1}{d})\Tr \tau^2].
\end{align}

Next, consider the case $\Lambda = \rho - \sigma$, where $\rho$ and $\sigma$ are independent, potentially different, 2-invariant distributions with the same average purity $\mathbb{E}[\Tr\rho^2] = \mathbb{E}[\Tr\sigma^2] = \omega$. Using the independence of $\rho$ and $\sigma$, which allows us to use $\operatorname{Var}[A+B] = \operatorname{Var}[A] + \operatorname{Var}[B]$ for independent variables, and the fact that the first two moments depend only on the average purity and trace, we find:
\begin{align}
    \mathbb{E}[T(\rho)-T(\sigma)] &= 0, \\
    \mathbb{E}[(T(\rho)-T(\sigma))^2] = \operatorname{Var}[T(\rho)-T(\sigma)] &= 
    \operatorname{Var}[T(\rho)]+ \operatorname{Var}[T(\sigma)] = 2 \operatorname{Var}[T(\rho)], \\
    \mathbb{E}[\norm{T(\rho)-T(\sigma)}_2^2] &= 2\Tr[\operatorname{Var}[T(\rho)]]. \label{eq:expval_HS_dist2}
\end{align}

Finally, consider the case $\Lambda = \rho - \frac{\mathbb{I}}{d}$, where $\rho$ follows a 2-invariant distribution with $\mathbb{E}[\Tr\rho^2] = \omega$. Using the same properties of the variance as above, we obtain:
\begin{align}
    \mathbb{E}[T(\rho)-\pi] &= 0, \\
    \mathbb{E}[(T(\rho)-\pi)^2] = 
    \operatorname{Var}[T(\rho)-\pi] &= \operatorname{Var}[T(\rho)], \\
    \mathbb{E}[\norm{T(\rho)-\pi}_2^2] &= \Tr[\operatorname{Var}[T(\rho)]].
\end{align}

As an example of how useful these calculations can be, we exactly compute the second moment of the 2-norm contraction. Note that although the 2-norm does not generally satisfy DPI \cite{PerezGarcia2006}, the following upper bound is indeed lower than one by $\Tr\pi^2 \geq \frac{1}{d} \Tr \tau^2$ (Fact~\ref{fact:purities}).

\begin{lemma}\label{lemma:2normcontract}  
Let $T:\cS(\cH_d)\to\cS(\cH_{d'})$ be a quantum channel, and let $\rho, \sigma \in \cS(\cH_d)$ be two density matrices distributed according to the measure $\nudr$. Then,  
\begin{equation}  
\bEx{\rho,\sigma\sim\nudr} \qty[ \frac{\norm{T(\rho)-T(\sigma)}_2^2}{\norm{\rho-\sigma}_2^2} ]  
= \frac{d}{d^2-1} \qty[d \Tr \tau^2 - \Tr \pi^2],  
\end{equation}  
where $\tau$ is the Choi state of $T$, and $\pi = T\qty(\frac{\mathbb{I}}{d})$ is the average output state.  
\end{lemma}

\begin{proof}
We use results from this Appendix together with $X=\rho-\sigma$ and the states having distribution $\nudr$. $X$ has an induced distribution $\nu$ that is unitarily invariant and can be regarded as a product of two independent distributions $\nu=\mu_H\times\Delta$, where $\mu_H$ is the Haar distribution on unitary matrices of size $d$ and $\Delta$ is a distribution on traceless Hermitian matrices of size $d$~\cite[Chapter 10]{Aubrun2017}. Here, the underlying spaces are identified through the spectral decomposition $X=U \Lambda U^{\dag}$, implying that $X\sim\nu$ is equivalent to $U\sim\mu_H$ and $\Lambda\sim \Delta$. Then we can compute the expectation value in two steps:
\begin{align}
    \bEx{\rho,\sigma\sim\nudr}\qty[ \frac{\norm{T(\rho)-T(\sigma)}_2^2}{\norm{\rho-\sigma}_2^2}]
    &=
    \bEx{X\sim\nu}\qty[ \frac{\norm{T(X)}_2^2}{\norm{X}_2^2}] \\
    &=\label{eq:exp_val_spectral_dec}
    \bEx{\Lambda\sim\Delta}\qty[ \frac{1}{{\Tr\Lambda^2}}\bEx{U\sim\mu_H} \qty[\norm{T(U\Lambda U^{\dag})}_2^2]],
\end{align}
where we used the invariance of the $2$-norm under unitary conjugations to take the denominator out of the first expectation value and that ${\norm{\Lambda}_2^2}=\Tr\Lambda^2$.

Calculating the expectation value of the unitary part,
\begin{equation}
    \bEx{U\sim\mu_H} \qty[\norm{T(U\Lambda U^{\dag})}_2^2] = \frac{d}{d^2-1}\qty[d \Tr \tau^2 - \Tr \pi^2] \; \Tr\Lambda^2.
\end{equation}
The distribution $(\rho,\sigma)\sim\nudr\times\nudr$ is concentrated in $\rho\neq \sigma$ which implies that $X\neq 0$ almost surely and consequently $\Tr\Lambda^2 > 0$ almost surely. Finally, we can introduce the latter result in the expression of Eq.~\eqref{eq:exp_val_spectral_dec}, which completes the proof.
\end{proof}

In the particular case where the distributions are concentrated on pure states, we recover state $t$-designs (Definition~\ref{def:t_designs}), and we can additionally state the following well-known facts.

\begin{lemma}\label{lemma:2design_collision}
	Let $\nu$ be a finite uniform 2-design distribution on $\cH_d$ and let $\rho,\sigma\sim\nu$ be two independent and identically distributed states. Then,
	\begin{equation}
		\Pr[\rho = \sigma] \leq \frac{1}{d^4-2d^2+2} = \order{d^{-4}}.
	\end{equation}
\end{lemma}

\begin{proof}
	If the distribution is continuous and uniform, then $\rho \neq \sigma$ a.s. and the inequality is trivially true. 
	
	Let $\{U_i\}_{i=1}^N$ be the set of different unitaries forming a 2-design which together with some fixed state $\ket{\phi_0}$ generate the states $\ket{\phi_i}=U_i\ket{\phi_0}$ which have a uniform probability weight $1/N$. As the unitaries are different, the collision $\ket{\phi_i}=\ket{\phi_j}$ can only happen if $i=j$, hence
	\begin{equation}
		\Pr[\rho = \sigma] = \frac{1}{N}.
	\end{equation}
	The lemma then follows from the well-known fact that any discrete 2-design must have at least $N\geq d^4-2d^2+2$ elements \cite{Gross2007, Roy2009}.
\end{proof}

\begin{lemma}[Tail bound trace distance in 1-designs]\label{lemma:concentrationoftrdist}
    Let $\nu$ be a 1-design over pure states in $\cS(\cH_d)$, and let $\rho$ and $\sigma$ be two states independently distributed according to $\nu$. Then,
    \begin{equation}
        \Pr[\frac{1}{2}\norm{\rho-\sigma}_1 \leq 1 - \delta] \leq \frac{1}{d \, \delta}
        \quad
        \text{for all} \quad 0< \delta < 1.
    \end{equation}
\end{lemma}

\begin{proof}
	Markov's inequality states that for $X\geq 0$ an arbitrary positive variable
    \begin{equation}
        \Pr[X \geq \delta] \leq \frac{\bE[X]}{\delta}.
    \end{equation}
    We set $X=1 - \frac{1}{2}\norm{\rho-\sigma}_1$ and observe that for pure states we have that
    \begin{equation}
      \frac{1}{2}\norm{\rho-\sigma}_1 = \sup_{0\leq P \leq \mathbb{I}} \Tr P(\rho-\sigma) \geq 1 - \Tr \sigma \rho,
    \end{equation}
    where the lower bound follows from $P=\rho$. Now we use that the distributions for both states are independent 1-designs implying $\mathbb{E}[\rho]=\mathbb{E}[\sigma]=\frac{1}{d}\mathbb{I}$, thus
    \begin{equation}
        \bEx{\rho,\sigma\sim\nu}\qty[\frac{1}{2}\norm{\rho-\sigma}_1] \geq 1-\frac{1}{d}.
    \end{equation}
\end{proof}

\section{Lower bound for average contraction for random circuits}\label{sec:random_circuits}
In this Appendix, we will prove the results of Sec.~\ref{sec:quantum_circuits}
Let us consider an $n$-qubit circuit $\cC$ implemented as a sequence of unitaries $U_i$ such that $U_{\cC}=U_D\cdots U_1$, and an $n$-qubit quantum channel $\Phi$ with Kraus decomposition $\qty{E_x}_{x\in\cX}$. Assume $\cC'$, the noisy implementation of $\cC$ is given by intercalating $\Phi$ in each unitary layer resulting in the quantum channel
\begin{equation}
    \cE_{\cC'} = \Phi \circ \cU_D \circ \Phi \circ \cdots \circ \Phi \circ \cU_2 \circ \Phi \circ \cU_1,
\end{equation}
where $\cU(\cdot) = U \cdot U^{\dag}$ denotes unitary conjugation.

Let $\ket{\psi}$ and $\ket{\phi}$ be two independent random states drawn from a 1-design. And let the unitary layers be drawn from a 1-design.

Let us assume the channel is unital, $\Phi(\mathbb{I})=\mathbb{I}$, which also implies $\cE_{\cC'}(\mathbb{I})=\mathbb{I}$.

\begin{theorem}[Average contraction for random circuits]
Under the assumptions above and letting $\mu_1$ be a 1-design, for any $\varepsilon\in (0, 1)$ and $\delta > 0$, there is a large enough $D$ such that
\begin{equation}
\mathbb{E}_{\psi,\phi, \cC}\qty[
    \frac{\norm{\cE_{\cC'}(\ketbra{\psi}) - \cE_{\cC'}(\ketbra{\phi})}_1}{\norm{\ketbra{\psi} - \ketbra{\phi}}_1}
    ]
    \geq 1 - \varepsilon - 
     \frac{1}{2^n} 2^{D (S(\tau)+\delta)}
    ,
\end{equation}
where $S(\tau)$ is the entropy of the Choi state of $\Phi$. Furthermore, we also have:
\begin{align}\label{equ:contraction_max_mixed_state}
    \mathbb{E}_{\psi, \cC}\qty[
        \frac{\norm{\cE_{\cC'}(\ketbra{\psi}) - \mathbb{I}/2^n}_1}{\norm{\ketbra{\psi} - \mathbb{I}/2^n}_1}
        ]
        \geq 1 - \varepsilon - 
         \frac{1}{2^n} 2^{D (S(\tau)+\delta)}
\end{align}   
\end{theorem}

\paragraph{Remark.} A necessary condition for this bound to be non-trivial is
\begin{equation}
    1 - \varepsilon - 
     \frac{1}{2^n} 2^{D (S(\tau)+\delta)}
     \geq 0
     \quad \Rightarrow \quad
     n \geq D S(\tau).
\end{equation}

\begin{proof}
    Consider the random variable $X\in \cX$ with probability distribution $\Pr(X=x)=p(x)=\frac{1}{2^n} \Tr E_x E_x^{\dag}$. And define the $\delta$-typical set of sequences of $X$ of length $D$:
    \begin{equation}
        T_{\delta}^{X^D} = \qty{x \in \cX^D \mid \abs{-\frac{1}{D} \log p(x) - H(X)} < \delta},
    \end{equation}
    where $H(X)$ is the entropy of the variable $X$. For large enough $D$ this set satisfies $\Pr[X^D \in T_{\delta}^{X^D}] \geq 1 - \varepsilon$ and $\abs{T_{\delta}^{X^D}} \leq 2^{D (H(X)+\delta)}$.

    Define $Q_{\psi}$ as the orthogonal projector onto the support of
    \begin{equation}
        \sum_{x \in T_{\delta}^{X^D}} E_{x_D} U_D \cdots E_{x_1} U_1 \ketbra{\psi} U_1^{\dag} E_{x_1}^{\dag} \cdots U_D^{\dag} E_{x_D}^{\dag},
    \end{equation}
    note that $Q_{\psi}$ also depends on the circuit $\cC$, the channel's Kraus representation and $\delta$, besides the input state $\ket{\psi}$. This projector inherits the following properties from the typical set, 
    \begin{equation}
        \Tr Q_{\psi} = \rank Q_{\psi} \leq \abs{T_{\delta}^{X^D}} \leq 2^{D(H(X)+\delta)},
    \end{equation}
    and
    \begin{equation}
        \Tr[Q_{\psi} \cE_{\cC'}(\ketbra{\psi})] \geq \Tr \sum_{x \in T_{\delta}^{X^D}} E_{x_D} U_D \cdots E_{x_1} U_1 \ketbra{\psi} U_1^{\dag} E_{x_1}^{\dag} \cdots U_D^{\dag} E_{x_D}^{\dag},
    \end{equation}
    hence,
    \begin{align}
        \bE_{\psi, \cC}\qty[\Tr[Q_{\psi} \cE_{\cC'}(\ketbra{\psi})]]
        &\geq 
        \bE_{\psi, \cC}\qty[\Tr \sum_{x \in T_{\delta}^{X^D}} E_{x_D} U_D \cdots E_{x_1} U_1 \ketbra{\psi} U_1^{\dag} E_{x_1}^{\dag} \cdots U_D^{\dag} E_{x_D}^{\dag}] \\
        &=
        \sum_{x \in T_{\delta}^{X^D}}
        \bE_{\cC}\qty[\Tr E_{x_D} U_D \cdots E_{x_1} U_1 \frac{\mathbb{I}}{2^n}U_1^{\dag} E_{x_1}^{\dag} \cdots U_D^{\dag} E_{x_D}^{\dag}] \\
        &=
        \sum_{x \in T_{\delta}^{X^D}}
        \frac{1}{2^{D n}}
        \Tr E_{x_D} E_{x_D}^{\dag} \cdots \Tr E_{x_1} E_{x_1} \\
        &= \sum_{x \in T_{\delta}^{X^D}} p(x_D\cdots x_1) = \Pr[X\in T_{\delta}^{X^D}] \geq 1- \varepsilon.
    \end{align}

    Let us choose the Kraus representation minimizing the entropy of $X$ such that it is equal to the entropy of the Choi state of $\Phi$, $H(X)=S(\tau)$.

    Finally,
    \begin{align}\label{eq:last_step_contraction_bound}
    \bE_{\psi,\phi, \cC}\qty[
    \frac{\norm{\cE_{\cC'}(\ketbra{\psi}) - \cE_{\cC'}(\ketbra{\phi})}_1}{\norm{\ketbra{\psi} - \ketbra{\phi}}_1}
    ]
    &\geq 
    \bE_{\psi,\phi, \cC}\qty[
        \frac{1}{2}
        \norm{\cE_{\cC'}(\ketbra{\psi}) - \cE_{\cC'}(\ketbra{\phi})}_1
    ] \\
    &\geq
    \bE_{\psi,\phi, \cC}\qty[
        \Tr [Q_{\psi} (\cE_{\cC'}(\ketbra{\psi}) - \cE_{\cC'}(\ketbra{\phi}))]
    ]\\
    &\geq
    1 - \varepsilon - 
     \frac{1}{2^n} 2^{D (S(\tau)+\delta)}.
    \end{align}
It is then clear that the same argument as in Eq.~\eqref{eq:last_step_contraction_bound} can be used to show Eq.~\eqref{equ:contraction_max_mixed_state}, as in Eq.~\eqref{eq:last_step_contraction_bound} we just replace $\ketbra{\phi}$ with $\mathbb{I}/2^n$.
\end{proof}
We can now specialize the result above to product channels.

\begin{corollary}[Constant depth, large number of qubits]

Let $\Phi = \phi^{\otimes n}$ be a product channel. For any $\varepsilon\in (0, 1)$ and $\delta > 0$, there is a large enough $n$ such that
\begin{equation}
    \bE_{\psi,\phi, \cC}\qty[
    \frac{\norm{\cE_{\cC'}(\ketbra{\psi}) - \cE_{\cC'}(\ketbra{\phi})}_1}{\norm{\ketbra{\psi} - \ketbra{\phi}}_1}
    ]
    \geq (1 - \varepsilon)^D - 
     2^{n[D (S(\tau)+\delta)-1]}
    ,
\end{equation}
where $S(\tau)$ is the entropy of the Choi state of $\phi$. In particular, if $D S(\tau) \leq 1$ then
\begin{equation}
    \lim_{n\rightarrow \infty}
    \bE_{\psi,\phi, \cC}\qty[
    \frac{\norm{\cE_{\cC'}(\ketbra{\psi}) - \cE_{\cC'}(\ketbra{\phi})}_1}{\norm{\ketbra{\psi} - \ketbra{\phi}}_1}
    ]
    = 1.
\end{equation}
\end{corollary}

\section{Relations between the Choi state and the image of the maximally mixed state}
\label{sec:tau_and_pi}

For a knowledgeable reader, it is likely unsurprising that both the Choi state and the image of the maximally mixed state play a central role in computing averages over unitarily invariant distributions. This section reviews the intuitive motivations behind this expectation and summarizes key properties relating a quantum channel, its Choi state, and the image of the maximally mixed state. This appendix aims to provide a concise reference for less experienced readers, collecting a handful of frequently used facts from the main text. We refer the reader to \cite{Wilde2017, Watrous2018} for a more comprehensive introduction.

In this section, we shall consider an arbitrary finite-dimensional quantum channel $T:\cS(\cH_d)\to\cS(\cH_{d'})$, we take $\tau\in\cS(\cH_{d'}\otimes\cH_d)$ and $\pi\in\cS(\cH_d)$ as the Choi state and the image of the maximally mixed state for the channel $T$:
\begin{equation}
    \tau:=\tau(T) =(T\otimes \text{id})(\ketbra{\Omega}) 
    \quad \text{with} \quad \ket{\Omega}=\frac{1}{\sqrt{d}}\sum_{i=1}^d \ket{i\, i},
    \quad \text{and} \quad
    \pi:=\pi(T) = T\qty(\frac{\mathbb{I}}{d}),
\end{equation}
From their definition, it is easy to see that
\begin{equation}
    \Tr_2 \tau = (T\otimes \Tr)(\ketbra{\Omega}) = T\qty(\frac{\mathbb{I}}{d}) = \pi.
\end{equation}

The first reason supporting the relevance of the Choi state is of a geometrical nature. The singular values of a linear map $T$ are defined as the square root of the eigenvalues of the positive semidefinite operator $T^* T: \cS(\cH_d)\to\cS(\cH_d)$, where $T^*$ denotes the dual map satisfying $\Tr(Y T(X))=\Tr (T^*(Y) X)$ \cite{Bhatia1996}. Since $T^* T$ is positive semidefinite, all its eigenvalues are non-negative, and we can write $s_k^2(T)=\lambda_k(T^* T)$, where $s_k(T)$ denotes the $k$-th singular value and $\lambda_k(T^*T)$ the $k$-th eigenvalue, both in decreasing order. Geometrically, singular values describe how a linear map stretches or contracts space. It is an elementary fact that under the action of $T$, a Euclidean sphere is transformed into an ellipsoid, with principal axes given by the singular values. In our context, the Euclidean norm corresponds to the 2-norm, making the mean squared singular value arise frequently in our analysis. Since $\sum_k \lambda_k (T^* T) = \Tr(T^* T)$, we can use the standard operator basis $\{\ketbra{i}{j}\}_{i,j=1}^d$, and find that
\begin{align}
    \frac{1}{d^2} \sum_{k=1}^{d^2} s_k(T)^2 &= \frac{1}{d^2} \sum_{i,j=1}^d \Tr[(\ketbra{i}{j})^{\dag} T^*(T(\ketbra{i}{j}))] = \frac{1}{d^2} \sum_{i,j=1}^d \Tr[T(\ketbra{j}{i})\; T(\ketbra{i}{j})] \\
    &=
    \frac{1}{d^2} \sum_{i,j,n,m=1}^d \Tr[(T(\ketbra{n}{m}) \otimes \ketbra{n}{m})\; (T(\ketbra{i}{j}) \otimes \ketbra{i}{j})] \\
    &= \Tr[(T\otimes \operatorname{id})(\ketbra{\Omega})\; (T\otimes \operatorname{id})(\ketbra{\Omega})]\\
    &= \Tr \tau^2.
\end{align}
That is, the purity of the Choi state is identical to the mean squared singular value of $T$:

\begin{fact}\label{fact:purity_to_singvals}
    Let $s_k(T)$ for $k=1,\dots, d^2$ be the singular values of $T$. Then,
    \begin{equation}
        \Tr \tau^2 = \frac{1}{d^2} \sum_{k=1}^{d^2} s_k(T)^2.
    \end{equation}
\end{fact}

Any completely positive linear map admits a Kraus decomposition
\begin{equation}
	T(\rho) = \sum_{x\in \mathcal{X}} E_x \rho E_x^\dag,
\end{equation}
where the operators $E_x :\cH_d \to \cH_{d'}$ are known as Kraus operators. Additionally, if the map $T$ is trace preserving, they satisfy the condition
\begin{equation}\label{eq:tp_condition}
	\sum_{x\in\mathcal{X}} E_x^\dag E_x = \mathbb{I}.
\end{equation}
In general, the Kraus representation is not unique, but we can find a canonical representation $\{A_k\}_{k=1}^K$ such that additionally
\begin{equation}
	\Tr[A_k^\dag A_l] = 0, \quad \text{if } k\neq l,
\end{equation}
and the number of Kraus operators $K$ is minimized, i.e. $K \leq \abs{\mathcal{X}}$ for any set of Kraus operators $\{E_x\}_{x\in\mathcal{X}}$ describing the channel. This can be proven using the spectral decomposition of the Choi state. Indeed, observe that the Choi state can be represented as a mixture of states 
\begin{equation}
	\tau = \sum_{x \in \mathcal{X}} p(x) \ketbra{\phi_x}, \quad \text{where } \ket{\phi_x} = \frac{1}{\sqrt{p(x)}} (E_x \otimes \mathbb{I})\ket{\Omega},
\end{equation}
with $p(x) = \frac{1}{d} \Tr(E_x^\dagger E_x)$ a probability distribution over the index set $\mathcal{X}$. In particular, choosing a canonical representation leads to a spectral decomposition of the Choi state with $\sqrt{\omega_k} \ket{\psi_k} =  (A_k \otimes \mathbb{I})\ket{\Omega}$ and
\begin{equation}\label{eq:tau_decomp}
    \tau = \sum_{k=1}^K \omega_k \ketbra{\psi_k},
    \quad \text{with} \quad \omega_k > 0, \; \sum_{k=1}^K\omega_k=1, \; \text{and} \; \braket{\psi_k}{\psi_l}=\delta_{kl},
\end{equation}
For this reason, $K$ is known as the Choi-rank or Kraus-rank of the channel and satisfies $1\leq K \leq d d'$.

The fact that all the terms in the left hand side of Eq.~\eqref{eq:tp_condition} are positive semidefinite implies that $E_x^\dag E_x \leq \mathbb{I}$, thus their eigenvalues are upper bounded by $1$, additionally, we have that $\rank(E_x)\leq \min\{d, d'\}$, in turn,
\begin{equation}\label{eq:d2dp_choi_spec}
    p(x) = \frac{1}{d} \Tr(E_x^\dag E_x) \leq \min\qty{1, \frac{d'}{d}}.
\end{equation}
In particular, $\omega_k \leq \min\{1, d'/d\}$. Introducing the states $\rho_k = \Tr_2 \ketbra{\psi_k}$ we find a related expression
\begin{equation}\label{eq:pi_decomp}
    \pi = \sum_{k=1}^K \omega_k \rho_k.
\end{equation}

Let $X\in\mathcal{X}$ be a random variable with distribution $p(x)$, then the Shannon entropy of $X$ is given by $H(X)=-\sum_{x\in\mathcal{X}} p(x)\log p(x)$ and it is an upper bound for the von Neumann entropy $S(\rho)=-\Tr[\rho\log\rho]$ of any mixture (Theorem 11.10 in \cite{NielsenChuang2010})
\begin{equation}
	S\qty(\sum_{x\in\cX} p(x) \rho_x) \leq H(X) + \sum_{x\in\cX} p(x) S(\rho_x).
\end{equation}
The bound is saturated when the states $\rho_x$ are pure and orthogonal. We get the following fact as a direct consequence.

\begin{fact}\label{fact:entropy_bound}
    Let $X$ be a random variable with probability distribution $p(x)$ determined by a Kraus representation of the channel $T$ as above. Then, its Shannon entropy is lower bounded by the von Neumann entropy of the Choi state,
    \begin{equation}
        S(\tau) \leq H(X).
    \end{equation}
    The bound is saturated if the Kraus representation is canonical.
\end{fact}

The rank, purity, and entropy of the Choi state serve as coarse indicators of noise in quantum channels. For purely unitary evolutions, the Choi state has minimal rank and entropy, $K = 1$ and $S(\tau)=0$, and maximal purity, $\Tr \tau^2 = 1$. In contrast, a replacer channel that maps every input to the maximally mixed state, $T(X) = \Tr X, \frac{\mathbb{I}}{d'}$, attains the opposite extremes: $K = dd'$, $S(\tau)=\log(dd')$ and $\Tr \tau^2 = \frac{1}{dd'}$. While the rank gives a reasonable description in such limits, it is highly sensitive to small eigenvalues and thus less informative for channel families with inherently high Choi rank, such as a global depolarizing channel. In comparison, purity and entropy are smoother quantities with a more direct physical interpretation. However, they cannot be relied on alone, as they are affected by the unitality of the channel. For instance, a replacer channel that maps every state to a fixed pure state has a Choi state with the same purity and entropy as that of a completely dephasing channel. Yet the former destroys all input information, whereas the latter preserves classical probability distributions. Indeed, observe that
\begin{equation}\label{eq:choi_rep_deph}
	\tau_{\text{rep.}} = \ketbra{\psi} \otimes \frac{\mathbb{I}}{d},
	\quad \text{and} \quad
	\tau_{\text{deph}} = \frac{1}{d} \sum_{i=1}^d \ketbra{ii}.
\end{equation}
Besides these limitations, in our work, the purity and entropy of the Choi state are instrumental for upper and lower bounding the average contraction.

To establish the valid regime for these purities, it is useful to introduce the unital channel $\mathcal{E}$ acting on bipartite states in $\cS(\cH_{d'} \otimes \cH_d)$ by replacing the second subsystem with the maximally mixed state:
\begin{equation}\label{eq:def_ptr_and_replace}
	\mathcal{E}(\rho) = \Tr_2 \rho \otimes \frac{\mathbb{I}}{d}.
\end{equation}
Acting on the Choi state, it results in
\begin{equation}
	\mathcal{E}(\tau) = \pi \otimes \frac{\mathbb{I}}{d}.
\end{equation}

\begin{fact}\label{fact:purities}
    The purities of $\tau$ and $\pi$ satisfy
    \begin{equation}
        \frac{1}{d'} \leq \Tr\pi^2 \leq 1,
        \quad \text{and} \quad
        \frac{1}{d d'} \leq \Tr\tau^2 \leq 1,
    \end{equation}
    moreover, are mutually bounded by the following inequalities,
    \begin{equation}
        \frac{1}{d} \Tr\pi^2 \leq \Tr\tau^2 \leq d \Tr\pi^2,
        \quad \text{and} \quad
        \Tr\tau^2 \leq \frac{d'}{d}.
    \end{equation}
\end{fact}

\begin{proof}
    The first set of inequalities follows from standard bounds on the purity of quantum states. Indeed, the purity $\Tr\rho^2$ of any $d$-dimensional quantum state $\rho$ satisfies $1/d \leq \Tr\rho^2 \leq 1$. 

    For $\Tr\tau^2\leq d'/d$, use the bound for the eigenvalues of $\tau$ in Eq.~\eqref{eq:d2dp_choi_spec},
    \begin{equation}
        \Tr\tau^2 = \sum_{k=1}^K \omega_k^2 \leq \frac{d'}{d} \sum_{k=1}^K \omega_k = \frac{d'}{d}.
    \end{equation}
    
    For $\Tr\tau^2\leq d \Tr\pi^2$, express the purity of $\pi$ as
    \begin{align}
        \Tr \pi^2 &= \sum_{k,l} \omega_k \omega_l \Tr \rho_k\rho_l = 
        \sum_{k=1}^{K} \omega_k^2 \Tr \rho_k^2 +  \sum_{k\neq l} \omega_k \omega_l \Tr \rho_k\rho_l\\
        &\geq  \sum_{k=1}^{K} \omega_k^2 \frac{1}{d} = \frac{1}{d} \Tr \tau^2.
    \end{align}
    For the last step, we noted that since $\rho_k\geq 0$ are positive semi-definite, they have a positive square root and $\Tr\rho_k\rho_l=\Tr\sqrt{\rho_k}\rho_k\sqrt{\rho_l}\geq 0$. This shows that the second sum results in a positive term. Additionally, the state $\rho_k$ comes from the partial trace of a pure state $\ket{\psi_k}\in\cH_{d'}\otimes \cH_d$, so its rank is upper bounded by the smaller dimension of the subsystems $\rank(\rho_k)\leq \min(d, d') \leq d$, thus, $\Tr\rho_k^2 \geq 1/ \rank(\rho_k) \geq 1/d$.

    For $\Tr\tau^2\geq \Tr\pi^2/d$, consider the channel $\mathcal{E}$ defined in Eq.~\eqref{eq:def_ptr_and_replace}. Since $\mathcal{E}$ is unital, it does not increase purity, i.e. $\Tr \mathcal{E}(\rho)^2 \leq \Tr \rho^2$. Acting on the Choi state, we obtain
    \begin{equation}
        \mathcal{E}(\tau) = \pi \otimes \frac{\mathbb{I}}{d} \quad \Rightarrow \quad \Tr\tau^2 \geq \frac{1}{d} \Tr\pi^2.
    \end{equation}
\end{proof}

When computing the expectation values in App.~\ref{sec:expvals_unit_invar}, we obtained,
\begin{equation}
	\operatorname{Var}[T(\rho)] \propto \Tr_2\tau^2 - \frac{1}{d} \pi^2,
\end{equation}
As this is the expectation value of a positive semidefinite variable, we must have that $d\Tr_2\tau^2 \geq \pi^2$. Here, we provide an alternative proof for this fact.

\begin{fact}
	The matrices $\Tr_2\tau^2$ and $\pi^2$ satisfy
	\begin{equation}
		d \Tr_2\tau^2 \geq \pi^2.
	\end{equation}
\end{fact}

\begin{proof}
    Let $\Phi$ be a unital completely positive map, then by Kadison-Schwarz inequality (see, for instance, Prop. 3.3 in \cite{Paulsen2003}), we have that for all normal $A$
    \begin{equation}
    	\Phi(A^\dag A) \geq \Phi(A) \Phi(A^\dag).
    \end{equation}
    Applying this to the channel $\cE$ defined in Eq.~\eqref{eq:def_ptr_and_replace} with $A=\tau$:
    \begin{equation}
    	\Tr_2\tau^2 \otimes \frac{\mathbb{I}}{d}
    	\geq (\Tr_2\tau)^2 \otimes \frac{\mathbb{I}}{d^2}
    	= \pi^2 \otimes \frac{\mathbb{I}}{d^2}.
    \end{equation}
    Noticing that $A \otimes \mathbb{I} \geq 0$ implies $A \geq 0$ and rearranging the terms finishes the proof.
\end{proof}

\section{Supplementary material for Sec.~\ref{sec:bound_for_dfs}}

\subsection{Pinsker-type inequalities}\label{sec:pinsker_ineqs}

We provide the proofs for the Pinsker-type inequalities in Prop.~\ref{prop:reverse_pinsker} and \ref{prop:pinsker}.

\begin{proposition}[Reverse Pinsker-type inequality, Prop. 5.2. in \cite{Hirche23}]
    Let $f:(0,\infty)\rightarrow \mathbb{R}$ be a twice differentiable convex function with $f(1)=0$ and $D_f$ be its associated $f$-divergence, we have
    \begin{equation}
            \begin{split}
            D_f(\rho \| \sigma) &\leq 
            \qty(\frac{f(e^{D_{\max}(\rho\|\sigma)})}{e^{D_{\max}(\rho\|\sigma)}-1} + \frac{e^{D_{\max}(\sigma\|\rho)} f(e^{-D_{\max}(\sigma\|\rho)})}{e^{D_{\max}(\sigma\|\rho)}-1})
             \, \frac{1}{2}\norm{\rho - \sigma}_1 \\
            &\leq 
            K_f\qty(e^{\Xi(\rho\|\sigma)}) \, \frac{1}{2}\norm{\rho - \sigma}_1,
        \quad
        \text{with}
        \quad
        K_f (x) = \frac{f \qty(x) + x f \qty(x^{-1})}{x - 1}.
        \end{split}
    \end{equation}
\end{proposition}

\begin{proof}
    Here, we solely evaluated the integrals of the first inequality in Prop. 5.2. \cite{Hirche23} explicitly. The original inequality reads
    \begin{equation}
        D_f(\rho \| \sigma) \leq \qty(\int_1^{\alpha} \frac{\alpha-\gamma}{\alpha-1} f''(\gamma) \dd \gamma + \int_1^{\beta} \frac{\beta-\gamma}{\beta-1} \gamma^{-3}f''(\gamma) \dd \gamma ) \frac{1}{2}\norm{\rho - \sigma}_1,
    \end{equation}
    with $\alpha = e^{D_{\max}(\rho\|\sigma)}$ and $\beta = e^{D_{\max}(\sigma\|\rho)}$, and we applied integration by parts
    \begin{align}
        \int_1^{\alpha} \frac{\alpha-\gamma}{\alpha-1} f''(\gamma) \dd \gamma
        &=
        \frac{f(\alpha)}{\alpha-1} - f'(1),\\
        \int_1^{\beta} \frac{\beta-\gamma}{\beta-1} \gamma^{-3}f''(\gamma) \dd \gamma
        &=
        \int_1^{\beta} \frac{\beta-\gamma}{\beta-1} \dv[2]{\gamma}(\gamma f(\gamma^{-1})) \dd \gamma\\
        &=\frac{\beta f(\beta^{-1})}{\beta - 1} + f'(1).
    \end{align}
\end{proof}

\begin{remark}
	Note that the coefficients of the bound are non-decreasing functions of $D_{\max}(\rho\|\sigma)$ or $D_{\max}(\sigma\|\rho)$, respectively. Indeed, as $f$ is convex, we have $f''\geq 0$ and both $f(x)/(x-1)$ and $x f(x)/(x-1)$ are definite integrals of positive functions with $x>1$ the upper limit of the integration, hence both expressions are non-decreasing with $x$.
\end{remark}

We prove a slightly more general form of Prop.~\ref{prop:pinsker}, relaxing the assumption that the quantum $f$-divergence admits an integral representation. We say a quantum $f$-divergence is \emph{consistent with the classical} if it agrees with the classical $f$-divergence on commuting states, via their spectral distributions, and satisfies DPI. The integral representation has these two properties \cite{Hirche23}. These two properties suffice to lower-bound the quantum $f$-divergence by the classical one evaluated on the distribution induced by any POVM. This makes the measured $f$-divergence particularly relevant \cite{Hirche23},
\begin{equation}\label{eq:measured_fdiv}
	\check{D}_f(\rho \| \sigma):= \sup_{M} D_f(P_{M,\rho} \| P_{M,\sigma}),
\end{equation}
where the supremum is taken over all discrete measurements $M$, and $P_{M,\rho}$ denotes the resulting probability distribution. Therefore, we have $\check{D}_f \leq D_f$ for any quantum generalization that is consistent with the classical, which allows us to apply classical results to obtain the following.

\begin{proposition}[Pinsker-type inequality]
    Let $f:(0,\infty)\rightarrow \mathbb{R}$ be a convex function differentiable up to order 3 at $u=1$ with $f''(1)>0$, and let $D_f$ be a quantum $f$-divergence consistent with he classical. Then, we have
    \begin{equation}
        D_f(\rho\|\sigma) \, \geq \frac{f''(1)}{2}\, \qty(\frac{1}{2}\norm{\rho-\sigma}_1)^2.
    \end{equation}
\end{proposition}

\begin{proof}
    The proof follows from an analogous inequality for the classical $f$-divergence, which we extend to the quantum setting using measured $f$-divergences.     
    
    For the classical $f$-divergence as defined in Eq.~\eqref{eq:df_classical}, from Thm.~3 in~\cite{Gilardoni10}, we have the bound
    \begin{equation} \label{eq:pinsker_ineq_classical}
        D_f(P \| Q) \, \geq \frac{f''(1)}{2} \operatorname{TV}^2(P \| Q),
    \end{equation}
    where $\operatorname{TV}^2(P \| Q)$ is the total variation distance between $P$ and $Q$,
    \begin{equation}
        \operatorname{TV}(P\| Q) = \frac{1}{2} \sum_{x \in \cX} \abs{P(x) - Q(x)}.
    \end{equation}
    
    Applying Eq.~\eqref{eq:pinsker_ineq_classical} to the outcomes of measurement $M$ on $\rho$ and $\sigma$, we obtain
    \begin{equation}
        \frac{f''(1)}{2}  \operatorname{TV}^2(P_{M,\rho} \| P_{M,\sigma}) \leq \, D_f(P_{M,\rho}\| P_{M,\sigma}),
    \end{equation}
    and including the variational characterization of the trace distance,
    \begin{equation}
        \frac{1}{2}\norm{\rho - \sigma}_1 = \sup_{M} \operatorname{TV}(P_{M,\rho} \| P_{M,\sigma}),
    \end{equation}
    we obtain
    \begin{align}
        \frac{f''(1)}{2}\, \qty(\frac{1}{2}\norm{\rho - \sigma}_1)^2
        &=\frac{f''(1)}{2}\,  \sup_{M}\, \operatorname{TV}^2(P_{M,\rho} \| P_{M,\sigma}) \\
        &\leq \sup_{M} D_f(P_{M,\rho}\| P_{M,\sigma})
        = \check{D}_f(\rho \| \sigma) \\
        &\leq D_f(\rho \| \sigma).
    \end{align} 
\end{proof}

\subsection{Average contraction of the \texorpdfstring{$\chi^2$}{Chi\^{}2}-divergence}\label{sec:supmat_chi2}

An equivalent expression for the $\chi^2$-divergence by integral definition can be found in \cite{Hirche23}:
\begin{equation}
    D_{x^2}(\rho \| \sigma) = \int_0^{\infty} \Tr[(\sigma + s \mathbb{I})^{-1} \rho (\sigma + s \mathbb{I})^{-1} \rho] \dd s - 1.
\end{equation}
We can compute the integral resorting to the identity $\Tr[AB]=\Tr[A\otimes B \, \mathbb{F}]$. Firstly, let $\sigma = \sum_\lambda \lambda \Pi_\lambda$ be a spectral decomposition for $\sigma$ with eigenvalues $\lambda$ and eigenprojectors $\Pi_\lambda$,
\begin{align}
    \kappa &= \int_0^{\infty} (\sigma + s \mathbb{I})^{-1} \otimes (\sigma + s \mathbb{I})^{-1}  \\
           &= \sum_{\lambda_1, \lambda_2} \int_0^{\infty} \frac{1}{(\lambda_1 + s) (\lambda_2 + s)}  \Pi_{\lambda_1} \otimes \Pi_{\lambda_2} \\
           &= \sum_{\lambda}  \frac{1}{\lambda} \Pi_{\lambda} \otimes \Pi_{\lambda}
           + \sum_{\lambda_1 \neq \lambda_2} \frac{\ln(\lambda_1 / \lambda_2)}{\lambda_1  - \lambda_2}  \Pi_{\lambda_1} \otimes \Pi_{\lambda_2}.
\end{align}
Then, $ D_{x^2}(\rho \| \sigma)= \Tr[(\rho \otimes \rho) \kappa \mathbb{F}] - 1$, that is
\begin{equation}
    D_{x^2}(\rho \| \sigma) = \sum_{\lambda}  \frac{1}{\lambda} \Tr[(\rho \Pi_\lambda)^2]
           + \sum_{\lambda_1 \neq \lambda_2} \frac{\ln(\lambda_1 / \lambda_2)}{\lambda_1  - \lambda_2}  \Tr[\rho\Pi_{\lambda_1} \rho \Pi_{\lambda_2}] - 1.
\end{equation}
We can compute each term for an arbitrary channel $T$ and state 2-design $\nu$ 
\begin{align}
    \bEx{\rho\sim\nu}[\Tr[T(\rho) \Pi_{\lambda_1} T(\rho) \Pi_{\lambda_2}] &=
    \Tr[\frac{d}{d+1}\qty(\pi \otimes \pi + (T\otimes T)(\frac{\mathbb{F}}{d^2})) \, \Pi_{\lambda_1} \otimes \Pi_{\lambda_2} \, \mathbb{F}]
    \\ 
    &=
    \frac{d}{d+1}\left( \Tr[\pi \Pi_{\lambda_1} \pi \Pi_{\lambda_2}]
    +\frac{1}{d^2} \sum_{k,l} \Tr[E_k E_l^\dag \Pi_{\lambda_1}] \Tr[E_l E_k^\dag \Pi_{\lambda_2}] \right).
\end{align}
In App.~\ref{sec:supmat_numerics}, we evaluate $\bE[D_{x^2}(T(\rho) \| T(\sigma_*))]$ for the single-qubit amplitude damping.

The expression takes the simpler form $D_{x^2}(\rho \| \sigma) = \Tr[\sigma^{-1} \rho^2] - 1$ for commuting $\rho$ and $\sigma$. This gives $D_{x^2}(\rho \| \sigma_*)  = d  \Tr[\rho^2] - 1$, and allows to easily evaluate the average contraction for the $\chi^2$ divergence under a unital channel $T$ and a state 2-design $\nu$,
\begin{equation}
    \eta(T, D_{x^2}, \nu\times\delta_{\sigma_*}) = \frac{d\, \bE[\Tr T(\rho)^2] - 1}{d - 1} = \frac{d^2}{d^2-1}\left[ \Tr\tau^2 - \frac{1}{d^2} \right] = \Tr\tau^2 + \order{d^{-2}},
\end{equation}
which shows that the average contraction is exponentially decreasing for noisy unital channels.

Additionally, this simpler form can be used to define an alternative quantum generalization for the $\chi^2$-divergence \cite{PetzRuskai1998},
\begin{equation}
    \chi^2_{0}(\rho \| \sigma) = \Tr[\sigma^{-1} (\rho - \sigma)^2] = \Tr[\sigma^{-1} \rho^2] - 1.
\end{equation}
It is well known to be consistent with the classical $\chi^2$-divergence and to satisfy DPI \cite{PetzRuskai1998}, thus it satisfies the Pinsker-type inequality of Prop.~\ref{prop:pinsker} and we can conclude upper bounds for the average contraction of the trace distance. For a state 2-design $\nu$ and fixed $\sigma > 0$, we have that
\begin{equation}
    \bEx{\rho\sim\nu}\qty[\Tr[T(\sigma)^{-1} T(\rho)^2]] = \frac{d}{d+1}\qty(\Tr[T(\sigma)^{-1} \pi^2] + \Tr[T(\sigma)^{-1} \Tr_2\tau^2] ).
\end{equation}
Together with the Pinsker-type inequality in Prop.~\ref{prop:pure_to_mix_ratio_bound}, this directly implies the following result.

\begin{proposition}
    Let $T:\cS(\cH_d)\to\cS(\cH_{d'})$ be a quantum channel, $\tau$ be its Choi state and $\pi=T(\frac{\mathbb{I}}{d})$. Let $\nu$ be a 2-design over pure states in $\cS(\cH_d)$.
    Then, the average contraction of the trace distance is upper bounded by
    \begin{equation}
    \eta_1(T,\norm{\cdot}_1,\nu\times\delta_{\sigma_*})
        \leq
        \sqrt{
        \Tr[\pi^{-1} \Tr_2\tau^2]} +\order{d^{-1}}.
    \end{equation}
\end{proposition}

\begin{proof}
    We introduce $f(x)=x^2-1$ in the upper bound of Prop.~\ref{prop:pure_to_mix_ratio_bound},
    \begin{equation}
        \eta_p(T, \norm{\cdot}_1, \nu\times \delta_{\sigma_*})^2
        \leq
        \frac{d}{(1-1/d)^2} \eta_p(T, \chi^2_{0}, \nu \times \delta_{\sigma_*}),
    \end{equation}
    and evaluate the left-hand side at $p=1$
    \begin{align}
        \eta_1(T, \chi^2_{0}, \nu \times \delta_{\sigma_*}) 
        &= \frac{1}{d-1} \qty(\frac{d}{d+1}\qty(\Tr[\pi^{-1} \pi^2] + \Tr[\pi^{-1} \Tr_2\tau^2] ) -1) \\
        &=\frac{d}{d^2-1} \qty(1 + \Tr[\pi^{-1} \Tr_2\tau^2] - \frac{d+1}{d}) \\
        &=\frac{d}{d^2-1} \qty(\Tr[\pi^{-1} \Tr_2\tau^2] - \frac{1}{d})
    \end{align}
    Combining the terms and discarding $ \order{d^{-1}}$ terms yields the desired result.
\end{proof}

This is comparable to Cor.~\ref{cor:avc_2design_mix}. Although it takes a simpler form, we can show that it is strictly worse, at least up to $\order{d^{-1}}$. The key objects are $\Tr\sqrt{\Tr_2\tau^2}$ and $\sqrt{\Tr[\pi^{-1} \Tr_2\tau^2]}$, we use that
\begin{align}
    \Tr\sqrt{\Tr_2\tau^2} &= \Tr[\pi^{1/2} \pi^{-1/2} \sqrt{\Tr_2\tau^2}] \\
    &\leq \sqrt{\Tr[(\pi^{-1/2} \sqrt{\Tr_2\tau^2})^2]}
    = \sqrt{\Tr[(\pi^{-1/4} \sqrt{\Tr_2\tau^2} \pi^{-1/4} )^2]} \\
    &\leq \sqrt{\Tr[\pi^{-1} \Tr_2\tau^2]},
\end{align}
where the last step follows from $\Tr[(B^{1/2} A B^{1/2})^2] \leq \Tr[B A^2 B]$ for $A, B \geq 0$ (Exercise IX.2.11 in \cite{Bhatia1996}).

\section{Supplementary material for Sec.~\ref{sec:numerics}}\label{sec:supmat_numerics}

\paragraph{Depolarizing.} Our results for the asymptotic behavior for the average trace distance are already presented in Prop.~\ref{prop:avc_local_depol}, where we show that the average contraction for the trace distance vanishes for $p \gtrsim 0.42$, and converges to one for $p \lesssim 0.25$. Additionally, we can evaluate the average contraction for the $\chi^2$-divergence:
\begin{equation}
    \eta(T_{\text{depol}}^{\otimes n}, \chi^2, \mu \times \delta_{\sigma_*}) = \left(\frac{1+3(1-p)^2}{4}  \right)^n + \order{2^{-n}}.
\end{equation}
To derive the bound for the max relative entropy, observe that for any pure $n$-qubit state $\rho$, the denominator satisfies $D_{\max}(\rho \| \mathbb{I}/2^n) = n \ln 2$ which is independent of the input state, so it suffices to analyze the max relative entropy between the output states. Let $\lambda = D_{\max}(T_{\text{depol}}^{\otimes n}(\rho) \| \mathbb{I}/2^n)$, and note that, by definition,
\begin{equation}
	T_{\text{depol}}^{\otimes n}(\rho) = \sum_{A\subset [n]} p^{|A|} (1 - p)^{n-|A|} \Tr_A(\rho) \otimes \frac{\mathbb{I}_A}{2^{|A|}} \leq \lambda \frac{\mathbb{I}}{2^n},
\end{equation}
where $\Tr_A(\rho)$ denotes the partial trace over the qubits in $A$, and $\mathbb{I}_A$ is the identity on subsystem $A$. In particular, if the input state is pure $\rho =\ketbra{\psi}$ we must have that
\begin{align}
	\lambda &\geq 2^n
	\sum_{A\subset [n]} p^{|A|} (1 - p)^{n-|A|} \expval{\Tr_A(\rho) \otimes \frac{\mathbb{I}_A}{2^{|A|}}}{\psi} \\
	&\geq
	2^n \sum_{A\subset [n]} p^{|A|} (1 - p)^{n-|A|} \frac{\Tr(\Tr_A\rho)^2}{2^{|A|}} \\
	&\geq
	2^n \sum_{A\subset [n]} p^{|A|} (1 - p)^{n-|A|} \frac{1}{2^{2|A|}} \\
	&=
	\sum_{k=0}^n \qty(\frac{p}{2})^{k} (2 - 2p)^{n-k} 
	= (2  - 2p + \frac{p}{2})^n = (2 - \frac{3}{2} p)^n,
\end{align}
where we used $\Tr[\rho \Tr_A(\rho) \otimes \mathbb{I}] = \Tr(\Tr_A\rho)^2 \geq 2^{-|A|}$, with the inequality generally valid only for pure $\rho$. Therefore, we arrive at
\begin{equation}
	D_{\max}(T_{\text{depol}}^{\otimes n}(\rho) \| \mathbb{I}/2^n) \geq n \cdot \ln(2 - \frac{3}{2} p),
\end{equation}
which remains informative as long as $p < 2/3$.

\paragraph{Dephasing.} We express the single-qubit dephasing channel with dephasing strength $0\leq\gamma \leq 1$ as a phase flip:
\begin{equation}
	T_{\text{deph}}(\rho) = (1 - \frac{\gamma}{2}) \rho + \frac{\gamma}{2} Z \rho Z,
	\quad Z = \mqty(1 & 0 \\ 0 & -1),
\end{equation}
whose Choi state is
\begin{equation}
	\tau_{\text{deph}} = \frac{1}{2}
	\begin{pmatrix}
	1 & 0 & 0 & 1-\gamma \\
	0 & 0 & 0 & 0 \\
	0 & 0 & 0 & 0 \\
	1-\gamma & 0 & 0 & 1
	\end{pmatrix}.
\end{equation}
We compute
\begin{equation}
	\Tr\sqrt{\Tr_2 \tau_{\text{deph}}^2} = \sqrt{1+(1-\gamma)^2},
	\quad
	S(\tau_{\text{deph}}) = H_2\qty(\frac{\gamma}{2}), 
    \quad \text{and} \quad
    \Tr \tau^2 = \frac{1+(1-\gamma)^2}{2}.
\end{equation}
where $H_2(p) = - [p \log p + (1-p) \log (1-p)]$ is the binary entropy function. Clearly, $\Tr\sqrt{\Tr_2 \tau_{\text{deph}}^2}\geq 1$ for all $\gamma$, while $H_2\qty(\frac{\gamma}{2}) < 1$ for any $\gamma < 1$. Therefore, the average contraction converges to one as $n\to\infty$ for all $\gamma < 1$, except for the extremal case $\gamma=1$.

Before considering that case, we generalize to dephasing channels that dampen off-diagonal terms in some preferred basis:
\begin{equation}
	T_{\Gamma}(\rho) = \Gamma \star \rho, 
	\text{ with } \Gamma \geq 0, \; \Gamma_{ii} = 1,
\end{equation}
where $(\Gamma \star \rho)_{ij} = \Gamma_{ij} \rho_{ij}$ is the entrywise multiplication. Letting $\Gamma=\sum_k \gamma_k \ketbra{\phi_k}$ be the spectral decomposition, and denoting the preferred basis by $\{\ket{i}\}$, the Kraus operators can be written as $A_k = \sqrt{\frac{\gamma_k}{d}} \sum_i \braket{i}{\phi_k} \ketbra{i}$. The entropy of the corresponding Choi state is
\begin{equation}
	S(\tau_\Gamma) = - \sum_k \frac{\gamma_k}{d} \log \frac{\gamma_k}{d} = S\qty(\frac{\Gamma}{d}) \leq \log d,
\end{equation}
with equality only if $\Gamma = \mathbb{I}$. This shows that unless the off-diagonal terms are fully discarded, the average contraction of $n$-fold products of the channel $T_{\Gamma}$ converges to one as $n\to\infty$.

The case $\Gamma = \mathbb{I}$, corresponding to $\gamma=1$ in the single-qubit dephasing channel, reduces to analyzing the trace distance between the maximally mixed state and the diagonal of a Haar-random pure state. This quantity is known (see Lemma 4 in \cite{Singh2016})
\begin{equation}
	\frac{1}{1-\frac{1}{d}}\bEx{\rho\sim\mu_d}\qty[\frac{1}{2}\norm{\mathbb{I}\star \rho - \frac{\mathbb{I}}{d}}_1]
	= \frac{1}{1-\frac{1}{d}} \qty(1-\frac{1}{d})^d \to  e^{-1} \approx 0.3679
	\text{ as } d \to \infty.
\end{equation}
This completes the characterization of the average contraction for any $n$-fold product of the channel $T_{\Gamma}$ in the $n\to\infty$ limit. For completeness, we also evaluate the upper bound of Cor.~\ref{cor:avc_2design_mix} for $\Gamma=\mathbb{I}$:
\begin{align}
	&\frac{1}{2(1-1/d)}\Tr\qty(\frac{1}{d+1}\qty(d \frac{\mathbb{I}}{d^2} - \frac{\mathbb{I}}{d^2}))^{1/2} 
	=
	\frac{d}{2(1-1/d)}\qty(\frac{1}{d+1}\qty(\frac{1}{d} - \frac{1}{d^2}))^{1/2} \\
	&=
	\frac{1}{2(1-1/d)} \qty(\frac{1}{d+1}\qty(d - 1))^{1/2}
	=
	\frac{1}{2\sqrt{1-d^{-2}}} \\
	&= \frac{1}{2} + \order{d^{-2}},
\end{align}
which is strictly larger than the exact value above.

\paragraph{Amplitude damping.}
The single-qubit amplitude damping channel, with damping strength $0\leq \lambda \leq 1$, is defined as
\begin{equation}
	T_{\text{damp}}(\rho) = K \rho K + \lambda \ketbra{0}{1} \rho \ketbra{1}{0},
	\quad K = \mqty(1 & 0 \\ 0 & \sqrt{1-\lambda}),
\end{equation}
and the Choi state is given by
\begin{equation}
	\tau_{\text{damp}} = \frac{1}{2}
	\begin{pmatrix}
	1 & 0 & 0 & \sqrt{1 - \lambda} \\
	0 & \lambda & 0 & 0 \\
	0 & 0 & 0 & 0 \\
	\sqrt{1-\lambda} & 0 & 0 & 1 -\lambda
	\end{pmatrix}		
	,
\end{equation}
with associated quantities
\begin{equation}
	\Tr\sqrt{\Tr_2 \tau_{\text{damp}}^2} =  \frac{1}{2} \left[ \sqrt{2-\lambda+\lambda^2} + \sqrt{(1-\lambda)(2-\lambda)} \right],
	\quad
	S(\tau_{\text{damp}}) = H_2\qty(\frac{\lambda}{2}),
\end{equation}
where $H_2(p)$ is the binary entropy function. The first quantity falls below $1$ for $\lambda > 2/3 \approx 0.67$, which implies asymptotically vanishing upper bounds on average contraction for those values. Since the channel is not unital, we compute $\norm{\pi}_{\infty} = (1+\lambda)/2$. The condition $S(\tau_{\text{damp}}) < \log \norm{\pi}_{\infty}^{-1} = -\log[ (1+\lambda) /2]$ holds for $\lambda \lesssim 0.20$, indicating that the lower bound in Thm.~\ref{thm:avc_lower_bound} predicts asymptotic convergence to $1$ in that regime.

In this case, we can evaluate the exact expected value for the average contraction of the $\chi^2$-divergence for $n=1$. We have $\lambda_{x} = (1 +(-1)^x \lambda) /2$ and $\Pi_x = \ketbra{x}$ with $x\in\{0,1\}$; $E_0=\mqty(\mathbb{I} & 0 \\ 0 & \sqrt{1-\lambda})$ and $E_1=\sqrt{\lambda} \ketbra{0}{1}$. Then,
\begin{align}
	\bEx{\rho\sim\nu}[\Tr[T(\rho)\Pi_x T(\rho) \Pi_y]]
	&=
	\frac{2}{3} \qty(
		\Tr[\pi \Pi_x \pi \Pi_y]
		+\frac{1}{4} \sum_{z,t\in\{0,1\}} \Tr[E_z E_t^\dag \Pi_x] \Tr[E_t E_z^\dag \Pi_y]
	) \\
	&=
	\frac{2}{3} \qty(
		\delta_{xy} \frac{(1 +(-1)^x \lambda)^2}{4}
		+\frac{1}{4} \sum_{z\in\{0,1\}} \Tr[E_z E_z^\dag \Pi_x] \Tr[E_z E_z^\dag \Pi_y]
	) \\
	&=
	\frac{2}{3} \left(
		\delta_{xy} \frac{(1 +(-1)^x \lambda)^2}{4}
		+\frac{1}{4} \left(
			\Tr[E_0 E_0^\dag \Pi_x] \Tr[E_0 E_0^\dag \Pi_y] \right. \right.\\
	&\left.\left. \qquad + \delta_{x0}\delta_{y0} \Tr[E_1 E_1^\dag \Pi_0] \Tr[E_1 E_1^\dag \Pi_0]
		\right)
	\right) \\
	&=
	\frac{2}{3} \qty(
		\delta_{xy} \frac{(1 +(-1)^x \lambda)^2}{4}
		+\frac{1}{4} \qty(
			\qty(1-\lambda)^{x+y}
			+ \delta_{x0}\delta_{y0} \lambda^2
		)
	) \\
	&=
	\frac{1}{6} \qty(
		\delta_{xy} (1 + 2 (-1)^x \lambda + (1+\delta_{x0}) \lambda^2)
		+ \qty(1-\lambda)^{x+y}
	),
\end{align}
which gives $(1+\lambda+\lambda^2)/3$ for $x=y=0$, $(1-\lambda)^2/3$ for $x=y=1$, and $(1-\lambda)/6$ for $x\neq y$. Therefore, introducing that $D_{x^2}(\rho\| \sigma_*)=1$,
\begin{align}
	\bEx{\rho\sim\nu}\qty[\frac{D_{x^2}(T(\rho)\| \pi)}{D_{x^2}(\rho\| \sigma_*)}]
	&= \sum_{x\in\{0,1\}} \frac{2}{1+(-1)^x\lambda} \bEx{\rho\sim\nu}[\Tr[(T(\rho)\Pi_x)^2]]
	+ 2 \frac{\ln\frac{1+\lambda}{1-\lambda}}{\lambda} \bEx{\rho\sim\nu}[\Tr[T(\rho)\Pi_0 T(\rho) \Pi_1]] - 1
	\\
	&=
	\frac{2}{3}\frac{1+\lambda+\lambda^2}{1+\lambda} + \frac{2}{3} (1-\lambda) + \frac{1}{3}\frac{1-\lambda}{\lambda} \ln \frac{1+\lambda}{1-\lambda} - 1 \\
	&=\frac{1}{3}\frac{1-\lambda}{1+\lambda} + \frac{1}{3}\frac{1-\lambda}{\lambda} \ln \frac{1+\lambda}{1-\lambda} \\
	&=\frac{1-\lambda}{3}\qty(\frac{1}{1+\lambda} + \frac{1}{\lambda} \ln \frac{1+\lambda}{1-\lambda}),
\end{align}
which very closely matches the numerical results.

\paragraph{Mixture of (orthogonal) unitaries.} Given a set of unitaries $\{U_i\}_{i=1}^M$ and a probability distribution $\mathbf{p}$, we consider the  quantum channel defined by applying unitary $U_i$ with probability $p_i$:
\begin{equation}
	T_{\text{MU}}(\rho) = \sum_{i=1}^M p_i U_i \rho U_i^{\dag},
	\quad \text{with Choi state} \quad
	\tau_{\text{MU}} = \sum_{i=1}^M p_i U_i \otimes \mathbb{I} \ketbra{\Omega} U_i^{\dag} \otimes \mathbb{I}.
\end{equation}
We say that two unitaries are orthogonal if they are orthogonal in the Hilbert-Schmidt sense, i.e. $\Tr[U_i^\dag U_j] = \delta_{ij} d$. Under this condition, the vectors $U_j\otimes\mathbb{I}\ket{\Omega}$ form an orthonormal set, since $\expval{U_i^\dag U_j\otimes\mathbb{I}}{\Omega}=\frac{1}{d}\Tr[U_i^\dag U_j] = \delta_{ij}$. Thus, the Choi state $\tau_{\text{MU}}$ is diagonal in this basis and its entropy is simply the Shannon entropy of the probability distribution $S(\tau_{\text{MU}})=H(\mathbf{p})=-\sum_{i=1}^M p_i\log p_i$. We can also compute
\begin{equation}
	\Tr\sqrt{\Tr_2 \tau_{\text{MU}}^2} =
	\Tr\sqrt{\Tr_2 
		\sum_{i=1}^M p_i^2 U_i \otimes \mathbb{I} \ketbra{\Omega} U_i^{\dag} \otimes \mathbb{I}
	}
	=
	\Tr\sqrt{
		\sum_{i=1}^M p_i^2 U_i \frac{\mathbb{I}}{d} U_i^{\dag}
	} 
	= \sqrt{d \sum_{i=1}^M p_i^2} = \sqrt{d} \norm{\mathbf{p}}_2.
\end{equation}
From this, we conclude that if $H(\mathbf{p}) < \log d$ the average contraction tends to one, and if $\norm{\mathbf{p}}_2 < d^{-1/2}$ it is exponentially vanishing.

\end{document}